\documentclass[lettersize,journal]{IEEEtran}
\usepackage{amsmath,amsfonts}
\usepackage{algorithmic}
\usepackage{algorithm}
\usepackage{array}
\usepackage[caption=false,font=normalsize,labelfont=sf,textfont=sf]{subfig}
\usepackage{textcomp}
\usepackage{stfloats}
\usepackage{url}
\usepackage{verbatim}
\usepackage{graphicx}
\usepackage{amssymb}
\usepackage{cite}
\usepackage{acronym}
\usepackage{mathrsfs}
\usepackage{amsthm}
\usepackage{multirow}
\usepackage{xcolor}
\usepackage{tikz}
\usepackage{stmaryrd}
\usepackage{pgfplots}
\usepackage{booktabs}
\usepackage[table]{xcolor}
\pgfplotsset{compat=1.17}
\usetikzlibrary{calc}
\usepackage{cuted}
\usepackage{lipsum}
\usepackage{mathtools}
\usepgfplotslibrary{groupplots}
\usetikzlibrary{patterns}

\allowdisplaybreaks

\DeclareMathOperator*{\argmin}{arg\,min}
\DeclareMathOperator*{\Var}{\text{Var}}

\DeclareMathOperator*{\sign}{\text{sign}}

\newcommand{\ceil}[1]{\lceil {#1} \rceil}	
\newcommand{\indep}{\perp \!\!\! \perp}
\newcommand{\brackets}[1]{\left(#1\right)}
\newcommand{\sbrackets}[1]{\left[#1\right]}

\newcommand{\abs}[1]{\left\lvert#1\right\rvert}

\newtheorem{theorem}{Theorem}[section]
\newtheorem{corollary}[theorem]{Corollary}
\newtheorem{lemma}[theorem]{Lemma}
\newtheorem{proposition}[theorem]{Proposition}

\newtheorem{definition}[theorem]{Definition}
\newtheorem{remark}[theorem]{Remark}

\newcommand{\mz}[1]{\textcolor{black}{#1}}

\acrodef{dmc}[DMC]{discrete memoryless channel}
\acrodef{mmse}[MMSE]{{minimum mean-squared error}}
\acrodef{dm}[DM]{{decision maker}}
\acrodef{lqg}[LQG]{Linear-Quadratic-Gaussian}
\acrodef{pdf}[p.d.f.]{probability density function}
\acrodef{dpc}[DPC]{Dirty Paper Coding}

\begin{document}


\title{\mz{Joint Communication and Control: Causal Coordination Coding for Witsenhausen Counterexample}}

\author{Mengyuan Zhao, 
Tobias J. Oechtering, 
Maël Le Treust, 
\thanks{This work is supported by Swedish Research Council (VR) under grant 2020-03884. The work of M. Le Treust is supported by the French National Agency for Research (ANR) via the project n°ANR-22-PEFT-0010 of the France 2030 program PEPR Futur Networks.}
\thanks{The results have been presented in part at the 2024 IEEE International Symposium on Information Theory (ISIT) \cite{zhao2024coordination}, at the 2024 IEEE Conference on Decision and Control (CDC) \cite{zhao2024CDC}, at the 2024 IEEE Information Theory Workshop (ITW) \cite{zhao2024causal}, and at the 2025 IEEE ISIT \cite{zhao2025zeroestimationcoststrategy}.}
\thanks{M. Zhao and T. J. Oechtering are with the Division of Information Science and Engineering,
        KTH Royal Institute of Technology, 100 44 Stockholm, Sweden
        {\tt\small \{mzhao, oech\}@kth.se }}
\thanks{M. Le Treust is with CNRS, University of Rennes, Inria, IRISA UMR 6074, F-35000 Rennes, France
        {\tt\small mael.le-treust@cnrs.fr}}
}



\maketitle

\begin{abstract}
In decentralized network control, communication plays a \mz{crucial} role by transforming local observations into shared knowledge, enabling agents to coordinate their actions. This paper investigates how communication facilitates cooperation behavior and therefore improves the overall performance in \mz{Witsenhausen counterexample}, a canonical toy problem in distributed decision-making. \mz{We consider the vector-valued setups} where the encoder, i.e., the first decision-maker (DM) acts causally and the decoder, i.e., the second DM, operates noncausally, 1) without and 2) with access to channel feedback. Using a coordination coding framework, we characterize the achievable power-estimation cost regions in single-letter expressions for both scenarios. 
\mz{We show that feedback enlarges the coordination region compared with the no-feedback setup.}  \mz{Building on the characterization of the power-estimation cost region}, we propose a hybrid \mz{strategy} that combines discrete quantization with a continuous Gaussian codebook, the Zero Estimation Cost (ZEC) \mz{strategy}. This \mz{strategy} uses coding tools that allow perfect reconstruction of the target symbols, leading to an asymptotic estimation cost equal to zero, while substantially reducing the power cost. \mz{Moreover}, when channel feedback is available, we propose an analogous \mz{ZEC strategy} that simultaneously achieves zero power and zero estimation \mz{costs} in the low-noise regime. \mz{Finally, we establish that, when restricted to jointly Gaussian random variables, the performance is identical across all setups featuring at least one causal DM, regardless of the presence of feedback information.}

\end{abstract}



\section{Introduction}

Decentralized control studies systems where multiple controllers or decision-makers (DMs) act based on different, locally available information. Among its most well-known challenges is the Witsenhausen counterexample \mz{\cite{witsenhausen1968}}, a deceptively canonical two-stage linear-quadratic-Gaussian (LQG) problem introduced in 1968, \mz{as illustrated in Figure \ref{fig: original model}. 
The objective is to design optimal controllers
\begin{align}
    f: \mathcal{X}_0\longrightarrow\mathcal{U}_1,\quad g: \mathcal{Y}_1\longrightarrow\mathcal{U}_2,\label{eq: single-shot controller}
\end{align}
that minimize the weighted sum of power and estimation costs:
\begin{align}
 \min_{f,g}\brackets{\lambda^2\cdot P+S},\label{eq: scalar obj function}
\end{align}
where
\vspace{-0.2cm}
\begin{align}
&\text{power cost: }&\quad &P= \mathbb E[U_1^2],\label{eq: scalar power cost}\\
&\text{estimation cost: }&\quad&S = \mathbb E[(X_1-U_2)^2],\label{eq: scalar est cost}
\end{align}\vspace{-0.2cm}
with\vspace{-0.2cm}
\begin{align}
 U_1 &= f(X_0),\nonumber\\
 X_1 &= U_1+X_0\label{eq: X_1 generation},\\
 Y_1 &= X_1+Z_1,\label{eq: Y_1 generation}\\
 U_2 &= g(Y_1),\nonumber
\end{align}
and $X_0\sim\mathcal{N}(0,Q),
Z_1\sim\mathcal{N}(0,N)$, for some $Q,N>0$, and $\lambda^2$ is the trade-off constant between the two costs.
It was shown in \cite{witsenhausen1968} that for a given $f$, the optimal $g$ is given by the minimum mean-square error (MMSE) estimator of $X_1$ given $Y_1$, see \cite{kay1993fundamentals}. As a result, the estimation cost $S$ becomes
\begin{align*}
    S = \mathbb E[(X_1-\mathbb E[X_1|Y_1])^2].
\end{align*}
Thus, this problem \eqref{eq: scalar obj function} effectively reduces to optimizing over only the first DM $f$.}

\begin{figure}[t]
  \centering



\begin{tikzpicture}[scale=0.83, black]

    \draw (2,0) rectangle (3,1);
    \draw (6.8,0) rectangle (7.8,1);

    \draw (4.2,0.5) circle (0.2) node {$+$};
    \draw (5.6,0.5) circle (0.2) node {$+$};
    \draw (8.8,-0.5) circle (0.2) node {$-$};

    \filldraw (1,1.5) circle (2pt) node[above] {$X_{0}\sim \mathcal{N}(0,Q)$};
    \filldraw (5.6,1.5) circle (2pt) node[above] {$Z_{1}\sim \mathcal{N}(0,N)$};

    \draw[->] (1,1.5) -- (1,0.5) -- (2,0.5);
    \draw[->] (1,1.5) -- (4.2,1.5) -- (4.2,0.7);
    \draw[->] (3,0.5) -- (4,0.5);
    \draw[->] (4.4,0.5) -- (5.4,0.5);
    \draw[->] (4.9,0.5) -- (4.9,-0.5) -- (8.6,-0.5);
    \draw[->] (5.6,1.5) -- (5.6,0.7);
    \draw[->] (5.8,0.5) -- (6.8,0.5);
    \draw[->] (7.8,0.5) -- (8.8,0.5) -- (8.8, -0.3);
    \draw[->] (9,-0.5) -- (10,-0.5);

    \node at (1.5,0.8) {$X_0$};
    \node at (3.5,0.8) {$U_1$};
    \node at (4.9,0.8) {$X_{1}$};
    \node at (6.3,0.8) {$Y_{1}$};
    \node at (8.3,0.8) {$U_2$};
    \node at (8.3,-0.2) {$X_{1}$};
    \node at (9.6,-0.2) {$X_2$};
    \node at (2.5,0.5) {$f$};
    \node at (7.3,0.5) {$g$};
\end{tikzpicture}


\caption{\mz{Original scalar-valued Witsenhausen counterexample. $X_0\sim\mathcal{N}(0,Q)$ is the source state. The first DM $f: X_0\mapsto U_1$. Then, $X_1=X_0+U_1$, and $Y_1 = X_1+Z_1$, where $Z_1\sim\mathcal{N}(0,N)$ is the Gaussian noise. At last, the second DM $g: Y_1\mapsto U_2$ estimates the interim state $X_1$.} }
\label{fig: original model}
\end{figure}

Despite its simplest formulation of scalar dynamics and quadratic cost, Witsenhausen showed that \mz{linear strategies are strictly suboptimal. This is in contrast to the fully observed LQG control problems, where the controllers have access to the complete system state; in such centralized settings, optimal control laws are linear \cite{witsenhausen1971separation}.} In fact, even a nonlinear binary strategy, \mz{i.e., Witsenhausen's two-point strategy of the form
\begin{align}
    U_1 = a\cdot \mathsf{sign}(X_0) - X_0, \quad
    X_1 = a\cdot \mathsf{sign}(X_0),\label{eq: 2-point}
\end{align}
for some $a\geq 0$, can outperform the optimal linear strategy.} More than five decades later, this problem remains unsolved in closed form and continues to serve as a benchmark for understanding the interplay between control and communication in decentralized systems. It has inspired a broad body of work on stochastic teams and optimality of \mz{linear} policies \cite{bansal1986stochastic, gupta2015existence}, control design under SNR and rate constraints \cite{silva2010control, martins2005fundamental, freudenberg2008feedback, derpich2012improved, Stavrou2022sequential}, and information-theoretic perspectives on coordination, communication, and security \cite{yuksel2013stochastic, agrawal2015implicit, Akyol2017information, charalambous2017hierarchical, wiese2018secure}.

The key challenge of the Witsenhausen counterexample is the \textit{nonclassical information structure} \cite{ho1980another, yuksel2009stochastic, ho2008review}: The second DM receives a noisy observation of the first DM’s action but does not know its input, which \mz{is distinct from} the fully observed LQG setting. This information structure turns the Witsenhausen optimization problem into a non-convex problem, making the search for an optimal strategy difficult \cite{bansal1986stochastic, baglietto1997nonlinear}.
A crucial factor is the \textit{dual role of control} at the first DM: it must spend some power to simultaneously regulate the system state $X_1$ while embedding and transmitting relevant information that aids the second DM’s estimation. Seen in this way, \mz{linear} strategies appear ineffective with respect to both objectives: they consume more power for control yet keep the complexity of the message, whereas Witsenhausen's two-point strategy \eqref{eq: 2-point} communicates more clearly by forcing more distinguishable signals \cite{mitter2007information}.

Inspired by this \mz{observation}, over the past five decades, researchers have developed different approaches based on quantization and lattice-based control policies. To name a few, \cite{tseng2017alocal, baglietto2001numerical} adopted the numerical optimization approach, where the latter one uncovered a sloped step-like structure in near-optimal policies. \mz{Generalizing \eqref{eq: 2-point}}, \cite{lee2001hierarchical} considered \mz{a $k$-point} quantization strategy. \mz{Quantized} team
policies have \mz{been} shown to be asymptotically optimal in \cite{saldi2016finite}. Using a game-theoretic approach, the slopey quantization approach has been shown to be locally optimal in \cite{ajorlou2016slopey}. From the optimal transport point of view, \cite{wu2011transport} showed that the optimal decision strategy is a strictly increasing
unbounded piecewise real analytic function with a real analytic left inverse. More control strategies have also been proposed in \cite{li2009learning, telsang2025generaldecentralizedstochasticoptimal, mceneaney2015optimization}.


Despite its control-theoretic aspect, Witsenhausen counterexample effectively shows a strong communication requirement: one DM needs to embed and share information to another through its action.
As a result, this problem has increasingly been interpreted through an information-theoretic lens, as a one-shot joint source–channel coding problem. For example, \cite{karlsson2011iterative} developed an algorithm using an iterative joint source-channel coding approach that achieves better performances. However, compared to Shannon's traditional point-to-point communication setup \cite{shannon1948mathematical}, the goal here is to communicate the controlled \textit{interim} state \mz{$X_1$} that results from the encoder's control action, rather than the intact original state \mz{$X_0$} itself. As a result, the problem becomes more dynamic and target-adjustable ('make the plant dance'\cite{sahai2006necessity}), which requires additional communication. This perspective motivates the study of vector-valued extensions of the Witsenhausen problem \cite{Grover2010Witsenhausen}, in which the scalar states \mz{in Figure \ref{fig: original model}} are replaced by vectors. \mz{In such a case, multiple transmissions are allowed, and the DMs are able to synthesize historical information to make better decisions, thereby providing a fundamental performance bound.} This formulation facilitates the adoption of many block-coding methods \cite{cover1999elements}, such as Costa's dirty-paper coding (DPC) \cite{costa1983writing}, state amplification \cite{kim2008state}, and state masking \cite{merhav2007information}, etc. For example, \mz{researchers have} derived lower and upper bounds on the averaged Witsenhausen pair of costs and established the constant-factor optimality of quantization-based strategies in both the finite-dimensional \cite{grover2009finite}, and infinite-dimensional \cite{Grover2010Witsenhausen} cases, when both DMs are noncausal. The upper bound for the latter case was \mz{derived} using a combination of linear and DPC schemes. Using this coding structure, \cite{choudhuri2012witsenhausen} further characterized the optimal power-distortion trade-off. \mz{More} results for the generalization \mz{of} the Witsenhausen counterexample are \mz{shown} in \cite{grover2009generalized,grover2015information,grover2010witsenhausenCDC,grover2010distributed}, building on the vector-valued extension.

A central question we want to ask is: under such a nonclassical information structure, how much communication is required, or more precisely, what information is relevant to communicate, so that the two potential conflicting objectives in the Witsenhausen counterexample no longer work against each other, but \mz{rather} cooperate to \mz{improve the joint} performance.
\cite{grover2010implicit, grover2011source} investigated the communication requirement by introducing an external communication channel connecting the two noncausal DMs. By transmitting a truncated binary expansion of the source state, thereby enforcing the target state to a quantization point, they showed that the resulting scheme can simplify the source, reduce the costs, and achieve an $\varepsilon$-optimal asymptotic performances. Related results have \mz{also} been obtained for scalar-valued counterparts in \cite{martins2006witsenhausen, shoarinejad2002stochastic}. 

In this paper, we adopt the empirical coordination coding framework, formulated in \cite{cuff2011coordination,raginsky2012empirical, cuff2010coordination}, where the ``coordination capacity'' determines the minimum communication rate needed to induce the \mz{desired} cooperative behavior of agents. Based on local observations, the encoder and decoder choose their actions, not only to ensure reliable message transmission, but also to achieve a broader common objective. More precisely, they implement a coding scheme that ensures that the empirical distribution of symbol sequences converges to a target probability distribution. The characterization of the set of achievable joint probability distributions relies on a single-letter information constraint that depends on the local information available in the system. In \cite{agrawal2015implicit}, the coordination of two agents with a common average payoff function is considered, the Witsenhausen cost function being a special case. In \cite{Treust2017joint}, a point-to-point scenario is studied where the encoder is noncausal, and the decoder is strictly causal based on its past observation. This work also provides a literature review on coordination results. In \cite{le2015empirical}, necessary and sufficient conditions are established for several causal coding cases with two-sided state information. More optimal conditions for causal and noncausal encoding and decoding configurations are accessible in \cite{le2015coding}. \mz{Recently, the empirical coordination framework has been extended to joint source-channel coding with Markov channel setting \cite{zhao2026empirical}.}

Prior to this work, \mz{\cite{Treust2024power} }extended the finite-alphabet coordination problem to the Gaussian setting using a non-standard weak typicality approach, and investigated the noncausal-encoding and causal-decoding Witsenhausen setup. They provided an optimal single-letter characterization for the achievable power and estimation costs, that is, the image of the asymptotically averaged power and estimation costs induced by all possible pairs of noncausal encoding and causal decoding functions. This single-letter characterization involves auxiliary random variables (RVs) that highlight the dual role of control. Furthermore, they proposed a continuous-discrete hybrid design of the auxiliary RVs, inspired by DPC, which was shown to outperform both Witsenhausen's two-point strategy \mz{\eqref{eq: 2-point}} and the \mz{optimal linear} policy.


However, in real-time communication systems, it is generally impractical for the encoder to wait until it fully aggregates the entire input sequence before acting. In this work, we consider a more realistic and operationally practical scenario: causal-encoding and noncausal-decoding setup. In this setup, the encoder observes and responds to data instantaneously as it arrives, while the decoder has the freedom to wait until all observations are available before deciding. Some causal transmission setups have been previously studied in the classical communication regime: Besides the pioneering work of Shannon \cite{shannonstrategy1958}, \cite{choudhuri2013causal} studied causal and strictly causal state communication through a state-dependent channel. The problem of causal state estimation with a helper is analyzed in \cite{chia2013estimation}. See more related work of side information available causally and noncausally to the transmitter in \cite{keshet2008channel}. \mz{Next, we discuss our contribution of this work.}

\subsection{\mz{Our Contribution}}

Our first contribution is the optimal single-letter characterization of the achievable power-estimation cost region in the causal encoder and noncausal decoder setup. This characterization employs two auxiliary RVs that \mz{explicitly} capture the dual role of control: regulating the system state and embedding information for the second DM. The achievability proof extends the Markov block-coding approach from \cite{choudhuri2013causal} with a refined cost analysis that improves upon the noncausal-encoding framework in \cite{Treust2024power}. This refinement is essential because, in the causal setting, the interim state depends dynamically on the first DM’s causal observations, instead of the static source sequence in the classical noncausal\mz{-encoding} setup. This result fully characterizes the fundamental trade-off between \mz{two-stage} costs and provides \mz{an operational meaning} of what information is relevant and what communication scheme should be used to transmit \mz{the relevant} information. 

Motivated by this key \mz{result}, we propose a novel coding \mz{strategy} — the Zero Estimation Cost (ZEC) \mz{strategy} — which leverages the causal block-coding structure that enables perfect state reconstruction at the decoder. Inspired by Witsenhausen's two-point strategy, the ZEC \mz{strategy} drives the achievable estimation cost to zero, while saving a significant amount of power. When combined with the \mz{optimal} linear strategy through time-sharing, the ZEC \mz{strategy} strictly outperforms both the classical two-point method and the \mz{optimal linear} policies across \textit{all} possible directions in the power-estimation cost trade-off space. Then, we extend this design to a more general $k$-point ZEC \mz{strategy}, which further reduces the minimum required power cost as the quantization level $k$ increases.

Feedback, often an intrinsic component of communication and control systems, offers the potential to enhance performance by enabling DMs to refine their actions based on previous outcomes, thus reducing uncertainty about the unknown state. While in the classical point-to-point communication setting, feedback \mz{has been shown not to} increase the channel capacity \cite{elgamal2011nit}, seminal work by Schalkwijk and Kailath \cite{schalkwijk1966coding} introduced a simple coding scheme that achieves an exponential decay of error probability, and error exponents for variable-length block codes are further investigated in \cite{nakibouglu2008error}. \mz{Contrary} to the point-to-point \mz{setup}, feedback has proven helpful in \mz{improving} the channel capacity in many multi-terminal scenarios. In particular, feedback can enlarge the capacity region and assist communication in the multiple-access channel \cite{gaarder1975capacity,ozarow1984capacity} as well as the broadcast channel \cite{ozarow1984achievable, dueck1979capacity}. Within the empirical coordination framework, feedback has been shown to simplify the coordination problem. Specifically, \cite{Letreust2015empirical} demonstrated that the perfect channel feedback enables the decoder to align its outputs more directly with the source state, reducing the complexity of the information constraint and decreasing the number of auxiliary RVs required. A state leakage and coordination problem with causal encoder \mz{in} the presence of noisy channel feedback is considered in \cite{Letreust2021state}. Together, these results show that channel feedback effectively enlarges the set of achievable joint distributions and strengthens coordination.

\mz{As the second model of our} work, we characterize the region of achievable power-estimation costs for causal-encoding and noncausal-decoding with perfect channel feedback at the encoder. Consistent with the prior observations, the presence of feedback allows the encoder to anticipate potential decoding errors so as to recommend the most efficient action to the decoder. When meeting the associated information constraint, this mechanism effectively circumvents the core difficulty of the Witsenhausen problem — the nonclassical information structure. Indeed, thanks to the feedback, the action that is recommended to the second DM is also available at the first DM, ensuring a stronger alignment between their actions.

For this setup, we develop a second novel \mz{coding strategy} — the Zero Estimation Cost with feedback (ZEC-f) \mz{strategy} — analogous to the one proposed \mz{for the no-feedback setup}. We show that feedback strictly improves performance: not only does the feedback-enabled scheme dominate its no-feedback counterpart, but surprisingly, in the low-noise regime, it can also achieve the extreme point of zero power and zero estimation costs. Under our definition of achievable costs, this means that with only a finite number of effective transmissions from the encoder, the long-run average estimation cost at the decoder can be made arbitrarily small. Furthermore, \mz{in our numerical study}, we \mz{show} that this \mz{strategy} even strictly outperforms the noncausal\mz{-encoding noncausal-decoding} scheme developed in \cite{Grover2010Witsenhausen}.

\mz{Lastly, we show a unified result that, when all the random variables involved are restricted to be Gaussian, \textit{all} vector-valued Witsenhausen setups with either causal encoder or causal decoder, regardless of the help of feedback or feedforward information, share the same optimal solution: a time-sharing strategy over the optimal linear policies. This emphasizes the importance of the non-Gaussian coding structure. }

To summarize, this work demonstrates that even a moderate communication requirement can remarkably enhance coordination among multiple DMs and improve the power-estimation cost trade-off in the practically relevant causal-encoding noncausal-decoding Witsenhausen counterexample. Exploiting the advantage of communication from the fundamental coordination coding results, we establish that asymptotically zero estimation cost is achievable with significantly reduced power, and that perfect feedback enables the simultaneous attainment of zero power and zero estimation cost. More broadly, Witsenhausen counterexample provides a canonical \mz{prototype} for exploring a central question in general distributed decision-making: what is relevant to be communicated to establish \mz{coordination among multiple agents}. Our framework demonstrates how control actions shape information flow among DMs while communication enables tighter cooperation, more reliable control, and ultimately superior performance. This perspective underscores the importance of integrating coordination-coding and control to tackle fundamental challenges in decentralized and distributed systems.



This paper is structured as follows: \mz{Section \ref{sec: system model} formulates the vector-valued Witsenhausen counterexample with causal encoder and noncausal decoder and presents its optimal single-letter cost region characterization. Based on this coding result, we propose the ZEC strategy in Section \ref{sec: zec c-n} and its $k$-point generalization. Section \ref{sec: c-n w-f characterization} formulates the causal-encoding noncausal-decoding with channel feedback model and provides the optimal cost region characterization. Then, we develop an analogous ZEC-f strategy for this setup. Section \ref{subsec: joint Gaussian} unifies the control scheme for both causal-encoding noncausal-decoding with and without feedback setups when all the random variables are restricted to be Gaussian.} A numerical study is given in Section \ref{sec: numerical results}. Lastly, the conclusion follows in Section \ref{sec: conclusion}.





    


\section{System Model}\label{sec: system model}

First, we introduce the notations applied throughout this paper. Capital letters, e.g. $X_0$, denote random variables, calligraphic fonts, e.g. $\mathcal{X}_0$, stand for alphabets where random variables take values, whereas lowercase letters, e.g. $x_0$, denote realisations of random variables. Sequences of length $n\in\mathbb{N}^{\star}$ of random variables and realisations are denoted respectively by $X_0^n = \brackets{X_{0,1},\dots X_{0,t},\dots X_{0,n}}$ and $x_0^n = \brackets{x_{0,1},\dots x_{0,t},\dots x_{0,n}}$ for all $t\in\{1,...,n\}$. Furthermore, incomplete sequences up to stage $t\in\{1,...,n\}$ are similarly denoted by $X_0^t = \brackets{X_{0,1},\dots X_{0,t}}$ and $x_0^t = \brackets{x_{0,1},\dots x_{0,t}}$. Besides, $\mathbb P(\cdot), \mathcal P(\cdot)$ and $\mathcal P(\cdot|\cdot)$ denote the general probability, probability density functions (PDFs) and conditional PDFs respectively. $\mathbf{1}_{\{\cdot\}}$ is the indicator function. $\Phi(\cdot)$, $\phi(\cdot)$ \mz{are} the Gaussian cumulative density function (CDF) and PDF. \mz{The random variables $X_0, U_1, X_1, Z_1, Y_1, U_2$ stand for the source state, channel input, interim system state, channel noise, channel output, receiver's output, respectively, and all take values from the real line $\mathbb{R}$. }

\IEEEpubidadjcol


\begin{figure}[t]
  \centering


\begin{tikzpicture}[scale=0.91]
    \draw (2,0) rectangle (3,1);
    \draw (6.8,0) rectangle (7.8,1);

    \draw (4.2,0.5) circle (0.2) node {$+$};
    \draw (5.6,0.5) circle (0.2) node {$+$};

    \filldraw (1,1.5) circle (2pt) node[above] {$X_{0,t}\sim \mathcal{N}(0,Q)$};
    \filldraw (5.6,1.5) circle (2pt) node[above] {$Z_{1,t}\sim \mathcal{N}(0,N)$};

    \draw[->] (1,1.5) -- (1,0.5) -- (2,0.5);
    \draw[->] (1,1.5) -- (4.2,1.5) -- (4.2,0.7);
    \draw[->] (3,0.5) -- (4,0.5);
    \draw[->] (4.4,0.5) -- (5.4,0.5);
    \draw[->] (4.9,0.5) -- (4.9,-0.5) -- (8.8,-0.5);
    \draw[->] (5.6,1.5) -- (5.6,0.7);
    \draw[->] (5.8,0.5) -- (6.8,0.5);
    \draw[->] (7.8,0.5) -- (8.8,0.5);

    \node at (1.5,0.8) {$X_0^t$};
    \node at (3.5,0.8) {$U_{1,t}$};
    \node at (4.9,0.8) {$X_{1,t}$};
    \node at (6.3,0.8) {$Y_{1}^n$};
    \node at (8.3,0.8) {$U_{2}^n$};
    \node at (8.3,-0.2) {$X_{1}^n$};
    \node at (2.5,0.5) {$f^{(t)}$};
    \node at (7.3,0.5) {$g$};
\end{tikzpicture}

\caption{Vector-valued Witsenhausen counterexample with causal encoder and noncausal decoder. The i.i.d. state and the channel noise are drawn according to Gaussian distributions $X_0^n\sim\mathcal{N}(0, Q\cdot\mathbb I_{n\times n})$ and $Z_1^n\sim\mathcal{N}(0,N\cdot\mathbb I_{n\times n})$. At each time instant $t\in\{1,\ldots,n\}$, the causal encoder takes the past observations $X_0^t$ and generates $U_{1,t}$. The noncausal decoder receives the whole sequence $Y_1^n$ and outputs $U_2^n$, the MMSE estimator of $X_1^n$.}
\label{fig:c-n wits}
\end{figure}

\mz{In contrast to the original scalar-valued model in Figure \ref{fig: original model}}, we now consider the vector-valued Witsenhausen counterexample with causal encoder and noncausal decoder depicted in Figure \ref{fig:c-n wits}. The source states and channel noises are drawn independently according to the i.i.d. Gaussian distributions $X_0^n\sim\mathcal{N}(0, Q\cdot\mathbb I_{\mz{n\times n}})$ and $Z_1^n\sim\mathcal{N}(0,N\cdot\mathbb I_{\mz{n\times n}})$, for some $Q>0$, $N>0$. We denote by $X_1^n$ the memoryless interim state and $Y_1^n$ the output of the memoryless additive channel, where each \mz{instant} is generated \mz{symbolwise} according to \eqref{eq: X_1 generation} and \eqref{eq: Y_1 generation}. \mz{We denote by $\mathcal{P}_{X_0} = \mathcal{N}(0,Q)$ the generative Gaussian probability distribution of the scalar source state, and by $\mathcal{P}_{X_1,Y_1|X_0,U_1}$ the channel probability distribution according to \eqref{eq: X_1 generation} and \eqref{eq: Y_1 generation}. These two distributions are given and fixed by the system.}

We define the control design for this setup as follows.

\begin{definition}
    For $n\in\mathbb{N}^{\star} = \mathbb{N} \backslash \{0\}$, a ``control design'' with causal encoder and noncausal decoder is a tuple of functions $c = (\{ f^{(t)}_{U_{1,t}|X_0^t}\}_{t=1}^n, g_{U_2^n|Y_1^n})$ defined by

    \begin{equation}
        f^{(t)}_{U_{1,t}|X_0^t}: \mathcal{X}_0^t \longrightarrow \mathcal{U}_1,\quad g_{U_2^n|Y_1^n}: \mathcal{Y}_1^n\longrightarrow \mathcal{U}_2^n,
    \end{equation}
    which induces a probability distribution over sequences of symbols
    \begin{equation}
         \prod_{t=1}^n \mathcal{P}_{X_{0,t}} \prod_{t=1}^n f^{(t)}_{U_{1,t}|X_0^t}\prod_{t=1}^n \mathcal{P}_{X_{1,t},Y_{1,t}|X_{0,t},U_{1,t}} g_{U_2^n|Y_1^n}. \label{distribution of sequences}
    \end{equation}

    We denote by $\mathcal{C}(n)$ the set of control designs with causal encoder and non-causal decoder.
\end{definition}

As in \cite{Treust2024power}, the objective is to determine the limit of the average pair of costs when block length $n\in \mathbb{N}^{\star}$ grows. 

\begin{definition}\label{def:ach cost, c-n}
    We define the two long-run cost functions 
    \vspace{-0.1cm}
    \begin{align}
        &c_P(u_1^n) = \frac{1}{n}\sum_{t=1}^n (u_{1,t})^2,\label{eq: c_P}\\
        &c_S(x_1^n, u_2^n) = \frac{1}{n}\sum_{t=1}^n(x_{1,t}-u_{2,t})^2,\label{eq: c_S}
    \end{align}
    averaged over $n\in\mathbb N^{\star}$ transmissions. The pair of costs $(P,S)\in\mathbb{R}^2$ is said to be achievable, if for all $\varepsilon>0$, there exists $\Bar{n}\in\mathbb N^{\star}$ such that for all $n\geq \Bar{n}$, there exists a control design $c\in \mathcal{C}(n)$ such that 
    \begin{equation}
        \mathbb E\Big[\big|P - c_P(U_1^n)\big| + \big|S - c_S(X_1^n, U_2^n)\big|\Big] \leq \varepsilon.\label{eq: achievable cost def}
    \end{equation}
    We denote by $\mathcal{R}\subset \mathbb{R}^2$ the set of all achievable pairs of costs obtained by the control designs with causal encoder and noncausal decoder.
\end{definition}

By construction, $\mathcal{R}$ is defined through asymptotic average costs as the block length $n \to \infty$. Consequently, $\mathcal{R}$ is a closed subset of $\mathbb{R}^2$, containing all limit points of achievable pairs of costs. \mz{Our main coding result characterizes the cost region $\mathcal{R}$ in a single-letter expression.}

\begin{theorem}\label{thm: c-n wits main}
    The pair of Witsenhausen costs $(P,S)\in \mathcal{R}$ if and only if there exists a joint distribution over the random variables $(X_0, W_1, W_2, U_1, X_1, Y_1, U_2)$ that decomposes according to
    \begin{equation}
\mathcal{P}_{X_0}\mathcal{P}_{W_1}\mathcal{P}_{W_2|X_0,W_1}\mathcal{P}_{U_1|X_0,W_1}\mathcal{P}_{X_1, Y_1|X_0, U_1}\mathcal{P}_{U_2|W_1, W_2, Y_1},\label{prob result}
    \end{equation}
    such that
    \begin{align}
        &I(W_1, W_2; Y_1) - I(W_2; X_0 | W_1) \geq 0,\label{info result}\\
        &P = \mathbb{E}\sbrackets{U_1^2}, \quad \quad S = \mathbb{E}\sbrackets{(X_1 - U_2)^2},\label{cost result}
    \end{align}
    where $\mathcal{P}_{X_0}$ and $\mathcal{P}_{X_1, Y_1|X_0, U_1}$ are the given  distributions, and where $W_1,W_2$ are two auxiliary random variables.
\end{theorem}
The achievability and converse proofs are provided in Appendices \ref{app: ach proof} and \ref{app: conv proof}. The coding construction relies on the block-Markov coding scheme of \cite{choudhuri2013causal}, extended to the coordination framework in \cite{le2015empirical}.

In Theorem \ref{thm: c-n wits main}, the source $\mathcal{P}_{X_0}$ and the channel $\mathcal{P}_{X_1, Y_1|X_0, U_1}$ are fixed by the system, and the other four probability distributions in \eqref{prob result} are the degrees of freedom, that must satisfy the information constraint \eqref{info result} and the cost constraints \eqref{cost result}. Furthermore, the two auxiliary RVs can be interpreted as follows: $W_1$ represents the independent codeword designed for the state-dependent channel with state $X_0$, consistent with the Shannon strategy \cite{shannonstrategy1958}. $W_2$ is correlated with both $X_0$ and $W_1$, acting as a description of these two symbols. Both $W_1$ and $W_2$ are made available to the noncausal decoder. This \mz{expression} explicitly captures the dual role of control in Witsenhausen counterexample: $W_1$ regulates the controlled state $U_1$, whereas $W_2$ facilitates communication to the second DM.

We can gain a few insights from Theorem \ref{thm: c-n wits main}.

\begin{remark}
    The decomposition of the probability distribution \eqref{prob result} is equivalent to the following Markov chains:
    \begin{align}
        \left\{  
        \begin{aligned}
            &X_0\text{ is independent of }W_1,\\
            & U_1 -\!\!\!\!\minuso\!\!\!\!- (X_0, W_1) -\!\!\!\!\minuso\!\!\!\!- W_2,\\
            &(X_1, Y_1)-\!\!\!\!\minuso\!\!\!\!- (X_0, U_1)    -\!\!\!\!\minuso\!\!\!\!- (W_1, W_2),\\
            & U_2 -\!\!\!\!\minuso\!\!\!\!- (W_1, W_2, Y_1) -\!\!\!\!\minuso\!\!\!\!- (X_0, U_1, X_1).
        \end{aligned}
        \right.
        \label{markov result}
    \end{align}
 The first two Markov chains are consequences of causal encoding. The third Markov chain is related to the processing order of the Gaussian channel. The last Markov chain comes from noncausal decoding and symbol-wise reconstruction. 
\end{remark}
\mz{\begin{remark}
   In our achievability proof, the information constraint \eqref{info result} reformulates as
   \begin{equation}
      I(W_2;X_0|W_1)< I(W_1;Y_1) + I(W_2;Y_1|W_1).\nonumber
   \end{equation}
   The LHS corresponds to the source-coding step in the block-Markov structure: Since the encoder is causal, upon it receives $X_0$ of the last block, it searches for the current index in $W_2$ such that they are jointly typical, given the previously selected codeword $W_1$. The RHS stands for the channel coding rate for two consecutive blocks since the decoder operates noncausally: 
   The current message is jointly decoded from $W_1$ with the current channel output $Y_1$, and $W_2$ with the previous channel output $Y_1$, given the last decoded codeword $W_1$.
   Thus, the achievability proof requires that the rate for the channel coding is bigger than the rate for the lossy source coding, as it in the joint source-channel coding problem.
\end{remark}}

\mz{Based on this coordination coding result for the cost region characterization}, we introduce the ZEC \mz{strategy}, a continuous-discrete hybrid code that achieves zero estimation cost with a substantially low power cost.




\section{Zero Estimation Cost Strategy}\label{sec: zec c-n}


\mz{Consider the following power minimization problem when the estimation cost is assumed to be zero:
\begin{align}
    \min_{f:X_0\mapsto U_1,\text{ s.t. } S=0}P\label{eq: power min problem}
\end{align}}
\noindent We begin with a helpful \mz{benchmark for the original scalar-valued setup \eqref{eq: single-shot controller}}. The proof of this \mz{result} is in Appendix \ref{app: proof of S=0 P=Q}.

\begin{proposition}\label{prop: single-shot zero est}
    For the original scalar-valued Witsenhausen counterexample with $N>0$, the minimum power cost \mz{required} for achieving zero estimation cost \mz{\eqref{eq: power min problem}} is $P = Q$. This is \mz{attained} by \mz{choosing} $U_1 = -X_0$.
\end{proposition}

This proposition states that \mz{in the scalar-valued setting, $S=0$ is possible only if the first DM fully cancels the source state. However, this requires an expensive power cost at least equal to the source variance $Q$.}

\mz{In contrast, our vector-valued formulation allows the DMs to exploit block coding so as to get a better result in \eqref{eq: power min problem}. The key idea is to use the coordination coding result in Theorem~\ref{thm: c-n wits main} to design suitable RVs $(X_0, W_1, W_2, U_1, X_1, Y_1, U_2)$, knowing that satisfying conditions \eqref{prob result}, \eqref{info result} and  \eqref{cost result} is equivalent to the existence of coding \mz{strategies} that achieve the pair of costs $(P,S)$.}

 Inspired by Witsenhausen's two-point approach \eqref{eq: 2-point}, we propose a novel coordination coding strategy that \mz{achieves zero estimation cost}, which we call the ZEC-2 \mz{strategy}. \mz{Note that under Definition \ref{def:ach cost, c-n} of achievable costs, \mz{$S=0$} means the averaged long-run estimation cost \eqref{eq: c_S} can be made arbitrarily small.} The ZEC-$2$ \mz{strategy} employs a pair of auxiliary RVs: $W_1$ is a continuous Gaussian codebook and $W_2$ is a discrete binary random variable that quantizes the source state. This resembles the hybrid analog-digital design investigated in \cite{skoglund2006hybrid}. Now, for given parameters $V_1\geq 0, a\geq 0, b\in\mathbb R$, the ZEC-2 \mz{strategy} is designed as:

\vspace{-0.3cm}
\begin{align}
\left|
\begin{aligned}
X_0 &\sim \mathcal{N}(0,Q),\\
W_1 &\sim \mathcal{N}(0,V_1),\quad W_1\indep X_0,\\
W_2 &= a\cdot \sign(X_0),\label{eq: ZEC-2 scheme}\\
U_1 &= W_1 + a\cdot \sign(X_0) + b\cdot X_0 \\ &= W_1 + W_2 + b\cdot X_0 ,\\
X_1 &= U_1 + X_0 = W_1 + W_2 + (b+1)X_0,\\
Y_1 &= X_1 + Z_1 = W_1 + W_2 + (b+1)X_0 + Z_1.
\end{aligned}
\right.
\end{align}

\mz{We can easily verify that the RVs above follow the joint distribution decomposition \eqref{prob result}.} Given that $(W_1,W_2,Y_1)$ are available at the decoder, the estimation cost \mz{becomes}
\begin{align*}
    S &= \mathbb E[(X_1-\mathbb E[X_1|W_1,W_2,Y_1])^2]\\
    &= (b+1)^2\mathbb E[(X_0 - \mathbb E[X_0|W_1,\mz{W_2},Y_1])^2].
\end{align*}
Hence, $S=0$ \mz{is achieved by choosing} $b = -1$. \mz{This} removes the residual uncertainty about $X_0$ and makes $X_1 = W_1+W_2$ a deterministic function of the auxiliary RVs that are \mz{subsequently recovered at} the decoder. Thus, the decoder can reconstruct $X_1$ perfectly, achieving zero long-run estimation cost. \mz{Therefore, }for the remainder of this section, we assume $b=-1$. This design resembles Witsenhausen’s two-point approach \eqref{eq: 2-point} \mz{in the sense} that the control action $U_1$ cancels the source signal $X_0$. \mz{The key difference is that the remaining components are the auxiliary RVs $(W_1, W_2)$, which can be reliably communicated to the decoder through block coding.}

The \mz{resulting power cost is}
\begin{align}
 P = \mathbb{E}[U_1^2]&= \mathbb E[(W_1 + a\cdot \sign(X_0) - X_0)^2]\nonumber \\
    &= V_1 + \mz{P_2(a),}\label{eq: ZEC2 power constraint}
\end{align}
\mz{where
\begin{align}
    P_2(a) = Q+a^2-2a\sqrt{\frac{2Q}{\pi}}.\label{eq: two-point cost P}
\end{align}
is actually the power cost function for Witsenhausen's two-point strategy \eqref{eq: 2-point}.}
\mz{Now, it remains to impose the information constraint \eqref{info result}:}
\begin{align}
    &I(W_1,W_2;Y_1) - I(W_2;X_0|W_1)\nonumber\\
    &=h(Y_1)-h(Y_1|W_1,W_2)-h(W_2|W_1)+h(W_2|X_0,W_1)\nonumber\\
    &=h(Y_1) - \frac{1}{2}\log(2\pi eN) - 1\geq 0.\label{eq: ZEC2 info constraint}
\end{align}
\mz{The differential entropy $h(Y_1)$ is computed} with regard to the Gaussian mixture distribution
\begin{align*}
    f_{Y_1}(y)\!=\! \frac{1}{2\sqrt{(V_1+N)}}\!\sbrackets{\phi\brackets{\frac{y-a}{\sqrt{V_1+N}}}+\phi\brackets{\frac{y+a}{\sqrt{V_1+N}}}}\!.
\end{align*}
\mz{Since the entropy of a Gaussian mixture generally does not admit a closed-form expression,} \mz{the} methods discussed in \cite{huber2008entropy, kim2015entropy} can be employed for numerical \mz{evaluation}.

We have the following theorem \mz{summarizing} the cost functions for the ZEC-2 \mz{strategy}.

\begin{theorem}\label{thm: zec}
The achievable long-run estimation cost for the ZEC-2 coding scheme \eqref{eq: ZEC-2 scheme} is given by
    \begin{align}
        S_{\mathsf{ZEC}\text{-}\mathsf{2}}(P) = 0, \text{ for  }P\geq P_{\mathsf{ZEC}\text{-}\mathsf{2}}^*, \label{eq: cost function lossless}
    \end{align}
    where the value 
    \begin{align}
P_{\mathsf{ZEC}\text{-}\mathsf{2}}^* = \quad &\min_{V_1\geq 0,a\geq 0}   \qquad  V_1 + P_2(a)\label{eq: opt ZEC2}\\
&\;\,\text{s.t.} \quad h(Y_1) - \frac{1}{2}\log(2\pi eN) - 1\geq 0.\nonumber
\end{align}

\end{theorem}
Note that the differential entropy term $h(Y_1)$ is not always convex with respect to the parameters $(V_1,a)$, hence this is not a convex optimization problem.

From \eqref{eq: opt ZEC2}, since \mz{$V_1\geq 0$}, the minimum power cost $P_{\mathsf{ZEC}\text{-}\mathsf{2}}^*$ \mz{required by} the ZEC-2 \mz{strategy} must satisfy
\begin{align}
    P_{\mathsf{ZEC}\text{-}\mathsf{2}}^*\geq \min_a P_2(a) = Q\brackets{1-\frac{2}{\pi}}\triangleq P_2^{\min} ,\label{eq: ZEC2 power more than 2-point}
\end{align}
where $P_2^{\min}$ denotes the lowest possible power \mz{cost of} the \mz{Witsenhausen's} two-point strategy \eqref{eq: 2-point}. This indicates that the ZEC-2 \mz{strategy} achieves an improved estimation performance, at the expense of a higher power consumption than the original two-point strategy.

We \mz{next} extend the ZEC-2 \mz{strategy} to the ZEC-$k$ \mz{strategy} by generalizing the binary auxiliary RV $W_2$ to a more expressive $k$-point quantization \mz{of the source state}.

\subsection{ZEC-$k$ \mz{Strategy}}\label{subsec: zec-n}


The idea of the ZEC-$k$ \mz{strategy} is inspired by the hierarchical search approach that generalizes Witsenhausen’s original two-point strategy \mz{\eqref{eq: 2-point}} to a $k$-point quantization strategy in \cite{lee2001hierarchical}. We adopt a symmetric indexing convention where the quantization levels and grids are placed symmetrically around zero. Thus, we optimize only the positive half of the parameters. \mz{First}, we formally define the $k$-point quantization function.

\begin{definition}[$k$-point quantization]\label{def: n-point quantization}
Let $k\in\mathbb N^{\star}$ and define $m=\ceil{\frac{k}{2}}$.
The construction of the $k$-point quantization function depends on whether $k$ is odd or even:

\begin{itemize}
        \item If $k=2m-1$, we take the center level \mz{$a_1=0$} and define the quantization levels and decision boundaries as symmetric sequences:
    \begin{align}
        \begin{aligned}&-a_m<...<-a_1 = 0 =a_1<...<a_m,\\
        &-B_m<...<-B_1 = 0 = B_1<...<B_m.
        \end{aligned}\label{eq: n-point odd}
    \end{align}
    \item If $k=2m$, we take $a_1\neq 0$ and the quantization levels and boundaries are defined as:
    \begin{align}
        \begin{aligned}&-a_m<...<-a_1<0<a_1<...<a_m,\\
        &-B_m<...<-B_1 = 0 = B_1<...<B_m.
        \end{aligned}\label{eq: n-point even}
    \end{align}    
\end{itemize}
In both cases, we define the $k$-point quantization function $Q_k(x):\mathbb R \rightarrow \{-a_m,..., -a_1,a_1,...,a_m\}$ by
    \begin{align}
        Q_k(x)& = \left\{
          \begin{aligned}
              &   -a_i   &\quad& \text{if } x\in(-B_{i+1},-B_i], \\
              &  +a_i   &\quad& \text{if } x\in[B_i,B_{i+1}), \\
          \end{aligned}   
        \right. \label{eq: n step quantization}\\
        &\quad \qquad\qquad\qquad i=1,...,m\nonumber\\
        &=\sum_{i=1}^{m}a_i( \mathbf{1}_{x\in[B_i, B_{i+1})} - \mathbf{1}_{x\in(-B_{i+1}, -B_i]}),\forall x\in\mathbb R.\nonumber
    \end{align}
    For notation consistency, we set $B_{k+1} = \infty$.
\end{definition} 

Figure~\ref{fig: 3, 4 point function} illustrates an example of the 3-point and 4-point quantization functions.


\begin{figure}[t]
  \centering


\begin{tikzpicture}[>=stealth, scale=1, every node/.style={scale=0.8}]

    \draw[->, thin] (-4,0) -- (4,0) node[right] {\(x\)};
    \draw[->, thin] (0,-2) -- (0,2.1) node[right] {\(Q_3(x)\)};

    \draw[very thick] (0,0) -- (-1.5,0);  
    \draw[very thick] (-1.5,0) -- (-1.5, -1);
    \draw[very thick] (1.5,0) -- (1.5, 1);

    \draw[very thick] (-1.5, -1) -- (-4,-1);
    \draw[very thick] (0, 0) -- (1.5, 0);
    \draw[very thick] (1.5, 1) -- (4,1);

    \fill (0,0) circle (1.5pt);
    \fill (0,1) circle (1pt) node[right] {\(a_2\)};
    \fill (0,-1) circle (1pt) node[right] {\(-a_2\)};
    \fill (1.5,0) circle (1pt) node[below] {\(B_2\)};
    \fill (-1.5,0) circle (1pt) node[above] {\(-B_2\)};


     \draw[->, thin] (-4.0,-5.0) -- (4.0,-5.0) node[right] {\(x\)};
    \draw[->, thin] (0,-7) -- (0,-2.9) node[right] {\(Q_4(x)\)};

\draw[very thick] (0,-5.7) -- (-2,-5.7);  
\draw[very thick] (0,-4.3) -- (2,-4.3);  
\draw[very thick] (-2,-6.4) -- (-4,-6.4);
\draw[very thick] (2,-3.6) -- (4,-3.6);
\draw[very thick] (-2,-6.4)-- (-2,-5.7);
\draw[very thick] (0,-5.7) -- (0,-4.3);
\draw[very thick] (2,-3.6) -- (2,-4.3);


    \fill (0,-5) circle (1.5pt);
    \fill (0,-3.6) circle (1.5pt) node[left] {\(a_2\)};
    \fill (0,-4.3) circle (1.5pt) node[left] {\(a_1\)};
    \fill (0,-6.4) circle (1.5pt) node[right] {\(-a_2\)};
    \fill (0,-5.7) circle (1.5pt) node[right] {\(-a_1\)};
    \fill (-2,-5) circle (1.5pt) node[above] {\(-B_2\)};
    \fill (2,-5) circle (1.5pt) node[above] {\(B_2\)};


\end{tikzpicture}

  \caption{Illustration of the quantization step function $Q_k(\cdot)$ with $k=3$ (up) and $k=4$ (down): The 3-point quantization is parameterized by $0=a_1\leq a_2,0=B_1\leq B_2$, whereas the 4-point function is parameterized by $0< a_1< a_2,0=B_1\leq B_2$}
  \label{fig: 3, 4 point function}
\end{figure}

Based on this $k$-point quantization, we first recall the Witsenhausen $k$-point control strategy proposed in \cite{lee2001hierarchical}. Since the estimation cost function derived in the reference involves Fisher information, we provide a more explicit expression of the cost functions in Theorem \ref{thm: n-point strategy}. In order to do this, we define the probability mass of $X_0$ lying in each quantization interval by
\begin{align}
    p_i& \triangleq \mathbb P(X_0\in(-B_{i+1}, -B_i])= \mathbb P(X_0\in[B_i, B_{i+1}))\nonumber\\
       &=\Phi\brackets{\frac{B_{i+1}}{\sqrt{Q}}} - \Phi\brackets{\frac{B_{i}}{\sqrt{Q}}},\quad \text{for }i=1,...,m.\label{eq: p_i}
\end{align}

\begin{figure*}[t]
\begin{align}
    P_k(\mathbf{a}_1^m,\mathbf{B}_1^m) &=  Q - 4\sqrt{Q}\sum_{i=1}^m a_i\brackets{\phi\brackets{\frac{B_i}{\sqrt{Q}}} - \phi\brackets{\frac{B_{i+1}}{\sqrt{Q}}}}+2\sum_{i=1}^ma_i^2p_i \label{eq: P_n, n-point}\\
    S_k(\mathbf{a}_1^m,\mathbf{B}_1^m) &= 2\sum_{i=1}^m a_i^2p_i  - \int \frac{\sbrackets{\sum_{i=1}^m a_i\cdot p_i\brackets{\phi\brackets{\frac{y-a_i}{\sqrt{N}}} - \phi\brackets{\frac{y+a_i}{\sqrt{N}}}}}^2}{\sqrt{N}\sum_{i=1}^m p_i\brackets{\phi\brackets{\frac{y-a_i}{\sqrt{N}}} + \phi\brackets{\frac{y+a_i}{\sqrt{N}}}}}dy\label{eq: S_n, n-point}
\end{align}
\noindent\rule{\textwidth}{0.4pt}
\vspace{-0.5em}
\end{figure*}

\begin{theorem}[$k$-point strategy]\label{thm: n-point strategy}
    For $k\in\mathbb N^\star$, $m=\ceil{\frac{k}{2}}$, and given $(\mathbf{a}_1^m = [a_1,...,a_m]^\top, \mathbf{B}_1^m = [B_1,...,B_m]^\top)$ satisfying either \eqref{eq: n-point odd} or \eqref{eq: n-point even}, the Witsenhausen $k$-point strategy is defined by
    \begin{align}
        U_1 = Q_k(X_0) - X_0,\label{eq: k-point}
    \end{align}
    where $Q_k(\cdot)$ is the quantization function defined in \eqref{eq: n step quantization} parameterized by $\mathbf{a}_1^m,\mathbf{B}_1^m$.
    The power and estimation costs induced by this strategy are given by \eqref{eq: P_n, n-point} and \eqref{eq: S_n, n-point}, where the quantity $p_i$ is defined by \eqref{eq: p_i}.
\end{theorem}

This control strategy \mz{\eqref{eq: k-point} designs} $X_1$ \mz{to be} a more refined quantization symbol, improving the trade-off between the power and estimation costs. We now integrate this idea to generalize our ZEC \mz{strategy}: Given $V_1\geq 0$ and $(\mathbf{a}_1^m = [a_1,...,a_m]^\top, \mathbf{B}_1^m = [B_1,...,B_m]^\top)$ satisfying either \eqref{eq: n-point odd} or \eqref{eq: n-point even}, the ZEC-$k$ \mz{strategy} is given by the following:
\begin{align}
\left|
\begin{aligned}
     X_0&\sim\mathcal{N}(0,Q),\\
     W_1&\sim\mathcal{N}(0,V_1),\quad \mz{W_1\indep X_0,}\\
     W_2& = Q_k(X_0),\\
    U_1 &= W_1 + Q_k(X_0) - X_0 = W_1+W_2-X_0,\\
    X_1 &= U_1+X_0 = W_1+W_2,\\
    Y_1 &= X_1+Z_1 = W_1+W_2+Z_1,\quad Z_1\sim\mathcal{N}(0,N),
\end{aligned}
\right.\label{eq: n-ZEC}
\end{align}
Here, $Q_k(\cdot)$ is the $k$-point quantizer defined as in \eqref{eq: n step quantization}, and as in ZEC-2 to ensure zero estimation cost, $X_1$ is fully determined by the auxiliary RVs $(W_1,W_2)$.

Following the same derivation procedure in the last section, the total power cost in this setup becomes:
\begin{align*}
    P = \mathbb E[U_1^2]= V_1 + P_k(\mathbf{a}_1^m,\mathbf{B}_1^m),
\end{align*}
where $P_k(\mathbf{a}_1^m,\mathbf{B}_1^m)$ is given in \eqref{eq: P_n, n-point}. 

To ensure the coordination coding scheme is valid, the ZEC-$k$ scheme must satisfy the information constraint

\begin{align}
    0\leq &I(W_1,W_2;Y_1) - I(W_2;X_0|W_1)\nonumber \\
    &= h(Y_1)-\frac{1}{2}\log 2\pi eN - H(Q_k(X_0))\label{eq: info constraint - n}
\end{align}
\mz{where the differential entropy $h(Y_1)$ is computed according to the marginal density function of $Y_1$} 
\begin{align*}
    f_{Y_1}(y)=\sum_{i=1}^m \frac{p_i}{\sqrt{V_1+N}}\sbrackets{\phi\brackets{\frac{y-a_i}{\sqrt{V_1+N}}} +\phi\brackets{\frac{y+a_i}{\sqrt{V_1+N}}}},
\end{align*}
and the term $H(Q_k(X_0))$ is given by
\mz{\begin{align*}
\left\{
          \begin{aligned}
              &  -2\brackets{p_1\log(2\cdot p_1) +\sum_{i=2}^m p_i\log p_i}    &\quad& \mz{k=2m-1}, \\
              &   -2\sum_{i=1}^m p_i\log p_i   &\quad& \mz{k=2m}, 
          \end{aligned}   
        \right.
\end{align*}}



We have the following result:
\begin{theorem}\label{thm: zec-n}
   For the given power cost $P\geq 0$, the estimation cost for the ZEC-$k$ coding \mz{strategy} \eqref{eq: n-ZEC} is given by
    \begin{align}
        S_{\mathsf{ZEC}\text{-}k}(P) = 0, \text{ for  }P\geq P^*_{\mathsf{ZEC}\text{-}k}, \label{eq: cost function ZEC-n}
    \end{align}
    where the value  
\begin{align}
P_{\mathsf{ZEC}\text{-}k}^* = \quad &\min   \qquad  V_1 + P_k(\mathbf{a}_1^m,\mathbf{B}_1^m)\label{eq: P_ZEC-n},\\
&\text{s.t.} \quad h(Y_1) - \frac{1}{2}\log(2\pi eN) - H(Q_k(X_0))\geq 0.\nonumber
\end{align}

\end{theorem}

Again, from \eqref{eq: P_ZEC-n}, we get that the power consumption needed for the ZEC-$k$ \mz{strategy} is more than the power required for the \mz{scalar} $k$-point strategy, denoted by
\begin{align}
    P_k^{\min} = \min_{\mathbf{a}_1^m,\mathbf{B}_1^m}P_k(\mathbf{a}_1^m,\mathbf{B}_1^m).\label{eq: P_k^min}
\end{align}

\section{Causal-encoding with channel feedback}\label{sec: c-n w-f characterization}
\begin{figure}[t]
  \centering


\begin{tikzpicture}[scale=0.9, every node/.style={scale=0.9}]
    \draw (2,0) rectangle (3,1);
    \draw (6.8,0) rectangle (7.8,1);

    \draw (4.2,0.5) circle (0.2) node {$+$};
    \draw (5.6,0.5) circle (0.2) node {$+$};

    \filldraw (1,-0.5) circle (2pt) node[left] {$X_{0,t}$};
    \filldraw (5.6,1.5) circle (2pt) node[above] {$Z_{1,t}$};

    \draw[->] (1,-0.5) -- (1,0.5) -- (2,0.5);
    \draw[->] (1,-0.5) -- (4.2,-0.5) -- (4.2,0.3);
    \draw[->] (3,0.5) -- (4,0.5);
    \draw[->] (4.4,0.5) -- (5.4,0.5);
    \draw[->] (4.9,0.5) -- (4.9,-0.5) -- (8.8,-0.5);
    \draw[->] (5.6,1.5) -- (5.6,0.7);
    \draw[->] (5.8,0.5) -- (6.8,0.5);
    \draw[->] (7.8,0.5) -- (8.8,0.5);
    \draw[->] (6.3,1) -- (6.3, 2.1) -- (2.5,2.1) -- (2.5,1); 
     \draw[dotted, ->] (0.6, -0.2) -- (0.6, 2.3) -- (7.3,2.3) -- (7.3, 1); 

    \node at (1.5,0.8) {$X_0^t$};
    \node at (3.5,0.8) {$U_{1,t}$};
    \node at (4.9,0.8) {$X_{1,t}$};
    \node at (6.3,0.8) {$Y_{1}^t$};
    \node at (3,1.5) {$Y_{1}^{t-1}$};
    \node at (7.9,1.7) {$X_{0}^{t-1}$};
    \node at (8.3,0.8) {$U_{2,t}$};
    \node at (8.3,-0.2) {$X_{1,t}$};
    \node at (2.5,0.5) {$C_1$};
    \node at (7.3,0.5) {$C_2$};
\end{tikzpicture}



  \label{fig:c-n w-f}
\end{figure}
We consider the vector-valued \mz{formulation} of Witsenhausen counterexample with the presence of perfect channel feedback \mz{sequence} at the causal encoder, see Figure \ref{fig:c-n w-f}.

\subsection{System Model with Feedback}

We first define control designs for this setup as in the last section.
\begin{definition}
    For $n\in\mathbb{N}$, a ``control design'' with causal encoder and noncausal decoder with channel feedback is a tuple of functions $c = (\{ f^{(\mathsf{f},t)}_{U_{1,t}|X_0^t,Y_1^{t-1}}\}_{t=1}^n, g_{U_{2}^n|Y_1^n})$ defined by
    \begin{align}
        f^{(\mathsf{f},t)}_{U_{1,t}|X_0^t,Y_1^{t-1}}: \mathcal{X}_0^t\times \mathcal{Y}_1^{t-1} \longrightarrow \mathcal{U}_{1},\quad g_{U_{2}^n|Y_1^n}: \mathcal{Y}_1^n\longrightarrow \mathcal{U}_{2}^n,\label{eq: control design with feedback}
    \end{align}
    which induces a distribution over sequences of symbols:
    \begin{equation}
         \prod_{t=1}^n \mathcal{P}_{X_{0,t}} \prod_{t=1}^n f^{(\mathsf{f},t)}_{U_{1,t}|X_0^t,Y_1^{t-1}}\prod_{t=1}^n \mathcal{P}_{X_{1,t},Y_{1,t}|X_{0,t},U_{1,t}} g_{U_{2}^n|Y_1^n}, \label{eq: c-n w-f distribution of sequences}
    \end{equation}
where $Y_1^0 = \emptyset$. We denote by $\mathcal{C}_{\mathsf{f}}(n)$ the set of control designs with causal encoder and noncausal decoder with channel feedback.
\end{definition} 

Analogous to the no-feedback setup, we define the achievable pair of costs $(P,S)$ as the limit point of the averaged long-run costs induced by any control design, as in Definition \ref{def:ach cost, c-n}, and we denote by $\mathcal{R}_{ \mathsf{f}}$ the region of achievable pairs of costs in this setting. For brevity, we omit restating the definition and proceed directly to characterize the region $\mathcal{R}_{\mathsf{f}}$.

\begin{theorem}\label{theorem: c-n w-f}
    The pair of Witsenhausen costs $(P,S)\in \mathcal{R}_{\mathsf{f}}$ if and only if there exists a joint distribution over the random variables $(X_0, W_1, U_1, X_1, Y_1, U_2)$ that decomposes according to
    \begin{equation}
\mathcal{P}_{X_0}\mathcal{P}_{W_1}\mathcal{P}_{U_1|X_0,W_1}\mathcal{P}_{X_1, Y_1|X_0, U_1}\mathcal{P}_{U_2|X_0, W_1, Y_1},\label{eq: c-n w-f prob result}
    \end{equation}
    such that
    \begin{align}
        &I(W_1; Y_1) - I(U_2; X_0 | W_1,Y_1) \geq 0,\label{eq: c-n w-f info result}\\
        &P = \mathbb{E}\sbrackets{U_1^2}, \quad \quad S = \mathbb{E}\sbrackets{(X_1 - U_2)^2},\label{eq: c-n w-f cost result}
    \end{align}
    where $\mathcal{P}_{X_0}$ and $\mathcal{P}_{X_1, Y_1|X_0, U_1}$ are the given  distributions, and $W_1$ is an auxiliary RV.
\end{theorem}

We provide the achievability proof and converse proof for Theorem \ref{theorem: c-n w-f} in Appendix \ref{app: achievability proof c-n w-f} and \ref{app: converse proof c-n w-f} respectively. \mz{The coding scheme for the achievability proof extends the approach in \cite[App. B.]{Letreust2015empirical} which dealt with a state-independent channel. 
The converse proof is adapted from the converse proof of Theorem III.2 in \cite[App. C.]{Letreust2015empirical}. We extend the arguments from this reference to apply to the state-dependent channel in our setting, without affecting the outcome of the resulting single-letter characterization.} \mz{It is important to note that the feedback in \eqref{eq: control design with feedback} is strictly causal: At time $t$, the encoder observes $Y_1^{t-1}$ but not $Y_{1,t}$. Consequently, $U_{1,t}$ cannot be coordinated with the current channel output $Y_{1,t}$, and hence the single-letter distribution of $U_1$ does not depend on $Y_1$ in \eqref{eq: c-n w-f prob result}.}

\begin{remark}
The decomposition of the distribution \eqref{eq: c-n w-f prob result} is equivalent to the following Markov chains
    \begin{align}
        \left\{  
        \begin{aligned}
            &X_0\text{ is independent of }W_1,\\
            &(X_1,Y_1)-\!\!\!\!\minuso\!\!\!\!- (X_0, U_1)    -\!\!\!\!\minuso\!\!\!\!- W_1,\\
            & U_2 -\!\!\!\!\minuso\!\!\!\!- (X_0, Y_1,W_1) -\!\!\!\!\minuso\!\!\!\!- (U_1, X_1).
        \end{aligned}
        \right.
        \label{c-n w-f markov result}
        \end{align}
\end{remark}

\mz{Moreover,} the information constraint \eqref{eq: c-n w-f info result} in Theorem \ref{theorem: c-n w-f} could be directly obtained from the information constraint in Theorem \ref{thm: c-n wits main} if we identify the following RVs based on the feature that perfect channel feedback is available:
\begin{align}
    &\tilde{X}_0  = (X_0,Y_1),\label{X_0tilde}\\
    &\tilde{W}_2  = U_2.\label{W_2tilde}
\end{align}
Hence we have 
\begin{align*}
    &I(W_1,\tilde{W}_2;Y_1) - I(\tilde{W}_2;\tilde{X}_0|W_1)\\
    &=I(W_1,U_2;Y_1) - I(U_2;X_0,Y_1|W_1)\\
    &= I(W_1;Y_1) + I(U_2;Y_1|W_1) -I(U_2;Y_1|W_1) \\
    &\qquad -I(U_2;X_0|W_1,Y_1)\\
    &=I(W_1;Y_1) -I(U_2;X_0|W_1,Y_1).
\end{align*}
The identification \eqref{X_0tilde} means that the channel feedback $Y_1$ plays the role of an additional source symbol. And \eqref{W_2tilde} means that it is optimal for the encoder to recommend to the decoder the action $U_2$ to implement.

\begin{remark}\label{rmk: feedback enlarges cost region}
    Compared to Theorem \ref{thm: c-n wits main} for the setup without channel feedback, the presence of channel feedback enables the decoder to directly coordinate with the source state $X_0$ instead of its noisy representation $W_2$ as in \eqref{prob result}. Therefore, feedback helps with better communication and also enlarges the set of achievable cost pairs, i.e., $\mathcal{R}\subseteq \mathcal{R}_{\mathsf{f}}$. Similar observations have been also pointed out in \cite{bross2017rate, Letreust2015empirical, Letreust2021state}
\end{remark}

\subsection{Zero Estimation Cost Strategy with Feedback}\label{subsec: zec c-n w-f}

In this section, we derive an analogous ZEC \mz{strategy}, so-called ZEC-f \mz{strategy}, that is based on the coordination coding result of Theorem \ref{theorem: c-n w-f} that further improves the power-estimation cost trade-off performance.

Given parameters $V_1\geq 0,a,b\in\mathbb R$, we consider the following design
\begin{align}
\left|
\begin{aligned}
&X_0 \sim \mathcal{N}(0,Q),\\
&W_1\sim\mathcal{N}(0,V_1),\quad \mz{W_1\indep X_0,}\\
&U_1 = W_1 + a\cdot \text{sign}(X_0) + b\cdot X_0,\\
    &X_1 = U_1 + X_0 = W_1 + a\cdot \text{sign}(X_0) + (b+1) X_0,\\
    &Y_1 = W_1 + a\cdot \text{sign}(X_0) + (b+1) X_0 + Z_1.
\end{aligned}
\right.\label{eq: ZEC feedback scheme}
\end{align}  
This strategy also follows the idea of designing the system state $X_1$ being a deterministic function of $(X_0,W_1)$, which is known to the decoder. Notably, due to channel feedback, the source state $X_0$ is transmitted to the decoder perfectly, hence, even if $b\neq -1$ \mz{as required in ZEC without feedback}, we still could exactly reconstruct $X_1$ from its MMSE estimator, \mz{i.e.,}
\begin{align}
    U_2 &= \mathbb E[X_1|X_0,W_1,Y_1]\nonumber\\
    &= W_1 + a\cdot \text{sign}(X_0) + (b+1) X_0=X_1\label{eq: U_2=X_1}
\end{align}
which results in an achievable estimation cost $S=0$.




Before stating our cost function result, we first define the skew-normal distribution.

\begin{definition}
    The random variable $X$ is distributed according to the skew-normal distribution $X\sim\mathcal{SN}(u, \sigma, \alpha)$ with location $u$, scale $\sigma>0$ and skewness parameter $\alpha$, if the probability density function writes
\begin{align}
    f(x) = \frac{2}{\sigma}\cdot\phi\brackets{\frac{x-u}{\sigma}}\Phi\brackets{\alpha\cdot\frac{x-u}{\sigma}}.
\end{align}
\end{definition}


\begin{theorem}\label{thm: ZEC feedback cost}
    For a given power budget $P\geq 0$, the achievable estimation cost for ZEC-f \mz{strategy} is
    \begin{align}
        S_{\mathsf{ZEC}\text{-}\mathsf{f}}(P) = 0, \text{ for }P\geq P_{\mathsf{ZEC}\text{-}\mathsf{f}}^*, 
    \end{align}
    where the minimum required power
\begin{align}
    P_{\mathsf{ZEC}\text{-}\mathsf{f}}^* &= \min_{V_1\geq 0,a,b\in\mathbb R}V_1 + a^2 + b^2 Q + 2ab\sqrt{\frac{2Q}{\pi}},\label{eq: power ZEC-f}\\
    &\quad \text{s.t. }h(Y_1)-h(X)-\frac{1}{2}\log 2\pi eN\geq 0,\label{eq: info constraint ZEC-f}
\end{align}
and the differential entropies are computed according to 
\begin{align}
    f_{Y_1}(y) &= \frac{1}{2}\mathcal{SN}(a, \sigma, \alpha)+ \frac{1}{2}\mathcal{SN}(-a, \sigma, -\alpha)\label{eq: pdf Y_1, gaussian W_1},\\
    f_X(x) &= \frac{1}{\sqrt{\delta^2Q}}\phi\brackets{\frac{x-a}{\sqrt{\delta^2Q}}}\mathbf{1}_{\text{sign}(\delta)\cdot (x-a)\geq 0} \nonumber\\
    &\qquad + \frac{1}{\sqrt{\delta^2Q}}\phi\brackets{\frac{x+a}{\sqrt{\delta^2Q}}}\mathbf{1}_{\text{sign}(\delta)\cdot (x+a)\leq 0}.\label{eq: pdf X, piecewise gaussian}
\end{align}
Here, we take $\delta = b+1, \sigma^2 = V_1 + \delta^2Q + N, \alpha = \frac{\sqrt{Q}\delta}{\sqrt{V_1+N}}$.
\end{theorem}
The proof of Theorem \ref{thm: ZEC feedback cost} is stated in Appendix \ref{app: proof of ZEC feedback info constraint}. The following corollary is an immediate consequence of Theorem~\ref{thm: ZEC feedback cost}.

\begin{corollary}\label{cor: ZECZEP}
    If the Gaussian variances \mz{$Q,N>0$} satisfy 
    \begin{align}
        2\pi e\cdot\frac{QN}{Q+N}\leq 1,\label{eq:ZECZEP}
    \end{align}
    then, the ZEC-f \mz{strategy} achieves zero-power zero-estimation cost performance.
\end{corollary}

The proof is provided in Appendix \ref{app: proof of cor ZECZEP}. Importantly, Corollary \ref{cor: ZECZEP} also provides a converse result: under the condition \eqref{eq:ZECZEP}, ZEC-f \mz{strategy} is the optimal strategy to conduct. To clarify this result, we have the following remark.

\begin{remark}
   Achieving zero-power and zero-estimation cost in Definition \ref{def:ach cost, c-n} means that the averaged long-run quantities \eqref{eq: c_P} and \eqref{eq: c_S} can be made arbitrarily close to zero as $n\rightarrow\infty$. This does not mean that the two quantities can be made strictly identical to zero. Using our coordination-coding results, we show the existence of a sequence of control designs whose induced cost pairs converge to zero. Usually, at the strictly zero boundary, a valid coding scheme is not guaranteed to exist. Analogous behavior appears in the wideband regime of communication channels  \cite{verdu2002spectral}, where the author showed that the minimum energy-per-bit is achieved only as spectral efficiency tends to zero, i.e., ``low, but non-zero'' rate.
\end{remark}

    This counterintuitive result naturally raises the question: what is communicated when the averaged long-run cost $P=0$, i.e., when the first DM only effectively transmits a finite number of times? In the presence of perfect channel feedback, the first DM has access to all the information available at the second DM. As a result, when the noise level $N\geq 0$ is small, the encoder can easily predict the noise to refine its actions and correct past errors online, such as in the classical Schalkwijk–Kailath coding scheme \cite{schalkwijk1966coding}.

\mz{\subsection{Jointly Gaussian Cost}\label{subsec: joint Gaussian}}

\mz{Before proceeding to the numerical evaluation, 
we conclude our discussion by determining the coordination coding cost region in Theorem~\ref{thm: c-n wits main} and Theorem \ref{theorem: c-n w-f} under the restriction that all involved RVs are \textit{jointly Gaussian}. This Gaussian restriction provides a benchmark in the spirit of Witsenhausen’s original analysis \cite{witsenhausen1968}.}

\mz{As a canonical jointly Gaussian policy for the classical scalar-valued Witsenhausen's counterexample, we first recall Witsenhausen's optimal linear strategy \cite[Lemma 11]{witsenhausen1968}: $U_1=-\sqrt{\frac{P}{Q}}X_0$, if $P\leq Q$, otherwise $U_1=-X_0 + \sqrt{P-Q}$. This induces the estimation cost function
    \begin{align}
        S_{\ell}(P)=\begin{cases}\frac{(\sqrt{Q}-\sqrt{P})^{2}\cdot N}{(\sqrt{Q}-\sqrt{P})^{2}+N} & \text{ if } P\in[0,\ Q],\\ 0 &\text{ otherwise } .\end{cases}\label{eq: opt linear cost}
    \end{align}
Note that the function $P\mapsto S_{\ell}(P)$ is in general not convex. }

\mz{From Theorem \ref{thm: c-n wits main}, when there's no channel feedback, given a power cost parameter $P\geq 0$, the optimal jointly Gaussian estimation cost writes
\begin{align}
&S_{\mathsf {G}}(P) = \inf_{\mathcal{P}\in\mathbb P_{\mathsf {G}}(P)}\mathbb E\Big[\big(X_1 - \mathbb E\big[X_1\big|W_1,W_2,Y_1\big]\big)^2\Big],\label{eq: gauss opt problem}
\end{align}
where the set
\begin{align}
&\mathbb P_{\mathsf {G}}(P) = \Bigl\{  \mathcal{P}_{W_1},\mathcal{P}_{W_2|X_0, W_1},\mathcal{P}_{U_1|X_0, W_1}:\quad P = \mathbb E\big[U_1^2\big],\nonumber\\
&\qquad I(W_1,W_2;Y_1) - I(W_2;X_0|W_1)\geq0,\text{ and }\label{eq: gauss opt domain w/o feedback}\\
&\qquad X_0,W_1,W_2,U_1,X_1,Y_1,U_2\text{ are jointly Gaussian}\Bigr\}.\nonumber
\end{align}
\indent Similarly to \cite[Theorem 2]{Treust2024power} for noncausal-encoding and causal-decoding setup, we now show that the estimation cost for joint Gaussian RVs in \eqref{eq: gauss opt problem} is given by the convex envelope of the best linear cost function $S_{\ell}$ in \eqref{eq: opt linear cost}.}

\mz{\begin{theorem}\label{thm: best gaussian policy}
When the RVs $(X_0, W_1, W_2, U_1, X_1, Y_1, U_2)$ are restricted to be jointly Gaussian, 
the closed form of the estimation cost \eqref{eq: gauss opt problem} is given by
\begin{flalign} S_{\mathsf {G}}(P) =& \begin{cases} \frac {N\cdot (Q-N-P)}{Q} & \text {if } Q> 4N \text { and } P\in [P_{1},P_{2}],\\ S_{\ell }(P) & \text {otherwise, } \end{cases} \label{opt gaussian cost}\end{flalign}
where 
\vspace{-0.2cm}
\begin{align}
P_1 &= \frac{1}{2}(Q-2N-\sqrt{Q^2-4QN}),\label{eq: p1}\\
P_2 &= \frac{1}{2}(Q-2N+\sqrt{Q^2-4QN}).\label{eq: p2}
\end{align}
\end{theorem}}

\mz{The proof of Theorem \ref{thm: best gaussian policy} is given in Appendix \ref{app: proof of gaussian design}. The estimation-cost function $S_{\mathsf {G}}(P)$ is attainable with a simple time-sharing strategy between the two operation points $P_1$ and $P_2$, even when both DMs operate causally.}

\mz{In fact, when all involved RVs are restricted to be Gaussian, the same cost $S_{\mathsf {G}}(P)$ is optimal for several superior causal configurations -- those that permit feedback or feedforward links. The following corollary unifies these cases.
\begin{corollary}\label{cor: causal gaussian cost}
    Under the jointly Gaussian constraint, the optimal cost for each of the following causal coding setups is given by the same function $S_{\mathsf G}(P)$ in \eqref{opt gaussian cost}:
    \begin{enumerate}
        \item causal-encoding noncausal-decoding with channel feedback, as in \eqref{eq: control design with feedback},
        \item noncausal-encoding causal-decoding with or without source feedforward, as in \cite{Treust2024power}: 
        \begin{align*}
            &f_{U_1^n|X_0^n}:\mathcal{X}_0^n\rightarrow\mathcal{U}_1^n,\quad \\
            &g^{(\mathsf{f},t)}_{U_{2,t}|Y_1^t(,X_0^{t-1})}: \mathcal{Y}_1^t(\times \mathcal{X}_0^{t-1}) \rightarrow \mathcal{U}_{2}
        \end{align*}
        \item causal-encoding causal-decoding with or without channel feedback, with or without source feedforward:
        \begin{align*}
            &f^{(t)}_{U_{1,t}|X_0^t(,Y_1^{t-1})}: \mathcal{X}_0^t (\times \mathcal{Y}_1^{t-1})\longrightarrow \mathcal{U}_{1,t},\\
            &g^{(t)}_{U_{2,t}|Y_1^t(,X_0^{t-1})}: \mathcal{Y}_1^t(\times \mathcal{X}_0^{t-1} )\longrightarrow \mathcal{U}_{2,t}.
        \end{align*}
    \end{enumerate}
\end{corollary}
The proof of this corollary is in Appendix \ref{app: proof of unified causal Gaussian cost}. This result is remarkable, since although these architectures have distinct fundamental achievable cost regions, these regions coincide when all RVs are jointly Gaussian. Hence, any improvement beyond $S_{\mathsf {G}}(P)$ must rely on non-Gaussian constructions, such as our hybrid discrete-continuous ZEC strategies.}

\section{Numerical Results}\label{sec: numerical results}

In this section, we examine the performance of the proposed ZEC-2, ZEC-$k$, ZEC-f coordination coding \mz{strategies} and we compare them with other \mz{relevant} control strategies.

First, we revisit the coding scheme proposed in \cite{Grover2010Witsenhausen} when \textit{both} encoder and decoder are \textit{noncausal}. Given a block-length $n\in\mathbb{N}^{\star}$, a control design is given by a tuple of stochastic function $c = (f,g)$ defined by
\begin{align}
    f:\mathcal{X}_0^n \longrightarrow \mathcal{U}_1^n,\quad g:\mathcal{Y}_1^n \longrightarrow \mathcal{U}_2^n.
\end{align}

This model is superior to our model described in Section~\ref{sec: system model} since the encoder is non-causal.

Besides their DPC-based scheme introduced in \cite[App. D.1-D.7]{Grover2010Witsenhausen}, they further extend it to a combination between the linear scheme and the DPC-based scheme \cite[App. D.8]{Grover2010Witsenhausen} summarized as follows: 

\vspace{0.2cm}

 \begin{figure}[t]
        \centering

\definecolor{airforceblue}{rgb}{0.000, 0.447, 0.741}   
\definecolor{antiquebrass}{rgb}{0.850, 0.325, 0.098}   
\definecolor{alizarin}{rgb}{0.929, 0.694, 0.125}       
\definecolor{amethyst}{rgb}{0.0, 0.55, 0.3}        
\definecolor{leafgreen}{rgb}{0.466, 0.674, 0.188}      
\definecolor{skyblue}{rgb}{0.301, 0.745, 0.933}        
\definecolor{crimsonred}{rgb}{0.635, 0.078, 0.184}     
\definecolor{darkgray}{rgb}{0.333, 0.333, 0.333}       
\definecolor{bronze}{rgb}{0.792, 0.569, 0.212}         
\definecolor{teal}{rgb}{0.494,0.184,0.556}          


\begin{tikzpicture}[scale=1.02]
\begin{axis}[
    xlabel={$N$},
    ylabel={$P$},
    legend style={at={(0.05,1.45)}, anchor=north west},
    axis lines=middle,
    axis line style={black, line width=0.8pt},
    xmin=0, xmax=0.83,
    ymin=-0.0005, ymax=1.1,
    xtick={0.062, 0.65, 0.8},
    xticklabels={0.062, 0.65, 0.8},
    ytick={0, 0.12, 0.19, 0.363, 1},
    scaled ticks=false,
    xlabel style={at={(ticklabel* cs:1)}, anchor=west},
    ylabel style={at={(ticklabel* cs:1)}, anchor=south east},
    axis x line=bottom,
    axis y line=left,
]

\addplot[skyblue, line width=1.2pt, mark=diamond*, mark size=1.8pt, domain=0:0.825] {1};
\addlegendentry{$P_{\mathsf{scalar}}^*$}

\addplot[ airforceblue, line width=1.2pt, mark=o, mark size=1.6pt, mark repeat = 2] 
  table [col sep=comma, x index=0, y index=1] {data/ZEC2NvsP.csv};
\addlegendentry{$P^*_{\mathsf{ZEC}\text{-}2}$}

\addplot[amethyst, line width=1.3pt, mark=square*, mark size=1.5pt, mark repeat = 2] 
  table [col sep=comma, x index=0, y index=1] {data/ZEC3NvsP.csv};
\addlegendentry{$P^*_{\mathsf{ZEC}\text{-}3}$}

\addplot[alizarin, line width=1.2pt, mark=triangle*, mark size=1.6pt, mark repeat = 2] 
  table [col sep=comma, x index=0, y index=1] {data/ZEC4NvsP.csv};
\addlegendentry{$P^*_{\mathsf{ZEC}\text{-}4}$}

\addplot[antiquebrass, line width=1.2pt, mark=diamond*, mark size=1.8pt, mark repeat = 2] 
  table [col sep=comma, x index=0, y index=1] {data/ZECfNvsP.csv};
\addlegendentry{$P^*_{\mathsf{ZEC}\text{-}\mathsf{f}}$}

\addplot[teal, line width=1.2pt, mark=o, mark size=1.7pt, mark repeat = 12] 
  table [col sep=comma, x index=0, y index=1] {data/dpcNvsP.csv};
\addlegendentry{$P^*_{\mathsf{lin+dpc}}$}



\end{axis}
\end{tikzpicture}

        \caption{Minimum required power cost $P^*$ as a function of the noise variance $N$ for the scalar-valued strategies, noncausal strategy that combines linear and DPC scheme, and our proposed ZEC-2, ZEC-3, ZEC-4, and ZEC-f \mz{strategies} when $Q=1$. 
        }
        \label{fig:NvsP}
    \end{figure}

Given $-\sqrt{\frac{P}{Q}}\leq \beta\leq \sqrt{\frac{P}{Q}}$, the power budget $P$ is divided into a linear part $U_{1,1} = -\beta X_0$ and the other part $U_{1,2}$, which is used to implement the DPC-based method against the state $(1-\beta) X_0\sim\mathcal{N}(0, (1-\beta)^2Q)$ with power constraint $\mathbb E[U_{1,2}^2]\leq P-\beta^2Q$. By taking $\beta = -\rho\sqrt{\frac{P}{Q}}$, we obtain the following correlation matrix of $(X_0,U_1)$
\begin{align*}
    K = \begin{pmatrix}
        Q & \rho\sqrt{PQ} \\ \rho\sqrt{PQ} &P
    \end{pmatrix}
\end{align*}
and the state-dependent channel given by
\begin{align*}
    &\Tilde{S} = \frac{\sqrt{Q} + \rho\sqrt{P}}{\sqrt{Q}}\cdot X_0,\\
    &\Tilde{X} \indep (\Tilde{S}, X_0),\text{ with }\mathbb E[\Tilde{X}^2]\leq P(1-\rho^2),\\
    &X_1  = X_0+U_1=\Tilde{X}+\Tilde{S},\\
    &Y_1 = \Tilde{X}+\Tilde{S}+Z.
\end{align*}
\begin{figure*}[t]
\begin{align}
    S_{\mathsf{lin+dpc}}(P) = \min_{\rho\in[-1,1]} \frac{N\brackets{P(1-\rho^2)\sqrt{P+Q+2\rho\sqrt{PQ} +N} - N(\sqrt{Q}+\rho\sqrt{P})}^2}{(P(1-\rho^2)+N)^2(P+Q+2\rho \sqrt{PQ} + N)}\label{eq: S lin+dpc}
\end{align}
\noindent\rule{\textwidth}{0.4pt}
\vspace{-0.5em}
\end{figure*}
This leads to the estimation cost function \eqref{eq: S lin+dpc}. In particular, the minimum power required for enabling zero estimation cost denoted by $P^*_{\mathsf{lin+dpc}}$ for this scheme is given by the unique positive root of equation
\begin{align}
    P^2(P+Q+N) = QN^2\label{eq: P^* dpc}.
\end{align}

This scheme has also been evaluated in \cite{Treust2024power} and is shown to outperform all other estimation costs, including their newly-proposed noncausal-causal coordination coding scheme in \cite[Sec. IV-C.]{Treust2024power}.


\subsection{Comparison of Minimum Power Cost for Zero Estimation Cost Performance}


In this section, for fixed $Q=1$, we compare the minimum power cost \mz{required} for achieving zero estimation cost, \mz{i.e., problem \eqref{eq: power min problem}}, given different noise levels $N\geq 0$ for the following strategies:
\begin{itemize}
        \item \mz{Scalar} control strategies: such as the \mz{optimal} linear policy \eqref{eq: opt linear cost}, the optimal joint Gaussian strategy \eqref{opt gaussian cost}, and the $k$-point strategy \eqref{eq: P_n, n-point} and \eqref{eq: S_n, n-point}. As summarized in Prop. \ref{prop: single-shot zero est}, all of the above-mentioned strategies achieve zero estimation cost with the minimum required power budget, denoted by $P_{\mathsf{scalar}}^* = Q = 1$.
        \item ZEC-$k$ \mz{strategies} given in Theorem \ref{thm: zec} and Theorem \ref{thm: zec-n}, denoted by $P^*_{\mathsf{ZEC}\text{-}k}$ for $k=2,3,4$.
        \item ZEC-f \mz{strategy} given in Theorem \ref{thm: ZEC feedback cost}, denoted by $P^*_{\mathsf{ZEC}\text{-}\mathsf{f}}$.
        \item The combination of linear and DPC-based scheme when both DMs are noncausal recalled earlier, where the power cost required for achieving zero estimation cost is given by $P^*_{\mathsf{lin+dpc}}$, i.e., the unique positive root of \eqref{eq: P^* dpc}.
\end{itemize}

As we can see from Figure \ref{fig:NvsP}, when $N$ increases from zero, the minimum required power $P^*_{\mathsf{ZEC}\text{-}k}$ grows from the corresponding values of the \mz{scalar-valued} $k$-point strategies defined in \eqref{eq: P_k^min}. Specifically, with our chosen parameter,
\begin{align}
P_2^{\min}= 0.363,\quad
    P_3^{\min}= 0.19,\quad
    P_4^{\min}= 0.12.\label{eq: P* n-point stratege}
\end{align}
This indicates that when noise level increases from zero, the power $V_1$ for the auxiliary RV $W_1$ also increases from zero, and zero-cost estimation can be achieved without requiring additional power beyond that of the original $k$-point strategy. However, when $N\geq 0.65$, $P_{\mathsf{ZEC}\text{-}\mathsf{2}}^*\geq Q=1$, meaning that the ZEC-2 \mz{strategy} can no longer offer a block-coding gain in the high-noise regimes.

Moving from ZEC-2 to ZEC-3 yields substantial improvements across both low- and high-noise regimes, and the power cost saturates as $N$ continuously increases. This highlights the critical role of including a zero quantization level, i.e., $Q_3(x)=0$ for $-B_2 \leq x \leq B_2$, see again Figure \ref{fig: 3, 4 point function}.

The transition from ZEC-3 to ZEC-4 provides further improvement only at small noise levels. Interestingly, for $N\geq 0.24$, the ZEC-4 \mz{strategy} effectively reduces to ZEC-3, since the optimizer becomes $a_1=0$. This occurs because at higher noise levels, using fewer quantization levels simplifies the decoder’s task of distinguishing signals, while still ensuring zero estimation cost.

Finally, the ZEC-f \mz{strategy} with feedback strictly outperforms all no-feedback \mz{strategies}, confirming the observation in Remark \ref{rmk: feedback enlarges cost region} that feedback enlarges the achievable cost region without feedback. Remarkably, it even effectively surpasses the noncausal linear and DPC-based scheme. Furthermore, as explained in Corollary \ref{cor: ZECZEP}, it achieves the extreme operating point of zero power and zero estimation cost whenever $N\leq \frac{Q}{2\pi e Q-1}\approx 0.062 $, and for larger $N$ it smoothly converges to a limiting power cost of $0.17$. 


\subsection{Optimized Parameters}

\begin{table*}[t]\centering
\renewcommand{\arraystretch}{1.4}
\setlength{\tabcolsep}{7.4pt}
\rowcolors{2}{gray!10}{white}
\scriptsize
\begin{tabular}{c  ccc  cccc  ccccc  ccc}
\toprule
\multirow{2}{*}{$N$} & \multicolumn{3}{c}{ZEC-2} & \multicolumn{4}{c}{ZEC-3} & \multicolumn{5}{c}{ZEC-4} & \multicolumn{3}{c}{ZEC-f} \\
\cmidrule(lr){2-4}\cmidrule(lr){5-8}\cmidrule(lr){9-13}\cmidrule(lr){14-16}
& $P^*$ & $V_1^*$ & $a^*$ & $P^*$ & $V_1^*$ & $a_2^*$ & $B_2^*$ & $P^*$ & $V_1^*$ & $a_1^*$ & $a_2^*$ & $B_2^*$ & $P^*$ & $a^*$ & $b^*$ \\
\midrule
0.1 & 0.367 & 0.003 & 0.814 & 0.216 & 0.017 & 1.389 & 0.696 & 0.184 & 0.027 & 0.660 & 2.040 & 1.359 & 0.027 & 0.263 & -0.247 \\
0.2 & 0.410 & 0.033 & 0.914 & 0.338 & 0.057 & 1.812 & 0.935 & 0.335 & 0.060 & 0.883 & 2.824 & 1.904 & 0.088 & 0.444 & -0.484 \\
0.3 & 0.501 & 0.081 & 1.035 & 0.468 & 0.084 & 2.150 & 1.139 & 0.468 & 0.084 & 0.000 & 2.150 & 1.139 & 0.115 & 0.496 & -0.556 \\
0.4 & 0.623 & 0.138 & 1.147 & 0.579 & 0.102 & 2.425 & 1.308 & 0.579 & 0.102 & 0.000 & 2.425 & 1.308 & 0.131 & 0.533 & -0.593 \\
0.5 & 0.768 & 0.201 & 1.249 & 0.670 & 0.111 & 2.661 & 1.455 & 0.670 & 0.111 & 0.000 & 2.661 & 1.455 & 0.142 & 0.548 & -0.617 \\
0.6 & 0.927 & 0.268 & 1.342 & 0.743 & 0.115 & 2.874 & 1.588 & 0.743 & 0.115 & 0.000 & 2.874 & 1.588 & 0.151 & 0.560 & -0.639 \\
0.7 & 1.099 & 0.339 & 1.428 & 0.801 & 0.113 & 3.073 & 1.712 & 0.801 & 0.113 & 0.000 & 3.073 & 1.712 & 0.158 & 0.568 & -0.654 \\
0.8 & 1.281 & 0.413 & 1.508 & 0.848 & 0.108 & 3.264 & 1.830 & 0.848 & 0.108 & 0.000 & 3.264 & 1.830 & 0.163 & 0.576 & -0.665 \\
\bottomrule
\end{tabular}
\vspace{4pt}
\caption{$P^*$ and corresponding optimized parameters for ZEC-2, ZEC-3, ZEC-4, and ZEC-f \mz{strategies} across different $N$.}
\label{tab:zec-par}
\renewcommand{\arraystretch}{1.0}
\vspace{-4pt}
\end{table*}

Table \ref{tab:zec-par} records the minimum required power cost $P^*$ and the corresponding optimized parameters $V_1^*, (\mathbf{a}_1^m)^*, (\mathbf{B}_1^m)^*$ for the shown ZEC-$k$ and ZEC-f \mz{strategies} across different noise levels. To better illustrate how these parameters evolve with the noise level, we also plot the curves of the optimized values $V_1^*$ and $a^*$ for ZEC-2, and $a^*$ and $b^*$ for ZEC-f in Figure \ref{fig:opt_data_ZEC2}.

\begin{figure}[t]
    \centering
    \resizebox{0.45\textwidth}{!}{
\begin{tikzpicture}
    \begin{groupplot}[
        group style={
            group size=2 by 2,
            horizontal sep=1.5cm,
            vertical sep=2cm,
        },
        width=5cm,
        height=5cm,
        legend pos=north east,
        axis lines=middle,
        axis line style={black, line width=1pt},
        ytick=\empty,
        xlabel style={at={(axis description cs:1,0)}, anchor=west},
        ylabel style={at={(axis description cs:0,1)}, anchor=south}
    ]

\nextgroupplot[
    xlabel={$N$},
    ylabel={$\log V_1^*$},
    xmin = 0, xmax = 0.81,
    ymin=-13, ymax=1,
    xtick={0,0.8},
    xticklabels={0,0.8},
    ytick={-13,0},
    yticklabels={-13,0},
    axis x line=bottom,
    axis y line=left,
    mark size=1pt
]
    \addplot+[black, line width=1pt, mark=o, mark size=1.3pt,mark repeat=2]
      table [x index=0, y index=1, col sep=comma] {data/logV1vsN.csv};

    \nextgroupplot[
        xlabel={$N$},
        ylabel={$a^*$},
        xmin = 0, xmax = 0.81,
        ymin=0.7, ymax=1.75,
        xtick={0,0.8},
        xticklabels={0,0.8},
        axis x line=bottom,
        axis y line=left,
        ytick={0.7, 1.7},
        yticklabels={0.7, 1.7},
        mark size=1pt
    ]
    \addplot+[black, line width=1pt, mark=o, mark size=1.3pt,mark repeat=2]
      table [x index=0, y index=1, col sep=comma] {data/avsN.csv};

    \nextgroupplot[
        xlabel={$N$},
        ylabel={$a*$},
        ymin=0.0, ymax=0.65,
        xmin=0, xmax=0.81,
        axis x line=bottom,
        axis y line=left,
        xtick={0,0.8},
        xticklabels={0,0.8},
        ytick={0, 0.6},
        yticklabels={0, 0.6},
        mark size=1pt
    ]
    \addplot+[black, line width=1pt, mark=o, mark size=1.3pt, mark repeat=3]
      table [x index=0, y index=1, col sep=comma] {data/ZECfNvsa.csv};

    \nextgroupplot[
        xlabel={$N$},
        ylabel={$b*$},
       xmin = 0, xmax = 0.81,
    ymin=-0.75, ymax=0.1,
    xtick={0,0.8},
    xticklabels={0,0.8},
    ytick={-0.7,0},
    yticklabels={-0.7,0},
    axis x line=bottom,
    axis y line=left,
    mark size=1pt
    ]
    \addplot+[black, line width=1pt, mark=o, mark size=1.3pt, mark repeat=2]
      table [x index=0, y index=1, col sep=comma] {data/ZECfNvsb.csv};

    \end{groupplot}
\end{tikzpicture}
}
    \caption{Evolution of optimized parameters for the ZEC-2 \mz{strategy} (top row) and the ZEC-f \mz{strategy} (bottom row). Top: $\log V_1^*$ versus noise $N$ (left), and $a^*$ versus $N$ (right). Bottom: $a^*$ versus $N$ (left), and $b^*$ versus $N$ (right).}
    \label{fig:opt_data_ZEC2}
\end{figure}

The results show that for ZEC-2 and ZEC-3, the optimized signal levels $V_1^*, a_i^*, B_i^*$ generally increase with $N$, reflecting the higher power required to maintain zero estimation cost under noisier conditions. On the other hand, as also observed in Figure \ref{fig:NvsP}, when $N \geq 0.24$, the ZEC-4 \mz{strategy} effectively collapses to ZEC-3, with $a_1^*=0$, in order to \mz{maintain} the sanity of zero-cost estimation at the decoder.

As for ZEC-f, we notice that the optimized $V_1^*$ remains small but nonzero across all noise levels. This is because the power cost to be minimized \eqref{eq: power ZEC-f} penalizes $V_1$ linearly while a vanishingly small $V_1$ suffices to satisfy the information constraint \eqref{eq: info constraint ZEC-f}. This means, the random codebook induced by $W_1$ becomes nearly deterministic, and we are close to the boundary of the achievable cost region characterized in Theorem~\ref{theorem: c-n w-f}.
Therefore, here, only the optimized parameters $a^*, b^*$ are recorded. As we can see, both $a^*$ and $|b^*|$ grow steadily with $N$, implying that the slope of the piecewise linear mapping $X_1 =  a\cdot\text{sign}(X_0) + (b+1)\cdot X_0$ (with $-1<b<0$) decreases, thereby converging toward an effective two-point quantization scheme with slope zero. This trend is consistent with the analysis of the role of slope uncovered in \cite{zhao2025lowpower}, which shows that higher power cost leads to flatter (less sloped) piecewise linear policies, making the estimation at the second DM more efficient.

Moreover, our simulations show that if we manually enforce $b=-1$ in the ZEC-f formulation \eqref{eq: ZEC feedback scheme}, i.e., we eliminate the continuous source state $X_0$ from $X_1$ as we did in the no-feedback ZEC schemes, the resulting performance ($P^*,V_1^*,a^*$) exactly coincides with that of the ZEC-2 \mz{strategy} without feedback. This observation further underscores the importance of feedback, as it enables direct communication of the original source signal $X_0$ to the decoder, which allows the minimum required power cost to be significantly reduced.

\subsection{Power-Estimation Performance Comparison}

We now compare the performance of the proposed ZEC \mz{strategies} with classical \mz{scalar-valued} strategies in the power–estimation cost space.

In order to evaluate the \mz{scalar-valued} $k$-point strategy given in Theorem \ref{thm: n-point strategy}, we normalize the original cost function \eqref{eq: scalar obj function} by dividing it by $\lambda^2+1$. This is equivalent to a weighted optimization problem parameterized by $\omega \triangleq \frac{\lambda^2}{\lambda^2+1}\in[0,1]$, each of which corresponds to a supporting hyperplane in a two-dimensional power-estimation cost space. Specifically, for a chosen $\lambda^2\in[0,\infty)$, a quantization level $k\in\mathbb N^{\star}$ and $m=\ceil{\frac{k}{2}}$, we determine the optimal parameters $\mathbf{a}_1^m, \mathbf{B}_1^m$ satisfying Definition \ref{def: n-point quantization}, such that
\begin{align*}
    (\mathbf{a}^*_{\omega}, \mathbf{B}^*_{\omega})& =  \argmin\frac{1}{\lambda^2+1}\sbrackets{ \lambda^2 P_k(\mathbf{a}_1^m,\mathbf{B}_1^m) + S_k(\mathbf{a}_1^m,\mathbf{B}_1^m) }\\
    &=
    \argmin\sbrackets{ \omega P_k(\mathbf{a}_1^m,\mathbf{B}_1^m) + (1-\omega) S_k(\mathbf{a}_1^m,\mathbf{B}_1^m) },
\end{align*}
which gives us the best configuration for the given parameter $\lambda^2$. Then, we plot the curve $\brackets{P_k(\mathbf{a}^*_\omega, \mathbf{B}^*_\omega), S_k(\mathbf{a}^*_\omega, \mathbf{B}^*_\omega)}\in\mathbb R^2_+$ for all $\omega=\frac{\lambda^2}{\lambda^2+1}\in[0,1]$, which gives us the optimized estimation cost value at each specific power consumption. For simplicity, from now on we denote this optimized power-estimation cost function for the $k$-point strategy by $S_k(P)$.

\begin{figure}[h]
        \centering

\definecolor{airforceblue}{rgb}{0.000, 0.447, 0.741}   
\definecolor{antiquebrass}{rgb}{0.850, 0.325, 0.098}   
\definecolor{alizarin}{rgb}{0.929, 0.694, 0.125}       
\definecolor{amethyst}{rgb}{0.0, 0.55, 0.3}        
\definecolor{leafgreen}{rgb}{0.466, 0.674, 0.188}      
\definecolor{skyblue}{rgb}{0.301, 0.745, 0.933}        
\definecolor{crimsonred}{rgb}{0.635, 0.078, 0.184}     
\definecolor{darkgray}{rgb}{0.333, 0.333, 0.333}       
\definecolor{bronze}{rgb}{0.792, 0.569, 0.212}         
\definecolor{teal}{rgb}{0.494,0.184,0.556}          


\begin{tikzpicture}[scale=1]
\begin{axis}[
    xlabel={$P$},
    ylabel={$\mathsf{MMSE}$},
    legend style={at={(1,1.1)}, anchor=north east, font=\small},
    axis lines=middle,
    axis line style={black, line width=1pt},
    xmin=0, xmax=1.02,
    ymin=0, ymax=0.132,
    axis x line=bottom,
    axis y line=left,
    xticklabel style={font=\scriptsize},
    xtick={0.065,0.13,  0.26, 0.27, 0.38, 1},
    xticklabels={
        {\scriptsize\hspace*{-0.5em} 0.065},
        {\scriptsize\hspace*{0.5em} 0.13},
        {\scriptsize\hspace*{-2.0em}0.26},
        {\scriptsize\hspace*{1.5em}0.27},
        {\scriptsize0.38},
        {\normalsize 1}
    },
    ytick={0},
    xlabel style={at={(ticklabel* cs:1.01)}, anchor=north west},
    ylabel style={at={(ticklabel* cs:1)}, anchor=south east},
]

\addplot[teal, line width=1.2pt, mark=o, mark size=1.7pt, mark repeat = 8] 
  table [col sep=comma, x index=0, y index=1] {data/dpc_015.csv};
\addlegendentry{$S_{\mathsf{lin+dpc}}(P)$}

\addplot[alizarin, line width=1pt, mark=diamond*, mark size=2pt,mark repeat=1] 
    table [col sep=comma, x index=0, y index=1] {data/ZEC4_015.csv};
\addlegendentry{$S_{\mathsf{ZEC}\text{-}4}(P)$}

\addplot[amethyst, line width=1pt, mark=diamond*, mark size=1.7pt,mark repeat=1] 
    table [col sep=comma, x index=0, y index=1] {data/ZEC3_015.csv};
\addlegendentry{$S_{\mathsf{ZEC}\text{-}3}(P)$}

\addplot[airforceblue, line width=1pt, solid, mark=diamond*, mark size=1.4pt,mark repeat=1]
    table [col sep=comma, x index=0, y index=1] {data/ZEC2_015.csv};
\addlegendentry{$S_{\mathsf{ZEC}\text{-}2}(P)$}

\addplot[alizarin, line width=1pt, solid, mark=o, mark size=1.6pt,mark repeat=5]
    table [col sep=comma, x index=0, y index=1] {data/4P_015.csv};
\addlegendentry{$S_4(P)$}

\addplot[
    amethyst, 
    line width=1pt, 
    solid, 
    mark=o, 
    mark size=1.6pt, 
    mark repeat=10 
] 
table [col sep=comma, x index=0, y index=1] {data/3P_015.csv};
\addlegendentry{$S_3(P)$}

\addplot[airforceblue, line width=1pt, solid, mark=o, mark size=1.6pt,mark repeat=40]
    table [col sep=comma, x index=0, y index=1] {data/2Popt_015.csv};
\addlegendentry{$S_2(P)$}


\addplot[skyblue, line width=1pt, mark=triangle*, mark size=1.7pt] 
    table [col sep=comma, x index=0, y index=1] {data/Sl_015.csv};
\addlegendentry{$S_{\ell}(P)$}

\addplot[leafgreen, line width=1pt, solid, mark=triangle*, mark size=1.7pt]
    table [col sep=comma, x index=0, y index=1] {data/SG_015.csv};
\addlegendentry{$S_{\mathsf{G}}(P)$}



\addplot[only marks, mark=x, mark size=3.5pt, line width=2pt, color=antiquebrass] coordinates {
     (0.065,0)
};
\addlegendentry{$S_{\mathsf{ZEC}\text{-}\mathsf{f}}(P)$}

\addplot[only marks, mark=|, mark size=3pt, line width = 0.7pt, color=black] coordinates {
   (0.1325,0)  
   (0.255,0) (0.272,0)(0.38,0)
};

\end{axis}
\end{tikzpicture}

        \caption{Performance comparison at $Q=1$ and $N=0.15$ between ZEC-$k$ \mz{strategies} ($S_{\mathsf{ZEC}\text{-}k}(P)$), the corresponding $k$-point strategies ($S_k(P)$) for $k=2,3,4$, the best linear strategy ($S_\ell(P)$), and the optimal joint Gaussian strategy ($S_G(P)$). The ZEC-2, ZEC-3, and ZEC-4 \mz{strategies} achieve zero-estimation-cost reconstruction at $P_{\mathsf{ZEC}\text{-}2}^*=0.38$, $P_{\mathsf{ZEC}\text{-}3}^*=0.27$, and $P_{\mathsf{ZEC}\text{-}4}^*=0.26$, respectively, while the feedback-enabled ZEC-f \mz{strategy} requires only $P_{\mathsf{ZEC}\text{-}f}^*=0.065$.}
        \label{fig:Q=1,N=0.15}
    \end{figure}

Now, we compare the cost functions of ZEC-2, ZEC-3, ZEC-4, and ZEC-f, denoted $S_{\mathsf{ZEC}\text{-}2}(P)$, $S_{\mathsf{ZEC}\text{-}3}(P)$, $S_{\mathsf{ZEC}\text{-}4}(P)$, and $S_{\mathsf{ZEC}\text{-}f}(P)$\footnote{For brevity, in the figures we plot only the minimum required power for ZEC-f.}, respectively, with those of the \mz{scalar-valued} 2-, 3-, and 4-point strategies $S_2(P), S_3(P), S_4(P)$ (after the weighted-optimization procedure mentioned-above), the best affine strategy $S_\ell(P)$ \eqref{eq: opt linear cost}, the optimal joint Gaussian strategy $S_\mathsf{G}(P)$ \eqref{opt gaussian cost}, and the noncausal strategy combining linear and DPC-based schemes \eqref{eq: S lin+dpc}. All comparisons are carried out at a fixed $Q=1$ and three different noise levels $N \in\{ 0.15, 0.3, 0.06\}$, as shown in Figures \ref{fig:Q=1,N=0.15}–\ref{fig:Q=1,N=0.06}.


As illustrated in Figure \ref{fig:Q=1,N=0.15}, when $N=0.15$, the ZEC-$k$ \mz{strategies} achieve zero-estimation-cost reconstruction with only a modest increase in power compared to their \mz{scalar-valued} counterparts $P_k^{\min}$. Specifically, while the minimum power requirements for the \mz{scalar-valued} strategies are given in \eqref{eq: P* n-point stratege}, the ZEC-2, ZEC-3, and ZEC-4 \mz{strategies} achieve zero-cost estimation once the power budget exceeds $P_{\mathsf{ZEC}\text{-}2}^*=0.38$, $P_{\mathsf{ZEC}\text{-}3}^*=0.27$, and $P_{\mathsf{ZEC}\text{-}4}^*=0.26$, respectively. Notably, the ZEC-f \mz{strategy} achieves zero-cost estimation at a significantly lower threshold of $P_{\mathsf{ZEC}\text{-}f}^*=0.065$ which is even lower than that of the noncausal linear$+$DPC-based scheme $P^*_{\mathsf{lin+dpc}}=0.12$. By contrast, all other \mz{scalar-valued} strategies can achieve zero-estimation-cost reconstruction only at the much larger power level $P=Q=1$.

 \begin{figure}[t]
        \centering

\definecolor{airforceblue}{rgb}{0.000, 0.447, 0.741}   
\definecolor{antiquebrass}{rgb}{0.850, 0.325, 0.098}   
\definecolor{alizarin}{rgb}{0.929, 0.694, 0.125}       
\definecolor{amethyst}{rgb}{0.0, 0.55, 0.3}         
\definecolor{leafgreen}{rgb}{0.466, 0.674, 0.188}      
\definecolor{skyblue}{rgb}{0.301, 0.745, 0.933}        
\definecolor{crimsonred}{rgb}{0.635, 0.078, 0.184}     
\definecolor{darkgray}{rgb}{0.333, 0.333, 0.333}       
\definecolor{bronze}{rgb}{0.792, 0.569, 0.212}         
\definecolor{teal}{rgb}{0.494,0.184,0.556}           


\begin{tikzpicture}[scale=1]
\begin{axis}[
    xlabel={$P$},
    ylabel={$\mathsf{MMSE}$},
    legend style={at={(1,1.1)}, anchor=north east, font=\small},
    axis lines=middle,
    axis line style={black, line width=1pt},
    xmin=0, xmax=1.02,
    ymin=0, ymax=0.25,
    axis x line=bottom,
    axis y line=left,
    xticklabel style={font=\scriptsize},
    xtick={0.12,0.24, 0.46, 0.5,  1},
    xticklabels={
        {\scriptsize 0.12},
        {\scriptsize 0.24},
        {\scriptsize \hspace*{-1.0em}0.46},
        {\scriptsize \hspace*{0.5em}0.5},
        {\normalsize 1}
    },
    ytick={0},
    xlabel style={at={(ticklabel* cs:1.01)}, anchor=north west},
    ylabel style={at={(ticklabel* cs:1)}, anchor=south east},
]

\addplot[teal, line width=1.2pt, mark=o, mark size=1.7pt, mark repeat = 8] 
  table [col sep=comma, x index=0, y index=1] {data/dpc_03.csv};
\addlegendentry{$S_{\mathsf{lin+dpc}}(P)$}

\addplot[alizarin, line width=1pt, mark=diamond*, mark size=1.7pt,mark repeat=1] 
    table [col sep=comma, x index=0, y index=1] {data/ZEC4_03.csv};
\addlegendentry{$S_{\mathsf{ZEC}\text{-}3}(P),S_{\mathsf{ZEC}\text{-}4}(P)$}

\addplot[airforceblue, line width=1pt, solid, mark=diamond*, mark size=1.3pt,mark repeat=1]
    table [col sep=comma, x index=0, y index=1] {data/ZEC2_03.csv};
\addlegendentry{$S_{\mathsf{ZEC}\text{-}2}(P)$}

\addplot[alizarin, line width=1pt, solid, mark=o, mark size=1.6pt,mark repeat=14]
    table [col sep=comma, x index=0, y index=1] {data/4P_0300.csv};
\addlegendentry{$S_4(P)$}

\addplot[
    amethyst, 
    line width=1pt, 
    solid, 
    mark=o, 
    mark size=1.6pt, 
    mark repeat=15,
    mark phase=1 
] 
table [col sep=comma, x index=0, y index=1] {data/3P_0300.csv};
\addlegendentry{$S_3(P)$}

\addplot[airforceblue, line width=1pt, solid, mark=o, mark size=1.6pt,mark repeat=50 ]
    table [col sep=comma, x index=0, y index=1] {data/2Popt_03.csv};
\addlegendentry{$S_2(P)$}



\addplot[skyblue, line width=1pt, solid, mark=triangle*, mark size=1.7pt, mark repeat=23 ]
    table [col sep=comma, x index=0, y index=1] {data/SG_03.csv};
\addlegendentry{$S_{\mathsf{G}}(P),S_{\ell}(P)$}


\addplot[only marks, mark=x, mark size=3.5pt, line width=2pt, color=antiquebrass] coordinates {
     (0.12,0)
};
\addlegendentry{$S_{\mathsf{ZEC}\text{-}\mathsf{f}}(P)$}

\addplot[only marks, mark=|, mark size=3pt, line width = 0.7pt, color=black] coordinates {
    (0.242,0)
    (0.46,0) (0.5,0)
};

\end{axis}
\end{tikzpicture}

        \caption{Comparison of cost functions $S_{\mathsf{ZEC}\text{-}k}(P)$, $S_k(P)$, $S_{\mathsf{ZEC}\text{-}\mathsf{f}}(P)$, and $S_{\mathsf{G}}(P)$ for $k=2,3,4$ at $Q=1$ and $N=0.3$. In this regime, the two-point strategy is outperformed by the optimal joint Gaussian (and equivalently, linear) strategy. Nevertheless, the ZEC-2 \mz{strategy} achieves zero-estimation-cost performance once the power exceeds $P_{\mathsf{ZEC}\text{-}2}^*=0.5$. The ZEC-4 \mz{strategy} collapses to ZEC-3, which attains zero-cost estimation at $P_{\mathsf{ZEC}\text{-}3}^*=0.46$. The feedback-enabled ZEC-f \mz{strategy} requires only $P_{\mathsf{ZEC}\text{-}\mathsf{f}}^*=0.12$ to achieve zero-estimation cost.}
        \label{fig:Q=1,N=0.3}
    \end{figure}


In Figure~\ref{fig:Q=1,N=0.3}, when $N=0.3$, the two-point strategy no longer outperforms the optimal linear strategy, which in this case coincides with the joint Gaussian solution. Nevertheless, the proposed ZEC \mz{strategies} still maintain a distinct advantage: ZEC-2 achieves zero-estimation-cost reconstruction with $P\geq P_{\mathsf{ZEC}\text{-}2}^*=0.50$, while ZEC-4 reduces to ZEC-3 and reaches zero-cost estimation at $P_{\mathsf{ZEC}\text{-}3}^*=0.47$. Most notably, the feedback-enabled ZEC-f \mz{strategy} requires only $P_{\mathsf{ZEC}\text{-}\mathsf{f}}^*=0.12$, even in this relatively high-noise regime.

 \begin{figure}[h]
        \centering

\definecolor{airforceblue}{rgb}{0.000, 0.447, 0.741}   
\definecolor{antiquebrass}{rgb}{0.850, 0.325, 0.098}   
\definecolor{alizarin}{rgb}{0.929, 0.694, 0.125}       
\definecolor{amethyst}{rgb}{0.0, 0.55, 0.3}        
\definecolor{leafgreen}{rgb}{0.466, 0.674, 0.188}      
\definecolor{skyblue}{rgb}{0.301, 0.745, 0.933}        
\definecolor{crimsonred}{rgb}{0.635, 0.078, 0.184}     
\definecolor{darkgray}{rgb}{0.333, 0.333, 0.333}       
\definecolor{bronze}{rgb}{0.792, 0.569, 0.212}         
\definecolor{teal}{rgb}{0.494,0.184,0.556}           


\begin{tikzpicture}[scale=1]
\begin{axis}[
    xlabel={$P$},
    ylabel={$\mathsf{MMSE}$},
    legend style={at={(1.06,1.2)}, anchor=north east, font=\small},
    axis lines=middle,
    axis line style={black, line width=1pt},
    xmin=0, xmax=1.02,
    ymin=0, ymax=0.06,
    axis x line=bottom,
    axis y line=left,
    xticklabel style={font=\scriptsize},
    xtick={0,0.06, 0.13, 0.19, 0.36, 1},
    xticklabels={
    {\normalsize 0},
    {\scriptsize\hspace*{-0.5em} 0.06},
        {\scriptsize 0.13},
        {\scriptsize\hspace*{1em} 0.19},
        {\scriptsize0.36},
        {\normalsize 1}
    },
    ytick=\empty,
    xlabel style={at={(ticklabel* cs:1.01)}, anchor=north west},
    ylabel style={at={(ticklabel* cs:1)}, anchor=south east},
]

\addplot[teal, line width=1.2pt, mark=o, mark size=1.7pt, mark repeat = 8] 
  table [col sep=comma, x index=0, y index=1] {data/dpc_006.csv};
\addlegendentry{$S_{\mathsf{lin+dpc}}(P)$}

\addplot[alizarin, line width=1pt, mark=diamond*, mark size=2pt,mark repeat=2] 
    table [col sep=comma, x index=0, y index=1] {data/ZEC4_006.csv};
\addlegendentry{$S_{\mathsf{ZEC}\text{-}4}(P)$}

\addplot[amethyst, line width=1pt, mark=diamond*, mark size=1.7pt,mark repeat=2] 
    table [col sep=comma, x index=0, y index=1] {data/ZEC3_006.csv};
\addlegendentry{$S_{\mathsf{ZEC}\text{-}3}(P)$}

\addplot[airforceblue, line width=1pt, solid, mark=diamond*, mark size=1.4pt,mark repeat=2]
    table [col sep=comma, x index=0, y index=1] {data/ZEC2_006.csv};
\addlegendentry{$S_{\mathsf{ZEC}\text{-}2}(P)$}

\addplot[alizarin, line width=1pt, solid, mark=o, mark size=1.6pt,mark repeat=30]
    table [col sep=comma, x index=0, y index=1] {data/4P_006.csv};
\addlegendentry{$S_4(P)$}

\addplot[
    amethyst, 
    line width=1pt, 
    solid, 
    mark=o, 
    mark size=1.6pt, 
    mark repeat=60 
] 
table [col sep=comma, x index=0, y index=1] {data/3P_006.csv};
\addlegendentry{$S_3(P)$}

\addplot[airforceblue, line width=1pt, solid, mark=o, mark size=1.6pt,mark repeat=1000]
    table [col sep=comma, x index=0, y index=1] {data/2Popt_006.csv};
\addlegendentry{$S_2(P)$}


\addplot[skyblue, line width=1pt, mark=triangle*, mark size=1.7pt,mark repeat=5] 
    table [col sep=comma, x index=0, y index=1] {data/Sl_006.csv};
\addlegendentry{$S_{\ell}(P)$}

\addplot[leafgreen, line width=1pt, solid, mark=triangle*, mark size=1.7pt,mark repeat=5]
    table [col sep=comma, x index=0, y index=1] {data/SG_006.csv};
\addlegendentry{$S_{\mathsf{G}}(P)$}


\addplot[only marks, mark=x, mark size=3.5pt, line width=2pt, color=antiquebrass] coordinates {
     (0,0)
};
\addlegendentry{$S_{\mathsf{ZEC}\text{-}\mathsf{f}}(P)$}

\addplot[only marks, mark=|, mark size=3pt, line width = 0.7pt, color=black] coordinates {
    (0.06,0)
    (0.13,0) (0.19,0)(0.36,0)
};

\end{axis}
\end{tikzpicture}

        \caption{Performance comparison at $Q=1$ and $N=0.06$ between ZEC-$k$ \mz{strategies} ($S_{\mathsf{ZEC}\text{-}k}(P)$), the corresponding $k$-point strategies ($S_k(P)$) for $k=2,3,4$, the best linear strategy ($S_\ell(P)$), and the optimal joint Gaussian strategy ($S_G(P)$). In this case, the ZEC-f \mz{strategy} achieves the extreme point of zero-power zero-estimation cost.}
        \label{fig:Q=1,N=0.06}
    \end{figure}

Figure \ref{fig:Q=1,N=0.06} shows that at the low noise level of $N=0.06$, the ZEC-f \mz{strategy} achieves the extreme point of zero-power zero-estimation cost performance.

\section{Conclusion}\label{sec: conclusion}

In this paper, we used coordination coding to characterize the communication constraints underlying cooperation between two DMs in the causal-encoding, noncausal-decoding formulation of Witsenhausen’s counterexample. The result is a single-letter achievable region with auxiliary RVs that capture the dual role of control — state regulation and implicit communication, and an information constraint specifying the required communication rate for reliable coordination. Building on this, we proposed the ZEC-$k$ \mz{strategy}, which significantly lowers the power needed for zero estimation cost as the quantization level increases. With channel feedback, the \mz{strategy} is further enhanced, attaining the extreme point of zero power and zero estimation cost in the low-noise regime.

Taken together, these results provide an insightful answer to the central question motivating this work: how much communication is needed, and what is relevant to be transmitted in Witsenhausen counterexample? Our findings show that when we have sufficient transmissions, effective cooperation between communication and control is established with carefully designed quantization indices together with an independent Gaussian codebook.
These components jointly suffice to regulate the power-controlled state, deterministically define the reconstruction target, and simplify the information constraint. Developed in the canonical benchmark, our framework for Witsenhausen counterexample extends naturally to more general distributed decision-making settings where multiple agents must coordinate actions under asymmetric, decentralized knowledge. This perspective invites deeper integration of coordination-coding and control, and highlights the fundamental role of communication in shaping information flow in decentralized network control.

\appendices
\section{Achievability Proof of Theorem \ref{thm: c-n wits main}}\label{app: ach proof}
Since our achievability proof involves both continuous and discrete RVs, our analysis must incorporate both the Lebesgue measure $\lambda$ (for continuous RVs) and the counting measure $\mu$ (for discrete RVs). The corresponding information-theoretic quantities are defined using the Radon-Nikodym derivative, which generalizes the concept of density with respect to a base measure. This approach enables us to define entropy for mixed discrete-continuous RVs in a way that preserves consistency with both discrete entropy and differential entropy; see \cite{pinsker1964information} and \cite[App. A]{Treust2024power} for further details. Similar to what the authors do in the latter reference, we extend the standard definition of (jointly) weak typicality to random vectors with either discrete or continuous components. We denote by $\mathcal{A}_\varepsilon^{(n)}(\mathcal{P}_{X,Y})$ the set of jointly typical sequences given by
\begin{align}
    \mathcal{A}_\varepsilon^{(n)}(\mathcal{P}_{X,Y})& = \bigg\{ (x^n,y^n)\in\mathbb{R}^{n\times2}:\\
    &\qquad\abs{-\frac{1}{n}\log 
    \prod_{i=1}^n \mathcal{P}_{X,Y}(x_i,y_i)
    -H(X,Y)}<\varepsilon,\nonumber \\
    &\qquad x^n\in\mathcal{A}_\varepsilon^{(n)}(\mathcal{P}_X),y^n\in\mathcal{A}_\varepsilon^{(n)}(\mathcal{P}_Y)\bigg\}.\nonumber
\end{align}
Moreover, the authors in \cite[App. A]{Treust2024power} verify that the joint asymptotic equipartition property (AEP), covering lemma, and packing lemma, see \cite{cover1999elements} \cite{elgamal2011nit}, can be straightforwardly extended to the above definition of typicality.

We consider an arbitrary but fixed $\varepsilon>0$ and assume the sequence $(X_0^n, W_1^n, W_2^n, U_1^n, X_1^n, Y_1^n, U_2^n)$ is generated i.i.d. according to a distribution that decomposes as \eqref{prob result}, with $P = \mathbb E[U_1^2]$, $S=\mathbb E[(X_1 - U_2)^2]$.
Let $\psi^{(n)}: \mathcal{X}_0^n\times \mathcal{W}_1^n\times \mathcal{W}_2^n\times \mathcal{U}_1^n\times\mathcal{X}_1^n\times\mathcal{Y}_1^n\times \mathcal{U}_2^n\rightarrow\{0,1\}$ denote an indicator function for sequences of length $n$ with
    \begin{align}
        &\psi^{(n)}(x_0^n,w_1^n,w_2^n,u_1^n,x_1^n,y_1^n,u_2^n)\nonumber \\
        &= 
        \left\{
          \begin{aligned}
              &1 &\text{if }|c_S(x_1^n,u_2^n)-S|\geq \frac{1}{12}\varepsilon\\
              & &\text{or }(w_1^n,w_2^n,y_1^n)\notin\mathcal{A}_\varepsilon^n(W_1,W_2,Y_1),\\
              &0 \quad &\text{otherwise}.
          \end{aligned}   
        \right.\label{eq: ach proof, indicator function}
    \end{align}
Using the weak law of large numbers (LLN) and the union bound we have
\begin{align*}
    \mathbb E[\psi^{(n)}(X_0^n,W_1^n,W_2^n,U_1^n,X_1^n,Y_1^n,U_2^n)]\leq \delta_n\rightarrow 0,
\end{align*}
as $n\rightarrow\infty$. And similarly, we define the set \begin{align*}
    \mathcal{S}_\varepsilon^{(n)} = \{(x_0^n,w_1^n,w_2^n)\mid\eta^{(n)}(x_0^n,w_1^n,w_2^n)\leq\sqrt{\delta_n}\},
\end{align*}
where $\eta^{(n)}(x_0^n,w_1^n,w_2^n) = \mathbb E[\psi^{(n)}(x_0^n,w_1^n,w_2^n,U_1^n,X_1^n,Y_1^n,\\U_2^n)\mid X_0^n = x_0^n,W_1^n = w_1^n,W_2^n = w_2^n] $. We then take
\begin{align}
    \mathcal{B}_\varepsilon^{(n)} = \mathcal{A}_\varepsilon^{(n)}\cap\mathcal{S}_\varepsilon^{(n)}.\label{eq: typical set B}
\end{align}
Then, we can easily prove the subsequent lemma:
    \begin{lemma}\label{lemma: probable typical set}
        Let the sequence $(X_0^n,W_1^n,W_2^n)$ i.i.d. $\sim \mathcal{P}_{X_0,W_1,W_2}$, then 
        \begin{align*}
        \mathbb{P}\brackets{(X_0^n,W_1^n,W_2^n)\in\mathcal{B}_\varepsilon^{(n)}}\xrightarrow[]{n\rightarrow\infty} 1.
        \end{align*}
    \end{lemma}

\begin{proof}[Proof of Lemma \ref{lemma: probable typical set}]
From the Markov inequality we obtain
\begin{align*} \mathbb{P}\brackets{(X_0^n,W_1^n,W_2^n)\notin\mathcal{S}_\varepsilon^{(n)}}\leq \frac{\delta_n}{\sqrt{\delta_n}} = \sqrt{\delta_n}.
\end{align*}
    \end{proof}
    
Lemma \ref{lemma: probable typical set} elaborates that the certainty of an outcome being jointly weakly typical \textit{and} satisfying the distortion cost constraint is guaranteed, as per the AEP, if it is generated i.i.d.

Moreover, it is also necessary to quantize the output of the non-causal decoder, as done in \cite{Treust2024power}, to ensure that the second cost constraint remains bounded in case an error happens. Given a joint distribution with $\mathbb E[(X_1 - U_2)^2] = S$, for any $\hat{\delta}>0$ there exists a quantization $q_{U_2}: \mathcal{U}_2\rightarrow\{ \hat{u}_{2,k}\}_{k=1}^{N_{U_2}}$, as explained in \cite{wyner1978rate}, such that
\begin{equation}
    \hat{S} = \mathbb E [(X_1 - q_{U_2}(U_2))^2]\leq (1+\hat{\delta}) S,\nonumber
\end{equation}
in particular $\hat{\delta}S<\frac{1}{4}\varepsilon$.

Now, we provide a coding scheme by considering a probability distribution $\mathcal{P}_{X_0, W_1, W_2, U_1, X_1, Y_1, U_2}$ that decomposes as \eqref{prob result} such that there exists a $\delta>0$, and a rate $R>0$, satisfying
\begin{align*}
    I(W_2; X_0|W_1) + \delta\leq R\leq  I(W_1; Y_1) + I(W_2; Y_1|W_1) - \delta.
\end{align*}
We consider a block-Markov code $c\in\mathcal{C}(Bn)$ defined over $B\in\mathbb N^\star$ blocks of length $n$, similar to the scheme in \cite{choudhuri2013causal}. Figure \ref{fig:Markov block codes} is an illustration of this coding scheme.

 \begin{figure}[t]
        \centering

\begin{tikzpicture}[x=0.5cm,y=0.5cm]

\foreach \y in {-8,-6,-4,-2,0,2,4,6,8} {
  \draw[line width=2pt] (0,\y) -- (12,\y);
}

\foreach \x in {2,4,6,8,10} {
  \draw[dashed,line width=0.5pt] (\x,-9) -- (\x,9);
}

\node[anchor=east] at (-0.5,8) {$X_0^n$};
\node[anchor=east] at (-0.5,6) {$W_1^n$};
\node[anchor=east] at (-0.5,4) {$W_2^n$};
\node[anchor=east] at (-0.5,2) {$U_1^n$};
\node[anchor=east] at (-0.5,0) {$X_1^n$}; 
\node[anchor=east] at (-0.5,-2) {$Y_1^n$};
\node[anchor=east] at (-0.5,-4) {$\hat W_1^n$};
\node[anchor=east] at (-0.5,-6) {$\hat W_2^n$};
\node[anchor=east] at (-0.5,-8) {$U_2^n$};

\node at (3,9) {$b-2$};
\node at (5,9) {$b-1$};
\node at (7,9) {$b$};
\node at (9,9) {$b+1$};

\newcommand{\hatchrect}[5]{%
  \path[pattern=north east lines,pattern color=#5,draw=#5]
    (#1,#2) rectangle (#3,#4);
}
\draw[->,line width=1.1pt,red] (5,7.7) -- (5,4.3); 
\draw[<-,line width=1.1pt,red] (6.5,5.7) -- (5.3,4.3); 
\draw[->,line width=1.1pt,red] (7,5.7) -- (7,2.3);
\draw[->,line width=1.1pt,red] (7,1.7) -- (7,0.3); 
\draw[->,line width=1.1pt,red] (7,-0.3) -- (7,-1.7);
\draw[->,line width=1.1pt,red] (7,-2.3) -- (7,-3.7);
\draw[->,line width=1.1pt,red] (7,-4.3) -- (5.6,-5.7);
\draw[->,line width=1.1pt,red] (5,-6.3) -- (5,-7.7);


\newcommand{\greenslashrect}[5]{%
  \path[fill=white,draw=green!70!black,line width=0.8pt] (#1,#2) rectangle (#3,#4); 
  \path[pattern=north east lines,pattern color=#5,draw=none]
       (#1,#2) rectangle (#3,#4); 
}

\newcommand{\redslashrect}[5]{%
  \path[fill=white,draw=red,line width=0.8pt] (#1,#2) rectangle (#3,#4); 
  \path[pattern=north east lines,pattern color=#5,draw=none]
       (#1,#2) rectangle (#3,#4); 
}

\newcommand{\blueslashrect}[5]{%
  \path[fill=white,draw=blue,line width=0.8pt] (#1,#2) rectangle (#3,#4); 
  \path[pattern=north east lines,pattern color=#5,draw=none]
       (#1,#2) rectangle (#3,#4); 
}

\greenslashrect{2}{7.8}{4}{8.2}{green!70!black}  
\greenslashrect{4}{5.8}{6}{6.2}{green!70!black}  
\greenslashrect{2}{3.8}{4}{4.2}{green!70!black}  
\greenslashrect{4}{1.8}{6}{2.2}{green!70!black}  
\greenslashrect{4}{-0.2}{6}{0.2}{green!70!black} 
\greenslashrect{4}{-2.2}{6}{-1.8}{green!70!black}
\greenslashrect{4}{-4.2}{6}{-3.8}{green!70!black}
\greenslashrect{2}{-6.2}{4}{-5.8}{green!70!black}
\greenslashrect{2}{-8.2}{4}{-7.8}{green!70!black}


\redslashrect{4}{7.8}{6}{8.2}{red}            
\redslashrect{6}{5.8}{8}{6.2}{red}            
\redslashrect{4}{3.8}{6}{4.2}{red}            
\redslashrect{6}{1.8}{8}{2.2}{red}            
\redslashrect{6}{-0.2}{8}{0.2}{red}           
\redslashrect{6}{-2.2}{8}{-1.8}{red}          
\redslashrect{6}{-4.2}{8}{-3.8}{red}          
\redslashrect{4}{-6.2}{6}{-5.8}{red}          
\redslashrect{4}{-8.2}{6}{-7.8}{red}          


\blueslashrect{6}{7.8}{8}{8.2}{blue}           
\blueslashrect{8}{5.8}{10}{6.2}{blue}          
\blueslashrect{6}{3.8}{8}{4.2}{blue}           
\blueslashrect{8}{1.8}{10}{2.2}{blue}          
\blueslashrect{8}{-0.2}{10}{0.2}{blue}         
\blueslashrect{8}{-2.2}{10}{-1.8}{blue}        
\blueslashrect{8}{-4.2}{10}{-3.8}{blue}         
\blueslashrect{6}{-6.2}{8}{-5.8}{blue}         
\blueslashrect{6}{-8.2}{8}{-7.8}{blue}         


\end{tikzpicture}
        \caption{Block-Markov codes with causal encoder.}
        \label{fig:Markov block codes}
    \end{figure}

\textit{Random codebook generation}: We generate $|\mathcal{M}| = 2^{nR}$ sequences $W_1^n(m)$ i.i.d. $\sim \mathcal{P}_{W_1}$ with index $m\in\mathcal{M}$. For each index $m\in\mathcal{M}$, we generate the same number $|\mathcal{M}| = 2^{nR}$ sequences $W_2^n(m, \hat{m})$ with index $\hat{m}\in\mathcal{M}$ i.i.d. $\sim \mathcal{P}_{W_2 | W_1}$ depending on sequence $W_1^n(m)$.

\textit{Encoding function: }Due to the causal feature, the encoder performs \textit{forward encoding}, where the encoder outputs the information associated with the \textit{past} block. Let $m_b$ denote the message generated during block $b\in[1:B]$. During the first block, without loss of generality, the encoder takes $m_1=1$ and returns $W_1^n(m_1)$. At the beginning of block $b\in[2:B]$, the encoder has observed the full sequence of source symbols $X_{0,b-1}^n\in\mathcal{X}^n$ of the previous block. It recalls the index $m_{b-1}\in\mathcal{M}$ of the sequence $W_1^n(m_{b-1})$ used for block $b-1$. It finds an index $m_b$ such that sequences
\begin{align*}
    (X^n_{0,b-1}, W_1^n(m_{b-1}),  W_2^n(m_{b-1}, m_b))\in\mathcal{B}_\varepsilon^{(n)}(\mathcal{P}_{X_0,W_1,W_2})
\end{align*}
are jointly typical. We denote by $W_{2,b-1}^n = W_2^n(m_{b-1}, m_b)$ corresponding to the past block. Then, it returns $W_{1,b}^n = W_1^n(m_b)$ corresponding to the current block $b$, and for each time $t\in[1:n]$, it sends the symbol $U_{1,t, b}$ i.i.d. $\sim \mathcal{P}_{U_1|X_0,W_1}$ depending on $W_{1,t}(m_b)$ and $X_{0,t,b}$ causally observed in the current block $b\in[1:B]$.


\textit{Decoding function:} The decoder first returns $\Tilde{m}_1 = 1$. During block $b\in[2:B]$, the decoder recalls the past sequence $Y_{1,b-1}^n$ and the index $\Tilde{m}_{b-1}$ that corresponds to the sequence $\Tilde{W}_{1,b-1}^n = W_1^n(\Tilde{m}_{b-1})$. It observes the channel output $Y_{1,b}^n$ and finds the unique index $\Tilde{m}_b$ such that 
\begin{align*}
  &(Y_{1,b}^n, W_1^n(\Tilde{m}_b))\in\mathcal{A}_\varepsilon^{(n)}(\mathcal{P}_{Y_1,W_1}),\\
  &(Y_{1,b-1}^n, W_1^n(\Tilde{m}_{b-1}), W_2^n(\Tilde{m}_{b-1}, \Tilde{m}_b))\in\mathcal{A}_\varepsilon^{(n)}(\mathcal{P}_{Y_1, W_1, W_2}).
\end{align*}
We denote by $\Tilde{W}_{1,b}^n = W_1^n(\Tilde{m}_{b})$ and  $\Tilde{W}_{2,b-1}^n = W_2^n(\Tilde{m}_{b-1}, \Tilde{m}_{b})$ as our choice.  

\textit{Forward transmission of the decoder:} Upon receiving all sequences for all first $(B-1)$ blocks, the decoder generates $U_{2,b}^n\sim \mathcal{P}_{U_2|Y_1,W_1,W_2}^{\otimes n}$ depending on sequences $(Y_{1,b}^n, \Tilde{W}_{1,b}^n, \Tilde{W}_{2,b}^n)$ for $b\in[1:B-1]$. Finally, the decoder outputs the quantized sequence $\hat{U}_{2,b}^n$ as the reconstruction for the interim state $X_{1,b}^n$.

\textit{Termination block $B$:} As for the last block, the decoder simply outputs an all zero sequence, i.e.,  $\hat{U}_{2,B}^n = \mathbf{0}$. Usually, sequences are \textit{not} jointly typical in the last block. Therefore we omit its error analysis in the subsequent part.

\textit{Error analysis per block: }We first focus on the encoding error $\mathcal{E}^e$. For $b\in[2:B]$, let $\mathcal{E}^e_b(m_{b-1})$ denote the event of a failed encoding process during block $b$ given the knowledge of $m_{b-1}$, i.e.,  $\mathcal{E}^e_b(m_{b-1}) =  \{ \forall m_b\in\mathcal{M}, (X^n_{0,b-1}, W_1^n(m_{b-1}),  W_2^n(m_{b-1}, m_b))\notin {\mathcal{B}}_\varepsilon^{(n)}(\mathcal{P}_{X_0,W_1,W_2})\}$. Due to the independence between the codewords and Markov blocks, the probability of an encoding error in block $b$ given no encoding errors in the previous blocks is
\begin{equation}
    \mathbb{P}( \mathcal{E}_b^e(M_{b-1})\mid
                \cap_{\beta  = 2}^{b-1}\Bar{\mathcal{E}}_{\beta}^e(M_{\beta - 1}))= \mathbb{P}\brackets{ \mathcal{E}_b^e(M_{b-1})}.\nonumber
\end{equation}
If $R \geq I(W_2; X_0|W_1) + \delta$, following the Covering Lemma \cite{elgamal2011nit} and Lemma \ref{lemma: probable typical set}, we have  $\mathbb{P}\brackets{ \mathcal{E}_b^{e}(M_{b-1})}\rightarrow 0$ as $n\rightarrow\infty$, $\forall b\in[2:B]$. Thus,  state sequences for the first $B-1$ blocks are successfully encoded with probability $\rightarrow 1$ as $n\rightarrow \infty$. 

Next, we analyze the decoding error $\mathcal{E}^d$. For $b\in[2:B]$, let $\mathcal{E}_b^d(m_{b-1})$ denote the event of a failed decoding process in block $b$ given the past estimated index $m_{b-1}$. Furthermore, let $\mathcal{E}_{b-1}^{d,1}(m_{b-1})$ denote the event that the sequence $Y_{1,b-1}^n$ given $m_{b-1}$ is not jointly typical, i.e., $\mathcal{E}_{b-1}^{d,1}(m_{b-1}) = \{(Y_{1,b-1}^n, W_{1,b-1}^n(m_{b-1}), W_2^n(m_{b-1}, M_b)) \notin \mathcal{A}_{\varepsilon}^{(n)}(\mathcal{P}_{Y_1, W_1, W_2})\}$ and let $\mathcal{E}_b^{d,2}$ denote that the current sequence $Y_{1,b}^n$ is not jointly typical, i.e., 
 $\mathcal{E}_b^{d,2} = \{(Y_{1,b}^n, W_1^n(M_b))\notin\mathcal{A}_\varepsilon^{(n)}(\mathcal{P}_{Y_1,W_1})\}$. Note that, if a sequence with fewer terms $(Y_{1,b}^n,W_1^n(m_b))$ is atypical, it implies that a sequence with more terms $(Y_{1,b}^n,W_1^n(m_b), W_2^n(m_b,m_{b+1}))$ is also atypical. Therefore, we have  $\mathcal{E}_b^{d,2}\subset\mathcal{E}_{b}^{d,1}(M_b)$. Then, the decoding error probability given no past decoding error or encoding error $\mathbb{P}(\mathcal{E}_b^d(m_{b-1})\mid\cap_{\beta=2}^{b-1}\Bar{\mathcal{E}}_{\beta}^d(m_{\beta-1})\cap \Bar{\mathcal{E}}^e)$ can be upperbounded by 
\begin{align}
                &\mathbb{P}(\mathcal{E}_b^d(m_{b-1})\mid \cap_{\beta=2}^{b-1}\Bar{\mathcal{E}}_{\beta}^d(m_{\beta-1})\cap \Bar{\mathcal{E}}^e\cap \Bar{\mathcal{E}}_{b-1}^{d,1}(m_{b-1})\cap\Bar{\mathcal{E}}_b^{d,2}) \nonumber \\
                &+ \mathbb{P}(\mathcal{E}_{b-1}^{d,1}(m_{b-1})\cup\mathcal{E}_b^{d,2}\mid\cap_{\beta=2}^{b-1}\Bar{\mathcal{E}}_{\beta}^d(m_{\beta-1})\cap \Bar{\mathcal{E}}^e)\label{decoding error union bound}
            \end{align}
using the union bound. The first term of \eqref{decoding error union bound} can be upperbounded by
\begin{align*}
                &\mathbb{P}(\mathcal{E}_b^d(m_{b-1})\mid \cap_{\beta=2}^{b-1}\Bar{\mathcal{E}}_{\beta}^d(m_{\beta-1})\cap \Bar{\mathcal{E}}^e\cap \Bar{\mathcal{E}}_{b-1}^{d,1}(m_{b-1})\cap\Bar{\mathcal{E}}_b^{d,2})\\
                & = \mathbb{P}\left( \exists m^\prime\neq M_{b} \text{, s.t. }\{ W_1^n(m^\prime)\in\mathcal{A}_\varepsilon^{(n)}(W_1|Y_{1,b}^n)\}\cap\right. \\
                &\quad\quad\left.\{ W_2^n(M_{b-1}, m^\prime)\in\mathcal{A}_\varepsilon^{(n)}(W_2^n|Y_{1,b-1}^n, W_{1,b-1}^n)\}\right)\nonumber\\
                & \stackrel{\text{(a)}}{\leq} 2^{-n\varepsilon},
            \end{align*}
where (a) is due to the joint packing lemma in order to satisfy both conditions at the same time given that $R\leq I(W_1;Y_1) + I(W_2;Y_1|W_1) - 7\varepsilon$. Moreover, as for the second term in \eqref{decoding error union bound}, we have
\begin{align*}
                &\mathbb{P}(\mathcal{E}_{b-1}^{d,1}(m_{b-1})\cup\mathcal{E}_b^{d,2}| \cap_{\beta=2}^{b-1}\Bar{\mathcal{E}}_{\beta}^d(m_{\beta-1})\cap \Bar{\mathcal{E}}^e)\\
                &\stackrel{\text{(b)}}{\leq} \mathbb{P}((Y_{1,b-1}^n, W_{1,b-1}^n,W_{2,b-1}^n)\notin \mathcal{A}_\varepsilon^{(n)} \mid \\
                &\quad\quad (X_{0,b-1}^n, W_{1,b-1}^n, W_{2,b-1}^n)\in\mathcal{B}_\varepsilon^{(n)} )
                \\& \quad+\mathbb P((Y_{1,b}^n, W_{1,b}^n,W_{2,b}^n)\notin\mathcal{A}_\varepsilon^{(n)}|(X_{0,b}^n, W_{1,b}^n,W_{2,b}^n)\in\mathcal{B}_\varepsilon^{(n)})\nonumber\\
                &\stackrel{\text{(c)}}{\leq} 2\cdot\max_{(x_0^n,w_1^n,w_2^n)\in\mathcal{B}_\varepsilon^{(n)}}\eta^{(n)}(x_0^n,w_1^n,w_2^n)\\
                &\leq 2\sqrt{\delta_n}\xrightarrow[]{n\rightarrow\infty} 0,
            \end{align*}
where (b) comes from the independence of each Markov block and the fact that $\{(Y_{1,b}^n, W_{1,b}^n)\notin\mathcal{A}_\varepsilon^{(n)}\}\subset \{(Y_{1,b}^n,W_{1,b}^n,W_{2,b}^n)\notin\mathcal{A}_\varepsilon^{(n)}\} $, and (c) comes from the definition of $\mathcal{B}_\varepsilon^{(n)}$.

Therefore, following the above arguments, the encoding error and the decoding error are both asymptotically zero.

\textit{Witsenhausen costs analysis:} We first analyze the power cost. To implement the power constraint, we impose a mean-square transmission power constraint as done in \cite{choudhuri2013causal, Grover2010Witsenhausen}:
\begin{equation}
    \sum_{t=1}^n \mathbb E\sbrackets{U_{1,t}^2}\leq nP,
    \label{eq:trans power constraint}
\end{equation}
where the expectation is taken over random sequences $X_0^n$ and $W_1^n$. Therefore, the generated sequence $U_1^n$ satisfies the condition of uniform integrability. Thus, for $b\in[1:B]$, since at each time $U_{1,t,b}$ is generated i.i.d. with finite second moment $\mathbb E\sbrackets {U_1^2}=P$, due to the LLN, the non-negative averaged power sequence $c_P(U_{1,b}^n) = \frac{1}{n}\sum_{t=1}^n(U_{1,t,b}^2)$ converges in probability to $P$. Moreover, since sequence $U_{1,b}^n$ is also uniformly integrable, it holds (stronger) that 
\begin{equation}
    \mathbb E\sbrackets{|c_P(U_{1,b}^n) - P|}<\varepsilon,\quad b\in[1:B]\label{eq: power block b}
\end{equation}
for large $n$, more details in \cite{Bogachev2007measure}. Therefore, for all $B$ blocks, we have 
\begin{align}
    \mathbb E\sbrackets{|c_P(U_1^{Bn}) - P|}\leq  \frac{1}{B}\sum_{b=1}^B\mathbb E\sbrackets{|c_P(U_{1,b}^n) - P|}\leq \varepsilon,\label{eq: power cost B blocks}
\end{align}
for $n$ sufficiently large.

Next, we examine the estimation cost. Define the indicator function $\rho_b = 0$ to denote the absence of coding error in block $b\in[1:B-1]$. Now, take $\phi_b = (1-\psi_b^{(n)})(1-\rho_b) $ to indicate the occurrence of desired sequences that satisfy both the estimation cost and the typicality constraints AND no coding error event in block $b$. From the LLN and the error analysis, we have $\mathbb E\sbrackets{\phi_b}\rightarrow 1$ as $n\rightarrow\infty$ for $b\in[1:B-1]$. In particular, if $\phi_b = 1$, we have $\mathbb E\sbrackets{|c_S(X_{1,b}^n , \hat{U}_{2,b}^n) - S|\mid\phi_b = 1}<\frac{1}{12}\varepsilon$. Therefore, from the law of total expectation, we have
\begin{align*}
    &\mathbb E\sbrackets{|c_S(X_{1,b}^n , \hat{U}_{2,b}^n) - \hat{S}|}\\
    &= \mathbb P\brackets{\phi_b=1}\cdot\mathbb E\sbrackets{|c_S(X_{1,b}^n , \hat{U}_{2,b}^n) - \hat{S}|\mid \phi_b = 1} \\
    &\quad + \mathbb P\brackets{\phi_b=0}\cdot\mathbb E\sbrackets{|c_S(X_{1,b}^n , \hat{U}_{2,b}^n) - \hat{S}|\mid\phi_b = 0}\nonumber\\
     &\leq \frac{1}{12}\varepsilon +  \mathbb P\brackets{\phi_b = 0}\cdot\hat{S} \\
     &\quad+ \mathbb P\brackets{\phi_b=0}\cdot\mathbb E\sbrackets{c_S(X_{1,b}^n , \hat{U}_{2,b}^n)\mid\phi_b = 0}.\nonumber
\end{align*}
For $n$ sufficiently large, the second term $\mathbb P\brackets{\phi_b = 0}\cdot\hat{S}\leq \frac{1}{12}\varepsilon$, since $\hat{S}<\infty$. Next, we evaluate the last term above. In particular, we need $c_S(X_{1,b}^n , \hat{U}_{2,b}^n)$ to be stochastically bounded when an unwanted error event happens. Since 
\begin{align}
    &\mathbb P\brackets{\phi_b=0}\cdot\mathbb E \sbrackets{c_S(X_{1,b}^n , \hat{U}_{2,b}^n)\mid\phi_b = 0}\label{eq: estimation cost error}\\ 
    &= \mathbb P\brackets{\phi_b=0}\cdot\mathbb E\sbrackets{\frac{1}{n}\sum_{t=1}^n |X_{0,t,b} + U_{1,t,b}- \hat{U}_{2,t,b}|^2\mid\phi_b = 0},\nonumber
\end{align}
by applying the inequality $(a+b)^2\leq 2a^2+2b^2$, $\forall  a,b\in\mathbb R$, we have
\begin{align}
    &\mathbb P(\phi_b=0)\cdot\mathbb E \sbrackets{c_S(X_{1,b}^n , \hat{U}_{2,b}^n)\mid\phi_b = 0}\nonumber\\
    &\leq 2\mathbb P(\phi_b=0)\cdot\mathbb E \sbrackets{c_P(U_{1,b}^n)\mid\phi_b = 0} \nonumber\\
    &\quad + 2\mathbb P(\phi_b=0)\cdot\mathbb E\sbrackets{c_S(X_{0,b}^n, \hat{U}_{2,b}^n)\mid\phi_b = 0}\nonumber\\
    & \leq 2\mathbb P(\phi_b=0)\cdot \mathbb E \sbrackets{|c_P(U_{1,b}^n) - P| + P\mid\phi_b = 0}\nonumber \\
    &\quad + 2\mathbb P(\phi_b=0)\cdot\mathbb E\sbrackets{c_S(X_{0,b}^n, \hat{U}_{2,b}^n)\mid\phi_b = 0}\nonumber\\
    &= 2P\mathbb P(\phi_b=0) +  2\mathbb P(\phi_b=0)\cdot\mathbb E \sbrackets{|c_P(U_{1,b}^n) - P| \mid\phi_b = 0}\nonumber\\
    &\quad+ 2\mathbb P(\phi_b=0)\cdot \mathbb E\sbrackets{c_S(X_{0,b}^n, \hat{U}_{2,b}^n)\mid\phi_b = 0}.\label{eq: need to use wyner's trick}
    \end{align}
Here, the first term $2P\mathbb P(\phi_b=0)\leq \frac{1}{36}\varepsilon$ since $P< \infty$. And the second term $2\mathbb P\brackets{\phi_b=0}\cdot\mathbb E \sbrackets{|c_P(U_{1,b}^n) - P| \mid\phi_b = 0}\leq 2\mathbb E\sbrackets{|c_P(U_{1,b}^n) - P|}\leq \frac{1}{36}\varepsilon$ due to the nonnegativity of $|c_P(U_{1,b}^n) - P|$ and the power constraint \eqref{eq: power block b}.

Next, we analyse the last term in \eqref{eq: need to use wyner's trick} with Wyner's trick \cite{wyner1978rate} exploiting the discretization of $U_2$ as follows
    \begin{align*}
        &\mathbb E\sbrackets{c_S(X_{0,b}^n, \hat{U}_{2,b}^n)\mid\phi_b = 0} \\
        &\leq \frac{1}{n}\sum_{t=1}^n\mathbb E\sbrackets{D(X_{0,t,b})\mid\phi_b = 0},
    \end{align*}
with $D(X_{0,t,b}) = \max_{\hat{u}_{2,k}}c_S(X_{0,t,b}, \hat{u}_{2,k})$. The random variables $\{D(X_{0,t,b})\}_{t=1}^n$ are i.i.d. and integrable since $c_S(\cdot, \cdot)$ is a quadratic measure and $X_{0,t,b}$ is i.i.d. Gaussian distributed. Next, let $\chi_{\{D(X_{0,t,b})>d\}}$ denote an indicator function indicating if $D(X_{0,t,b})>d$ for some $0<d<\infty$. Then we have 
\begin{align}
    &\mathbb P\brackets{\phi_b=0}\cdot\mathbb E\sbrackets{c_S(X_{0,b}^n, \hat{U}_{2,b}^n)\mid\phi_b = 0}\label{eq: upper bound wyner} \\
    &\leq d\cdot\mathbb P\brackets{\phi_b=0} +\frac{1}{n}\sum_{t=1}^n \mathbb E\sbrackets{D(X_{0,t,b})\chi_{\{D(X_{0,t,b})>d\}}}.\nonumber
\end{align}
Given that $X_{0,t,b}$ is generated i.i.d, and $D(X_{0,t,b})$ is integrable, for any $\varepsilon_d>0$, there must exist a $d_0$ such that\\ $\mathbb E\sbrackets{D(X_{0,t,b})\chi_{\{D(X_{0,t,b})>d\}}}<\varepsilon_d$ for all $d>d_0$, as per the Monotone Convergence Theorem. Thus, for a sufficiently small $\varepsilon_d$ and a sufficiently large $n$, term $\eqref{eq: upper bound wyner}$ can be upper bounded by $\frac{1}{72}\varepsilon$, so that \eqref{eq: estimation cost error} $\leq\frac{1}{12}\varepsilon$.

Thus, compared to the original cost constraint $S$, for block $b\in[1:B-1]$ we have
    \begin{align}
        &\mathbb E\sbrackets{|c_S(X_{1,b}^n, \hat{U}_{2,b}^n) - S|}\nonumber\\
        &\leq |\hat{S} - S| + \mathbb E\sbrackets{|c_S(X_{1,b}^n , \hat{U}_{2,b}^n) - \hat{S}|}\nonumber\\
        &\leq \frac{1}{4}\varepsilon + \mathbb E\sbrackets{|c_S(X_{1,b}^n , \hat{U}_{2,b}^n) - \hat{S}|}\leq \frac{1}{2}\varepsilon.\label{eq: compared to original S}
    \end{align}

Lastly, we estimate the cost of the last block where the output $\hat{U}_{2,B}^n$ is set to be an all-zero sequence. In this case
    \begin{align*}
        &\mathbb E\sbrackets{|c_S(X_{1,B}^n, \hat{U}_{2,B}^n) - \hat{S}|}\\
        &\leq  \mathbb E\sbrackets{|2c_P(X_{0,B}^n)+2c_P(U_{1,B}^n) - \hat{S}|}\\
        &\leq 2\cdot\mathbb E\sbrackets{c_P(X_{0,B}^n)} + 2\cdot\mathbb E\sbrackets{c_P(U_{1,B}^n)} + \hat{S}\\
        &\stackrel{\text{(d)}}{\leq}  2\cdot Q+ 2\cdot P + \hat{S}<\infty,
    \end{align*}
where (d) comes from that $X_{0,B}^n$ is generated i.i.d. $\sim\mathcal{N}(0, Q\mathbb I)$. Therefore, similar to \eqref{eq: compared to original S}, we have
\begin{equation}
    \mathbb E\sbrackets{|c_S(X_{1,B}^n, \hat{U}_{2,B}^n) - S|}\leq \frac{1}{4}\varepsilon + 2\cdot Q+ 2\cdot P + \hat{S} \triangleq \Tilde{S}.\nonumber
\end{equation}

For the average estimation cost, the impact of the last pathological block can be effectively diminished, since
    \begin{align}
        &\mathbb E\sbrackets{|c_S(X_1^{Bn}, \hat{U}_{2}^{Bn}) - S|}\nonumber\\
        &\leq \frac{1}{B}\sum_{b=1}^{B}\mathbb E\sbrackets{|c_S(X_{1,b}^n, \hat{U}_{2,b}^n) - S|}\nonumber\\
        & \leq \frac{1}{B}\Tilde{S} + \frac{B-1}{B}\cdot \frac{1}{2}\varepsilon\leq\varepsilon,
        \label{eq: est cost B blocks}
    \end{align}
for $n$ and $B$ sufficiently large.

\textit{Existence of the desired code: }With the aid of random coding arguments and applying the Selection Lemma \cite{bloch2011physical} to both conditions \eqref{eq: power cost B blocks} \eqref{eq: est cost B blocks}, we can obtain that there exists a control design $c\in\mathcal{C}(Bn)$ that simultaneously satisfies
\begin{align*}
    \mathbb E\sbrackets{|c_P(U_1^{Bn}) - P|}<\varepsilon,\quad \mathbb E[|c_S(X_1^{Bn}, \hat{U}_2^{Bn}) - S|]<\varepsilon.
\end{align*}

\textit{Closedness: }So far we have only shown under strict inequality of rate constraint \eqref{info result} that the cost constraints hold. As for equality, we consider a sequence of probability distribution $\mathcal{P}^k\xrightarrow[]{k\rightarrow\infty}\mathcal{P}^*$, where each $\mathcal{P}^k$ satisfies the rate constraint with strict inequality, and $\mathcal{P}^*$ satisfies it with equality. Therefore, 
\begin{align*}
   &\mathbb E_{\mathcal{P}^*}[\abs{c_S(X_1^n, U_2^n) - S}]\\
   &\leq\mathbb E_{\mathcal{P}_k}[\abs{c_S(X_1^n, U_2^n) - S}]+\mathbb E_{\mathcal{P}^*-\mathcal{P}_k}[\abs{c_S(X_1^n, U_2^n) - S}]\\
   &\leq 1/2\cdot \varepsilon + 1/2\cdot \varepsilon\leq\varepsilon
\end{align*}
where the first $1/2\cdot \varepsilon$ comes from the achievability proof with strict inequality as $n$ large enough, and the second $1/2\cdot \varepsilon$ could similarly be derived repeating the estimation cost analysis steps above since the probability measure $\mathcal{P}^*-\mathcal{P}_k\rightarrow 0$ as $k\rightarrow \mz{\infty}$. Moreover, the power constraint for $\mathcal{P}^*$ also holds due to \eqref{eq:trans power constraint}. This concludes our achievability proof.

\hfill \qedsymbol{}

\section{Converse Proof of Theorem \ref{thm: c-n wits main}}\label{app: conv proof}
We consider a control design $c\in \mathcal{C}(n)$ of block length $n\in\mathbb N^\star$ that achieves a pair of costs $(P,S)$. According to Csiszár sum identity, we have
  \begin{align*}
    0&=\sum_{t=1}^n I\brackets{X_0^{t-1}; Y_{1,t}|Y_{1, t+1}^n} - \sum_{t=1}^n I\brackets{Y_{1,t+1}^n; X_{0, t}| X_0^{t-1}}\\
    &\leq \sum_{t=1}^n I\brackets{X_0^{t-1}, Y_{1, t+1}^n; Y_{1,t}} - \sum_{t=1}^n I\brackets{Y_{1,t+1}^n; X_{0, t}|X_0^{t-1}}\\
    & \stackrel{\text{(a)}}{=} \sum_{t=1}^n I(W_{1, t}, W_{2, t}; Y_{1,t}) - \sum_{t=1}^n I(W_{2, t}; X_{0, t}|W_{1, t})\\
    & \stackrel{\text{(b)}}{=} n\cdot ( I(W_{1, T}, W_{2, T}; Y_{1,T}|T) - I(W_{2, T}; X_{0, T}|W_{1, T}, T))\\
    & \leq n\cdot ( I(W_{1, T}, T,  W_{2, T}; Y_{1,T}) - I(W_{2, T}; X_{0, T}|W_{1, T}, T))\\
    & \stackrel{\text{(c)}}{=} n\cdot (I(W_1, W_2; Y_1) - I(W_2; X_0|W_1)),
\end{align*}
where (a) is from the identification of the auxiliary RVs $W_{1, t} = X_0^{t-1}$, $W_{2, t} = Y_{1, t+1}^n$, (b) is from introducing the uniform random variable $T\in\{1,...,n\}$ and the corresponding mean auxiliary random variables $X_{0, T}, W_{1, T}, W_{2, T}, Y_{1, T}$, where $Y_{1, T}$ is distributed according to
\begin{equation}
    \mathbb{P}(Y_{1, T} = y_1) = \frac{1}{n}\sum_{t=1}^n\mathbb{P}(Y_{1,t} = y_1), \quad \forall y_1\in\mathcal{Y}_1
    \nonumber
\end{equation}
and (c) comes from the introduction of $X_0 \triangleq X_{0, T}, W_1 \triangleq (W_{1, T}, T), W_2 \triangleq W_{2, T}, Y_1 \triangleq Y_{1, T}$.

Next, we show that the introduced auxiliary RVs $W_{1,t} = X_0^{t-1}$, and $W_{2, t} = Y_{1, t+1}^n$ in equation (a) satisfy the Markov chain properties \eqref{markov result}.

For $t\in\{1,...,n\}$:
\begin{itemize}
    \item $X_{0, t}$ and $W_{1,t} = X_0^{t-1}$ are independent because of the i.i.d. property of the source.
    \item $U_{1,t} -\!\!\!\!\minuso\!\!\!\!- (X_{0,t}, W_{1,t}) -\!\!\!\!\minuso\!\!\!\!- W_{2,t}$ is given by the causal encoding, that the current encoder's action $U_{1,t}$ depends on the past and current input states $(X_{0,t}, X_0^{t-1})$, not on the future sequence $W_{2,t}$.
    \item $(X_{1,t}, Y_{1,t})-\!\!\!\!\minuso\!\!\!\!- (X_{0,t}, U_{1,t})    -\!\!\!\!\minuso\!\!\!\!- (W_{1,t}, W_{2,t})$ comes from the memoryless property of the channel.
    \item $U_{2, t} -\!\!\!\!\minuso\!\!\!\!- (W_{1,t}, W_{2,t}, Y_{1, t}) -\!\!\!\!\minuso\!\!\!\!- (X_{0, t}, U_{1, t}, X_{1, t})$ is shown as Lemma \ref{lemma: c-n markov chain} in Appendix \ref{subsec: markov chain}. 
\end{itemize}
This implies that the auxiliary random variables $X_0 = X_{0, T}, W_1 = (W_{1, T}, T), W_2 = W_{2, T}, U_1 = U_{1, T}, X_1 = X_{1,T}, Y_1 = Y_{1, T}, U_2 = U_{2, T}$ satisfy the Markov chains \eqref{markov result}. Equivalently, the distribution of auxiliary random variables decomposes as in \eqref{prob result}.

Now, we reformulate $n$-stage costs with a control design $c\in\mathcal{C}(n)$. By defining distributions $\mathbb{P}(U_{1, T} = u_1) = \frac{1}{n}\sum_{t=1}^n\mathbb{P}(U_{1,t} = u_1), \quad \forall u_1\in\mathcal{U}_1$ and $\mathbb{P}(X_{1,T}=x_{1,T}, U_{2,T} = u_{2,T})= \frac{1}{n}\sum_{t=1}^n\mathbb{P}(X_{1,t} = x_1, U_{2,t}=u_2), \quad \forall (x_1,u_2)\in\mathcal{X}_1\times\mathcal{U}_2$ and using auxiliary random variables $U_1 = U_{1, T}, X_1 = X_{1,T}, U_2 = U_{2, T}$, we have
\begin{align*}
    &\mathbb E[c_P(U_1^n)] = \mathbb E\sbrackets{\frac{1}{n}\sum_{t=1}^n U_{1, t}^2} = \mathbb E\sbrackets{U_{1, T}^2} = \mathbb E \sbrackets{U_1^2},\\
    &\mathbb E[c_S(X_1^n,U_2^n)]= \mathbb E\sbrackets{\frac{1}{n}\sum_{t=1}^n (X_{1, t} - U_{2, t})^2} \\
    &\quad\quad\quad\quad\quad\quad= \mathbb E\sbrackets{(X_{1, T} - U_{2, T})^2} 
    = \mathbb E \sbrackets{(X_{1} - U_{2})^2}.
\end{align*}

In conclusion, if a pair of costs $(P, S)\in\mathbb R^2$ is achievable, then for all $\varepsilon>0$, there exists $\Bar{n}\in\mathbb N^\star$ such that for all $n\geq \Bar{n}$, there exists a control design $c\in\mathcal{C}(n)$ such that its induced long-run costs
\begin{align}
&\mathbb E\Big[\big|P - c_P(U_1^n)\big| + \big|S - c_S(X_1^n, U_2^n)\big|\Big] \leq \varepsilon\nonumber
    \\
     &\Longrightarrow|P - \mathbb E[c_P(U_1^n)]| + |S - \mathbb E[c_S(X_1^n,U_2^n)]| \leq \varepsilon\nonumber
     \\
     &\Longrightarrow |P - \mathbb E \sbrackets{U_1^2}| + |S - \mathbb E \sbrackets{(X_{1} - U_{2})^2}| \leq \varepsilon
     \label{eq: converse achievable result}
\end{align}
The equation \eqref{eq: converse achievable result} is valid for all $\varepsilon>0$. This shows that $(X_0,W_1,W_2,U_1,X_1,Y_1,U_2)$ satisfy \eqref{prob result}, \eqref{info result}, \eqref{cost result}.
\subsection{Proof of the Markov Chain}\label{subsec: markov chain}
\begin{lemma}\label{lemma: c-n markov chain}
    For each time $t\in\{1,...,n\}$, it holds that $U_{2, t} -\!\!\!\!\minuso\!\!\!\!- (W_{1,t}, W_{2,t}, Y_{1, t}) -\!\!\!\!\minuso\!\!\!\!- (X_{0, t}, U_{1, t}, X_{1, t})$.
\end{lemma}
\begin{proof}[Proof of Lemma \ref{lemma: c-n markov chain}]
\begin{align}
         &\!\mathcal{P}(U_{2,t}|Y_{1,t}, W_{1, t}, W_{2, t}, X_{0, t},U_{1, t}, X_{1, t})\nonumber\\
        &=\quad \mathcal{P}(U_{2,t}|Y_{1,t}, X_0^{t-1}, Y_{1, t+1}^n, X_{0, t},U_{1, t}, X_{1, t})\nonumber\\
        &\stackrel{\text{(a)}}{=} \sum \mathcal{P}(U_1^{t-1}, X_1^{t-1}, Y_1^{t-1},U_{2,t}|Y_{1,t}, X_0^{t-1}, Y_{1, t+1}^n, X_{0, t},\nonumber\\
        &\qquad\qquad\qquad\qquad\qquad\qquad\qquad\qquad U_{1, t}, X_{1, t})\nonumber\\
        &=\sum\mathcal{P}(U_1^{t-1}|Y_{1,t}, X_0^{t-1}, Y_{1, t+1}^n, X_{0, t}, U_{1, t}, X_{1, t})\tag{\text{b1}}\\
        &\quad\times \mathcal{P}(X_1^{t-1}|Y_{1,t}, X_0^{t-1}, Y_{1, t+1}^n, X_{0, t},U_{1, t}, X_{1, t}, U_1^{t-1})\tag{\text{c1}}\\
        &\quad\times \mathcal{P}(Y_1^{t-1}|Y_{1,t}, X_0^{t-1}, Y_{1, t+1}^n, X_{0, t},U_{1, t}, X_{1, t}, U_1^{t-1}, X_1^{t-1})\tag{\text{d1}}\\
        &\quad\times\mathcal{P}(U_{2,t}|Y_{1,t}, X_0^{t-1}, Y_{1, t+1}^n, X_{0, t},U_{1, t}, X_{1, t}, U_1^{t-1}, \tag{\text{e1}}\\
        &\qquad\qquad \quad X_1^{t-1}, Y_1^{t-1})\nonumber\\
        &=\quad \sum\mathcal{P}(U_1^{t-1}|Y_{1,t}, X_0^{t-1}, Y_{1, t+1}^n)\tag{\text{b2}}\\
        &\quad\times\quad  \mathcal{P}(X_1^{t-1}|Y_{1,t}, X_0^{t-1}, Y_{1, t+1}^n, U_1^{t-1})\tag{\text{c2}}\\
        &\quad\times \quad \mathcal{P}(Y_1^{t-1}|Y_{1,t}, X_0^{t-1}, Y_{1, t+1}^n, U_1^{t-1}, X_1^{t-1})\tag{\text{d2}}\\
        &\quad\times \quad \mathcal{P}(U_{2,t}|Y_{1,t}, X_0^{t-1}, Y_{1, t+1}^n, U_1^{t-1}, X_1^{t-1}, Y_1^{t-1})\tag{\text{e2}},
    \end{align}
where in equality (a) we sum over $(U_1^{t-1}, X_1^{t-1}, Y_1^{t-1})$

In the last equality, we remove the dependence on $ X_{0, t},U_{1, t}, X_{1, t}$ from all the terms. This is because
\begin{itemize}
    \item (b1) $\rightarrow$ (b2) comes from the causal encoding: $U_1^{t-1}$ fully relies on $X_0^{t-1}$ and not on the future sequence.
    \item (c1) $\rightarrow$ (c2) comes from the deterministic generation of the interim state: $X_{1,t} = X_{0, t}+U_{1, t}$, $\forall t\in\{1,...,n\}$.
    \item (d1) $\rightarrow$ (d2) comes from the memoryless property of the channel: $Y_1^{t-1}$ only depends on $(X_0^{t-1 }, U_1^{t-1}, X_1^{t-1})$.
    \item (e1) $\rightarrow$ (e2) comes from the noncausal decoding: $U_{2,t}$ is a deterministic function of $(Y_{1}^{t-1}, Y_{1,t}, Y_{1,t+1}^n)$.
\end{itemize}
Hence we have,
    \begin{align*}
        &\mathcal{P}(U_{2,t}|Y_{1,t}, W_{1, t}, W_{2, t}, X_{0, t},U_{1, t}, X_{1, t})\\
        &=\!\sum_{U_1^{t-1}, X_1^{t-1}, Y_1^{t-1}} \!\mathcal{P}(U_1^{t-1}, X_1^{t-1}, Y_1^{t-1} U_{2,t}|Y_{1,t}, X_0^{t-1}, Y_{1, t+1}^n)\\
        &= \quad  \mathcal{P}(U_{2,t}|Y_{1,t}, X_0^{t-1}, Y_{1, t+1}^n)\\
        &= \quad  \mathcal{P}(U_{2,t}|Y_{1,t}, W_{1,t}, W_{2,t}).
    \end{align*}
\end{proof}

\section{Proof for Proposition \ref{prop: single-shot zero est}}\label{app: proof of S=0 P=Q}



\mz{The idea of this proof is that, given the estimation cost is zero, i.e., $S=0$, first, we show that $X_1$ has to be almost surely (a.s.) a constant $c\in\mathbb R$, then we optimize the power cost over $c$.\\
\indent Suppose $(X_1,Z_1,Y_1) $ is a random vector
\begin{align*}
    (X_1,Z_1,Y_1):(\Omega, \mathcal{F},\mathbb P)\longrightarrow\mathbb R^3.
\end{align*}
The optimal decision rule for DM2 is 
\begin{align}
    m(y) = \mathbb E[X_1|Y_1 = y],\forall y.\label{eq: borel function m}
\end{align}
Thus, if we want to have the MMSE
\begin{align*}
    \mathbb E[(X_1 - \mathbb E[X_1|Y_1])^2] = 0,
\end{align*}
then we need
\begin{align*}
    X_1 &= m(Y_1), \mathbb P\text{-a.s.}
\end{align*}
Now, denote the distribution of $Y_1$ to be $P_{Y_1}$, i.e., for every Borel set $B\subseteq\mathbb R$
\begin{align*}
    P_{Y_1}(B) = \mathbb P(\{\omega: Y_1(\omega)\in B\}) = \mathbb P(Y_1\in B).
\end{align*}
Then, from \eqref{eq: borel function m}, we get that
the conditional distribution of $X_1$ given $Y_1 = y$ is a point mass for $P_{Y_1}$-almost every $y$, i.e,
for every Borel set $A\subseteq \mathbb R$,
\begin{align}
    P_{X_1|Y_1=y}(A) = \delta_{ m(y)}(A)\in\{0,1\},  \text{ for $P_{Y_1}$ almost every }y.\label{eq: point mass X_1 given y}
\end{align}
Now, in order to show that $X_1 = c$ a.s. for some $c\in\mathbb R$, we aim to show that the probability distribution of $X_1$, denoted by $P_{X_1}$, is a point mass. In particular, we want to have, for every Borel set $A\subseteq\mathbb R$,
\begin{align*}
    P_{X_1}(A) = \mathbb P(\{\omega: X_1(\omega)\in A\}) = \mathbb P(X_1\in A) =\delta_c(A)\in\{0,1\}.
\end{align*}
\indent Assume, for contradiction, that $P_{X_1}$ is not a Dirac measure. Then there exists a Borel set $C\subset\mathbb R$ such that
\begin{align*}
    0<P_{X_1}(C)<1.
\end{align*}
Now, consider the posterior distribution of $X_1$ given $Y_1 = y$. Since $Y_1 = X_1+Z_1$, with an additive noise $Z_1\sim\mathcal{N}(0,N)$, $N>0$, that is, independent, by the Bayesian formula, we have\begin{align}
    P_{X_1|Y_1=y}(C) = &\mathbb P\brackets{X_1\in C|Y_1=y} =\frac{\int_{C}\phi_N(y-x)dP_{X_1}(x)}{\int_{\mathbb R}\phi_N(y-u)dP_{X_1}(u)},\label{eq: bayes rule}
\end{align}
where
\begin{align*}
    \phi_N(t) = \frac{1}{\sqrt{2\pi N}}e^{-t^2/(2N)}.
\end{align*}
Now, because
\begin{align*}
    \phi_N(y-x)>0,\quad\forall x,y\in\mathbb R, \quad\text{and }P_{X_1}(C)>0,
\end{align*}
from Lemma \ref{lemma:positive integral} shown below, we get that the numerator of \eqref{eq: bayes rule}
\begin{align*}
    \int_{C}\phi_N(y-x)dP_{X_1}(x)>0.
\end{align*}
Also, since $P_{X_1}(C)<1$, we have $P_{X_1}(C^c)>0$, so
\begin{align*}
    \int_{C^c}\phi_N(y-x)dP_{X_1}(x)>0.
\end{align*}
And since the denominator of \eqref{eq: bayes rule} is
\begin{align*}
    &\int_{\mathbb R}\phi_N(y-x)dP_{X_1}(x)\\
    &=\int_{C^c}\phi_N(y-x)dP_{X_1}(x)+\int_{C}\phi_N(y-x)dP_{X_1}(x)\\
    &>\int_{C}\phi_N(y-x)dP_{X_1}(x)>0,
\end{align*}
we have that
\begin{align*}
    0<P_{X_1|Y_1=y}(C)<1,
\end{align*}
which brings us the contradiction that $P_{X_1|Y_1=y}(C)\in\{0,1\}$ for $P_{Y_1}$ almost every $y$, i.e., eq. \eqref{eq: point mass X_1 given y}. Therefore, $X_1=c$, a.s.\\
\indent Hence, our power optimization problem (17) subject to zero estimation cost, i.e,
\begin{align*}
    \min_{f:X_0\mapsto U_1,\text{ s.t. } S=0}\mathbb E[(X_1-X_0)^2],
\end{align*}
now reduces to 
\begin{align*}
    \min_c\mathbb E[(c- X_0)^2],
\end{align*}
whose minimum value $Q$ is attained by $c = \mathbb E[X_0]=0$.\hfill \qedsymbol{}}
\mz{\begin{lemma}\label{lemma:positive integral}
    If a function $f(x)>0$ for every $x\in\mathbb R$, and the probability distribution $\mu(A)>0$, then
    \begin{align*}
        \int_A f(x)d\mu(x)>0.
    \end{align*}
    \begin{proof}
        Since $f(x)>0$ on $A$,
        \begin{align*}
            A = \bigcup_{n=1}^\infty \brackets{A\cap \bigg\{ x:f(x)\geq \frac{1}{n} \bigg\}}.
        \end{align*}
        If every set in this union had $\mu$-measure zero, then by countable subadditivity,
        \begin{align*}
            \mu(A) = 0,
        \end{align*}
        contradicting $\mu(A)>0$. Hence there exists some $n\in\mathbb N$ such that
        \begin{align*}
            \mu\brackets{A\cap \bigg\{ x:f(x)\geq \frac{1}{n} \bigg\}}>0.
        \end{align*}
        Therefore,
        \begin{align*}
            \int_A f(x)d\mu(x)&\geq \int_{A\cap \{ f\geq \frac{1}{n} \}}f(x)d\mu(x)\\
            &\geq 
            \frac{1}{n}\mu\brackets{A\cap \bigg\{ f\geq \frac{1}{n} \bigg\}}>0.
        \end{align*}
    \end{proof}
\end{lemma}}

\section{Achievability Proof of Theorem \ref{theorem: c-n w-f}}\label{app: achievability proof c-n w-f}
We consider an arbitrary but fixed $\varepsilon>0$ and assume the sequence $(X_0^n, W_1^n, U_1^n, X_1^n, Y_1^n, U_2^n)$ is generated i.i.d. according to a distribution that decomposes as $\mathcal{P}_{X_0}\mathcal{P}_{W_1}\mathcal{P}_{U_1|X_0,W_1}\mathcal{P}_{X_1, Y_1|X_0, U_1}\mathcal{P}_{U_2|X_0, W_1, Y_1}$, with $P = \mathbb E[U_1^2]$, $S=\mathbb E[(X_1 - U_2)^2]$.
Let $\psi^{(n)}: \mathcal{X}_0^n\times \mathcal{W}_1^n\times \mathcal{U}_1^n\times\mathcal{X}_1^n\times\mathcal{Y}_1^n\times \mathcal{U}_2^n\rightarrow\{0,1\}$ denote an indicator function for sequences of length $n$ with
    \begin{align*}
        &\psi^{(n)}\!(x_0^n,w_1^n,u_1^n,x_1^n,y_1^n,u_2^n)\nonumber \!=\! 
        \left\{
          \begin{aligned}
              &1 \!\!&\text{if }|c_S(x_1^n,u_2^n)-S|\geq \frac{1}{12}\varepsilon\\
              &0 \!\! &\text{otherwise}.
          \end{aligned}   
        \right.
    \end{align*}
Using the weak law of large numbers (LLN) and the union bound we have
\begin{align*}
    \mathbb E[\psi^{(n)}(X_0^n,W_1^n,U_1^n,X_1^n,Y_1^n,U_2^n)]\leq \delta_n\rightarrow 0,
\end{align*}
as $n\rightarrow\infty$. 

Now, we provide a coding scheme by considering a probability distribution $\mathcal{P}_{X_0, W_1, U_1, X_1, Y_1, U_2}$ that decomposes as \eqref{eq: c-n w-f prob result} such that exists a $\delta>0$ and a rate $R>0$ such that 
\begin{align}
    I(U_2;X_0,Y_1|W_1) + \delta\leq R\leq I(W_1,U_2;Y_1)-\delta.
\end{align}
where both the lower bound and the upper bound are also assumed to be finite. 
Let $c\in\mathcal{C}_{\mathsf{f}}(Bn)$ a block-Markov random code over $B\in\mathbb{N}^\star$ blocks each of length $n\in\mathbb N^\star$. See the illustration in Figure \ref{fig: w_f Markov block codes}.

 \begin{figure}[t]
        \centering

\begin{tikzpicture}[x=0.5cm,y=0.5cm]

\foreach \y in {-6,-4,-2,0,2,4,6,8} {
  \draw[line width=2pt] (0,\y) -- (12,\y);
}

\foreach \x in {2,4,6,8,10} {
  \draw[dashed,line width=0.5pt] (\x,-7) -- (\x,9);
}

\node[anchor=east] at (-0.5,8) {$X_0^n$};
\node[anchor=east] at (-0.5,6) {$W_1^n$};
\node[anchor=east] at (-0.5,4) {$U_2^n$};
\node[anchor=east] at (-0.5,2) {$U_1^n$};
\node[anchor=east] at (-0.5,0) {$X_1^n$}; 
\node[anchor=east] at (-0.5,-2) {$Y_1^n$};
\node[anchor=east] at (-0.5,-4) {$\hat W_1^n$};
\node[anchor=east] at (-0.5,-6) {$\hat U_2^n$};

\node at (3,9) {$b-2$};
\node at (5,9) {$b-1$};
\node at (7,9) {$b$};
\node at (9,9) {$b+1$};

\newcommand{\hatchrect}[5]{%
  \path[pattern=north east lines,pattern color=#5,draw=#5]
    (#1,#2) rectangle (#3,#4);
}
\draw[->,line width=1.1pt,red] (5,7.7) -- (5,4.3); 
\draw[<-,line width=1.1pt,red] (6.5,5.7) -- (5.3,4.3); 
\draw[->,line width=1.1pt,red] (7,5.7) -- (7,2.3);
\draw[->,line width=1.1pt,red] (7,1.7) -- (7,0.3); 
\draw[->,line width=1.1pt,red] (7,-0.3) -- (7,-1.7);
\draw[->,line width=1.1pt,red] (7,-2.3) -- (7,-3.7);
\draw[->,line width=1.1pt,red] (7,-4.3) -- (5.6,-5.7);


\newcommand{\greenslashrect}[5]{%
  \path[fill=white,draw=green!70!black,line width=0.8pt] (#1,#2) rectangle (#3,#4); 
  \path[pattern=north east lines,pattern color=#5,draw=none]
       (#1,#2) rectangle (#3,#4); 
}

\newcommand{\redslashrect}[5]{%
  \path[fill=white,draw=red,line width=0.8pt] (#1,#2) rectangle (#3,#4); 
  \path[pattern=north east lines,pattern color=#5,draw=none]
       (#1,#2) rectangle (#3,#4); 
}

\newcommand{\blueslashrect}[5]{%
  \path[fill=white,draw=blue,line width=0.8pt] (#1,#2) rectangle (#3,#4); 
  \path[pattern=north east lines,pattern color=#5,draw=none]
       (#1,#2) rectangle (#3,#4); 
}

\greenslashrect{2}{7.8}{4}{8.2}{green!70!black}  
\greenslashrect{4}{5.8}{6}{6.2}{green!70!black}  
\greenslashrect{2}{3.8}{4}{4.2}{green!70!black}  
\greenslashrect{4}{1.8}{6}{2.2}{green!70!black}  
\greenslashrect{4}{-0.2}{6}{0.2}{green!70!black} 
\greenslashrect{4}{-2.2}{6}{-1.8}{green!70!black}
\greenslashrect{4}{-4.2}{6}{-3.8}{green!70!black}
\greenslashrect{2}{-6.2}{4}{-5.8}{green!70!black}


\redslashrect{4}{7.8}{6}{8.2}{red}            
\redslashrect{6}{5.8}{8}{6.2}{red}            
\redslashrect{4}{3.8}{6}{4.2}{red}            
\redslashrect{6}{1.8}{8}{2.2}{red}            
\redslashrect{6}{-0.2}{8}{0.2}{red}           
\redslashrect{6}{-2.2}{8}{-1.8}{red}          
\redslashrect{6}{-4.2}{8}{-3.8}{red}          
\redslashrect{4}{-6.2}{6}{-5.8}{red}          


\blueslashrect{6}{7.8}{8}{8.2}{blue}           
\blueslashrect{8}{5.8}{10}{6.2}{blue}          
\blueslashrect{6}{3.8}{8}{4.2}{blue}           
\blueslashrect{8}{1.8}{10}{2.2}{blue}          
\blueslashrect{8}{-0.2}{10}{0.2}{blue}         
\blueslashrect{8}{-2.2}{10}{-1.8}{blue}        
\blueslashrect{8}{-4.2}{10}{-3.8}{blue}         
\blueslashrect{6}{-6.2}{8}{-5.8}{blue}         


\end{tikzpicture}
        \caption{Block-Markov codes for causal encoding noncausal decoding setup with channel feedback.}
        \label{fig: w_f Markov block codes}
    \end{figure}

\textit{Random codebook: }We generate $|\mathcal{M}| = 2^{nR}$ sequences $W_1^n(m)$ i.i.d. $\sim \mathcal{P}_{W_1}$ with index $m\in\mathcal{M}$. For each index $m\in\mathcal{M}$, we generate the same number $|\mathcal{M}| = 2^{nR}$ sequences $U_2^n(m, \hat{m})$ with index $\hat{m}\in\mathcal{M}$ i.i.d. $\sim \mathcal{P}_{U_2 | W_1}$ depending on sequence $W_1^n(m)$.

\textit{Encoding scheme: }Let $m_b$ denote the message the encoder communicates at block $b\in[1:B]$. For the first block, WLOG, the encoder takes $m_1 = 1$ and returns $W_{1,1}^n = W_1^n(m_1)$. Then, at each transmission $b\in[2:B]$, it recalls the index $m_{b-1}$ and finds $m_b\in\mathcal{M}$ such that sequences 
\begin{align*}
    (X_{0,b-1}^n, Y_{1,b-1}^n, W_1^n(m_{b-1}), U_2^n(m_{b-1},m_b))\in\mathcal{A}_\varepsilon^{(n)}.
\end{align*}
It deduces $W_{1,b}^n\triangleq W_1^n(m_b)$ for the current block $b$ and sequentially draws the symbol $U_{1,t,b}$ i.i.d. $\sim\mathcal{P}_{U_1|X_0,W_1}$ conditioning on $W_{1,t}(m_b)$ and $X_{0,b,t}$ received causally at each time instant $t\in[1:n]$.

\textit{Decoding scheme: }The decoder first returns $\Tilde{m}_1 = 1$ for block $b=1$. At block $b\in[2:B]$, it recalls $\Tilde{m}_{b-1}$ from the last block and finds $\Tilde{m}_b\in\mathcal{M}$ such that
\begin{align*}
    &(Y_{1,b}^n, W_1^n(\Tilde{m}_b))\in\mathcal{A}_\varepsilon^{(n)}(\mathcal{P}_{Y_1,W_1}),\\
    &(Y_{1,b-1}^n, W_1^n(\Tilde{m}_{b-1}), U_2^n(\Tilde{m}_{b-1}, \Tilde{m}_b))\in\mathcal{A}_\varepsilon^{(n)}(\mathcal{P}_{Y_1, W_1, U_2}).
\end{align*}
We denote by $\Tilde{W}_{1,b}^n = W_{1}^n(\Tilde{m}_b)$ and $\Tilde{U}_{2,b-1}^n = U_{2}^n(\Tilde{m}_{b-1},\Tilde{m}_b)$ as our choice. For $b\in[1:B-1]$, the returned codeword $\Tilde{U}_{2,b}^n$ serves as the reconstruction for the interim state $X_{1,b}^n$.

\textit{Termination block: }The decoder simply outputs $\Tilde{U}_{2,B}^n = \boldsymbol{0}$. Usually, sequences are not jointly typical in the last block.

\textit{Error analysis per block: }We first focus on the encoding error $\mathcal{E}^e$. For $b\in[2:B]$, let $\mathcal{E}^e_b(m_{b-1})$ denote the event of a failed encoding process during block $b$ given the knowledge of $m_{b-1}$, i.e.,  $\mathcal{E}^e_b(m_{b-1}) =  \{ \forall m_b\in\mathcal{M}, (X^n_{0,b-1}, Y_{1,b-1}^n, W_1^n(m_{b-1}), \\ U_2^n(m_{b-1}, m_b))\notin {\mathcal{A}}_\varepsilon^{(n)}(\mathcal{P}_{X_0,W_1,Y_1,U_2})\}$. Due to the independence between the codewords and Markov blocks, the probability of an encoding error in block $b$ given no encoding errors in the previous blocks is
\begin{equation}
    \mathbb{P}( \mathcal{E}_b^e(M_{b-1})\mid
                \cap_{\beta  = 2}^{b-1}\Bar{\mathcal{E}}_{\beta}^e(M_{\beta - 1}))= \mathbb{P}\brackets{ \mathcal{E}_b^e(M_{b-1})}.\nonumber
\end{equation}
If $R \geq I(U_2; X_0,Y_1|W_1) + \delta$, following the Covering Lemma \cite{elgamal2011nit}, we have  $\mathbb{P}\brackets{ \mathcal{E}_b^{e}(M_{b-1})}\rightarrow 0$ as $n\rightarrow\infty$, $\forall b\in[2:B]$. Thus,  state sequences for the first $B-1$ blocks are successfully encoded with probability $\rightarrow 1$ as $n\rightarrow \infty$. 

Next, we analyze the decoding error $\mathcal{E}^d$. For $b\in[2:B]$, let $\mathcal{E}_b^d(m_{b-1})$ denote the event of a failed decoding process in block $b$ given the past estimated index $m_{b-1}$. Furthermore, let $\mathcal{E}_{b-1}^{d,1}(m_{b-1})$ denote the event that the sequence $Y_{1,b-1}^n$ given $m_{b-1}$ is not jointly typical, i.e., $\mathcal{E}_{b-1}^{d,1}(m_{b-1}) = \{(Y_{1,b-1}^n, W_{1,b-1}^n(m_{b-1}), U_2^n(m_{b-1}, M_b)) \notin \mathcal{A}_{\varepsilon}^{(n)}(\mathcal{P}_{Y_1, W_1, U_2})\}$ and let $\mathcal{E}_b^{d,2}$ denote that the current sequence $Y_{1,b}^n$ is not jointly typical, i.e., 
 $\mathcal{E}_b^{d,2} = \{(Y_{1,b}^n, W_1^n(M_b))\notin\mathcal{A}_\varepsilon^{(n)}(\mathcal{P}_{Y_1,W_1})\}$. Note that, if a sequence with fewer terms $(Y_{1,b}^n,W_1^n(m_b))$ is atypical, it implies that a sequence with more terms $(Y_{1,b}^n,W_1^n(m_b), W_2^n(m_b,m_{b+1}))$ is also atypical. Therefore, we have  $\mathcal{E}_b^{d,2}\subset\mathcal{E}_{b}^{d,1}(M_b)$. Then, the decoding error probability given no past decoding error or encoding error $\mathbb{P}(\mathcal{E}_b^d(m_{b-1})\mid\cap_{\beta=2}^{b-1}\Bar{\mathcal{E}}_{\beta}^d(m_{\beta-1})\cap \Bar{\mathcal{E}}^e)$ can be upperbounded by 
\begin{align}
                &\mathbb{P}(\mathcal{E}_b^d(m_{b-1})\mid \cap_{\beta=2}^{b-1}\Bar{\mathcal{E}}_{\beta}^d(m_{\beta-1})\cap \Bar{\mathcal{E}}^e\cap \Bar{\mathcal{E}}_{b-1}^{d,1}(m_{b-1})\cap\Bar{\mathcal{E}}_b^{d,2}) \nonumber \\
                &+ \mathbb{P}(\mathcal{E}_{b-1}^{d,1}(m_{b-1})\cup\mathcal{E}_b^{d,2}\mid\cap_{\beta=2}^{b-1}\Bar{\mathcal{E}}_{\beta}^d(m_{\beta-1})\cap \Bar{\mathcal{E}}^e)\label{decoding error union bound w-f}
            \end{align}
using the union bound. The first term of \eqref{decoding error union bound w-f} can be upperbounded by
\begin{align*}
                &\mathbb{P}(\mathcal{E}_b^d(m_{b-1})\mid \cap_{\beta=2}^{b-1}\Bar{\mathcal{E}}_{\beta}^d(m_{\beta-1})\cap \Bar{\mathcal{E}}^e\cap \Bar{\mathcal{E}}_{b-1}^{d,1}(m_{b-1})\cap\Bar{\mathcal{E}}_b^{d,2})\\
                & = \mathbb{P}\left( \exists m^\prime\neq M_{b} \text{, s.t. }\{ W_1^n(m^\prime)\in\mathcal{A}_\varepsilon^{(n)}(W_1|Y_{1,b}^n)\}\cap\right. \\
                &\quad\quad\left.\{ U_2^n(M_{b-1}, m^\prime)\in\mathcal{A}_\varepsilon^{(n)}(U_2^n|Y_{1,b-1}^n, W_{1,b-1}^n)\right)\nonumber\\
                & \stackrel{\text{(a)}}{\leq} 2^{-n\varepsilon},
            \end{align*}
where (a) is due to the joint Packing Lemma \cite{elgamal2011nit} in order to satisfy both conditions at the same time given that $R\leq I(W_1;Y_1) + I(U_2;Y_1|W_1) - 7\varepsilon$. Moreover, as for the second term, we have
\begin{align*}
                &\mathbb{P}(\mathcal{E}_{b-1}^{d,1}(m_{b-1})\cup\mathcal{E}_b^{d,2}| \cap_{\beta=2}^{b-1}\Bar{\mathcal{E}}_{\beta}^d(m_{\beta-1})\cap \Bar{\mathcal{E}}^e)\\
                &\stackrel{\text{(b)}}{\leq} \sum_{i=b-1,b}\mathbb{P}((Y_{1,i}^n, W_{1,i}^n,U_{2,i}^n)\notin \mathcal{A}_\varepsilon^{(n)} \mid \\
                &\qquad\qquad\qquad (X_{0,i}^n, Y_{1,i}^n, W_{1,i}^n, U_{2,i}^n)\in\mathcal{B}_\varepsilon^{(n)} )
                \\
                &\xrightarrow[n\rightarrow\infty]{\text{(c)}} 0,
            \end{align*}
where (b) comes from the independence of each Markov block and the fact that $\{(Y_{1,b}^n, W_{1,b}^n)\notin\mathcal{A}_\varepsilon^{(n)}\}\subset \{(Y_{1,b}^n,W_{1,b}^n,U_{2,b}^n)\notin\mathcal{A}_\varepsilon^{(n)}\} $, and (c) comes from that $(Y_{1,b}^n,W_{1,b}^n,U_{2,b}^n)\subseteq(X_{1,b}^n,Y_{1,b}^n,W_{1,b}^n,U_{2,b}^n)$ where the latter sequence is already jointly typical given no encoding errors.

Therefore, by the above arguments, both the encoding and decoding errors vanish asymptotically. Since the generated sequences are jointly typical with high probability, the cost analysis and the existence of a desired control design $c\in\mathcal{C}_{\mathsf{f}}(Bn)$ follow directly from the no-feedback achievability proof detailed in Appendix \ref{app: ach proof}.

\section{Converse Proof of Theorem \ref{theorem: c-n w-f}}\label{app: converse proof c-n w-f}
We consider a control design with feedback $c\in \mathcal{C}_{\mathsf{f}}(n)$ of block length $n\in\mathbb N$ that achieves a pair of costs $(P,S)$. According to Csiszár sum identity, we have
\begin{align*}
        &0= \sum_{t=1}^n I(X_0^{t-1}, Y_1^{t-1}; Y_{1,t}|Y_{1,t+1}^n)  \\
        &\qquad\quad -\sum_{t=1}^n I(Y_{1,t+1}^n; X_{0,t}, Y_{1,t}|X_0^{t-1}, Y_1^{t-1})\\
        &\leq \sum_{t=1}^n I(X_0^{t-1}, Y_1^{t-1}, Y_{1,t+1}^n; Y_{1,t}) \\
        &\qquad\quad -\sum_{t=1}^n I(Y_{1,t+1}^n; X_{0,t}, Y_{1,t}|X_0^{t-1}, Y_1^{t-1})\\
        &= \sum_{t=1}^n \bigg[I(X_0^{t-1}, Y_1^{t-1}; Y_{1,t})+I(Y_{1,t+1}^n; Y_{1,t}|X_0^{t-1}, Y_1^{t-1})\bigg]\\
        &\qquad\quad- \sum_{t=1}^n I(Y_{1,t+1}^n; X_{0,t}, Y_{1,t}|X_0^{t-1}, Y_1^{t-1})\\
        &=\! \sum_{t=1}^n\! I(X_0^{t-1}\!,\! Y_1^{t-1};\! Y_{1,t})\!-\!\! \sum_{t=1}^n\! I(Y_{1,t+1}^n;\! X_{0,t} |X_0^{t-1},\! Y_1^{t-1},\! Y_{1,t})\\
        &\stackrel{\text{(a)}}{=}\!  \sum_{t=1}^n\!  I(X_0^{t-1}\! ,\!  Y_1^{t-1};\!  Y_{1,t})\! -\!\! \!  \sum_{t=1}^n \! \bigg[\! I(Y_{1,t+1}^n ;\!  X_{0,t}  |X_0^{t-1}\! , \! Y_1^{t-1}\! ,\!  Y_{1,t}) \\
        &\qquad\quad  +  I(U_{2,t}; X_{0,t}|X_0^{t-1}, Y_1^{t-1}, Y_{1,t}, Y_{1,t+1}^n)\bigg]\\
        &=\sum_{t=1}^n I(X_0^{t-1}, Y_1^{t-1}; Y_{1,t}) \\
        &\qquad\quad - \sum_{t=1}^n I(U_{2,t}, Y_{1,t+1}^n; X_{0,t} |X_0^{t-1}, Y_1^{t-1}, Y_{1,t})\\
        &\leq\! \sum_{t=1}^n\! I(X_0^{t-1}\!,\! Y_1^{t-1};\! Y_{1,t})\! - \!\sum_{t=1}^n\! I(U_{2,t};\! X_{0,t} |X_0^{t-1}\!, \!Y_1^{t-1}\!,\! Y_{1,t})\\
        &\stackrel{\text{(b)}}{=} \sum_{t=1}^n I(W_{1,t}; Y_{1,t}) - \sum_{t=1}^n I(U_{2,t}; X_{0,t} |W_{1,t}, Y_{1,t}).
    \end{align*}
Here, (a) comes from the result of noncausal decoding, which 
induces the Markov chain $U_2^n-\!\!\!\!\minuso\!\!\!\!- Y_1^n    -\!\!\!\!\minuso\!\!\!\!- X_0^n$, which is equivalent to $I(U_2^n;X_0^n|Y_1^n) = 0$. Therefore, according to the chain rule of mutual information, the term $I(U_{2,t}; X_{0,t}|X_0^{t-1}, Y_1^{t-1}, Y_{1,t}, Y_{1,t+1}^n) = 0$, for each $t\in\{1,...,n\}$. (b) comes from the identification of $W_{1,t} = (X_0^{t-1},Y_1^{t-1})$ for each $t\in[1:n]$. Therefore, the introduction of the auxiliary random variable $W_{1,t}$ gives us the following Markov chains:
\begin{itemize}
    \item $X_{0,t}$ is independent of $W_{1,t} = (X_0^{t-1},Y_1^{t-1})$, due to the i.i.d. source and the memoryless property of the channel.
    \item $(X_{1,t},Y_{1,t})-\!\!\!\!\minuso\!\!\!\!- (U_{1,t}, X_{0,t})    -\!\!\!\!\minuso\!\!\!\!- W_{1,t}$, because of the memoryless channel and the symbolwise generation of $X_1$ in \eqref{eq: X_1 generation} and the memoryless channel given in \eqref{eq: Y_1 generation}. 
    \item  $U_{2,t} -\!\!\!\!\minuso\!\!\!\!- (X_{0,t},Y_{1,t},W_{1,t}) -\!\!\!\!\minuso\!\!\!\!- (U_{1,t},X_{1,t})$ comes from the property of causal encoding function with feedback and the deterministic symbolwise formulation \eqref{eq: Y_1 generation}. This Markov chain is proved in Lemma \ref{lemma: c-n w-f Markov chain} in Appendix \ref{subsec: c-n w-f Markov chain}. 
\end{itemize}

Given the above three Markov chains that holds for each time stage, we now define an independent time-mixing random variable $T\sim\text{Unif}\{1,...,n\}$ such that
    \begin{align*}
    &\sum_{t=1}^n I(W_{1,t}; Y_{1,t}) - \sum_{t=1}^n I(U_{2,t}; X_{0,t} |W_{1,t}, Y_{1,t})\nonumber\\
    &=\!\sum_{t=1}^n\! I(W_{1,t}; Y_{1,t}|T=t)\! -\!\! \sum_{t=1}^n\! I(U_{2,t}; X_{0,t} |W_{1,t}, Y_{1,t}, T=t)\\
    &\stackrel{\text{(c)}}{=}n\cdot \brackets{I(W_{1,T}; Y_{1,T}|T) - I(U_{2,T}; X_{0,T} |W_{1,T}, Y_{1,T}, T)}\\
    &\leq n\cdot \brackets{I(W_{1,T},T; Y_{1,T}) -  I(U_{2,T}; X_{0,T} |W_{1,T}, T, Y_{1,T})}\\
    &\stackrel{\text{(d)}}{=} n\cdot \brackets{I(W_{1}; Y_{1}) -  I(U_{2}; X_{0} |W_{1}, Y_{1})}
    \end{align*}
where (c) comes from that we introduce the random variables of $W_{1,T}, Y_{1,T}, U_{2,T}, X_{0,T}$, 
and (d) comes from the identification of $W_1\triangleq (W_{1,T},T),Y_1\triangleq Y_{1,T}, U_2\triangleq U_{2,T}, X_0\triangleq X_{0,T}$.
Therefore, these newly introduced random variables satisfy
\begin{itemize}
    \item $X_0\indep W_1$.
    \item $(X_{1},Y_{1})-\!\!\!\!\minuso\!\!\!\!-(X_{0},U_{1})-\!\!\!\!\minuso\!\!\!\!-W_1$.
    \item $U_2-\!\!\!\!\minuso\!\!\!\!-(X_0,Y_1,W_1)-\!\!\!\!\minuso\!\!\!\!-(U_1,X_1)$.
\end{itemize}
Hence, the joint distribution over the random variables $(X_0,W_1,U_1,X_1,Y_1,U_2)$ decomposes of the form of \eqref{eq: c-n w-f prob result} and satisfies the information constraint \eqref{eq: c-n w-f info result}. Now we reformulate the $n$-stage long-run cost
 \begin{align}
        \gamma_p^n(c) = \mathbb E\sbrackets{\frac{1}{n}\sum_{t=1}^nU_{1,t}^2}=\mathbb E[U_{1}^2]     
    \end{align}
And also, 
\begin{align}
     \gamma_s^n(c)= \mathbb E\sbrackets{\frac{1}{n}\sum_{t=1}^n(X_{1,t}-U_{2,t})^2}= \mathbb E[(X_1 - U_2)^2]
\end{align}

Moreover, from \eqref{eq: achievable cost def}, we also have
\begin{align}
&\mathbb E\Big[\big|P - c_P(U_1^n)\big| + \big|S - c_S(X_1^n, U_2^n)\big|\Big] \leq \varepsilon\nonumber
    \\
     &\Longrightarrow|P - \mathbb E[c_P(U_1^n)]| + |S - \mathbb E[c_S(X_1^n,U_2^n)]| \leq \varepsilon\nonumber
     \\
     &\Longrightarrow |P - \mathbb E \sbrackets{U_1^2}| + |S - \mathbb E \sbrackets{(X_{1} - U_{2})^2}| \leq \varepsilon.
     \label{achievable result}
\end{align}
This is valid for all $\varepsilon>0$. This shows that the sequence $(X_0,W_1,U_1,X_1,Y_1,U_2)$ satisfy \eqref{eq: c-n w-f prob result}, \eqref{eq: c-n w-f info result}, \eqref{eq: c-n w-f cost result}.

\subsection{Proof of the Markov chain}\label{subsec: c-n w-f Markov chain}
\begin{lemma}\label{lemma: c-n w-f Markov chain}
    For every $t\in\{1,...,n\}$, it holds that $$U_{2,t} -\!\!\!\!\minuso\!\!\!\!- (X_{0,t},Y_{1,t},W_{1,t}) -\!\!\!\!\minuso\!\!\!\!- (U_{1,t},X_{1,t}).$$
\end{lemma}
\begin{proof}[Proof of Lemma \ref{lemma: c-n w-f Markov chain}]
    \begin{align*}
        &\mathbb P(U_{2,t}|X_{0,t},Y_{1,t},W_{1,t},U_{1,t},X_{1,t})\\
        &=\mathbb P(U_{2,t}|X_{0}^t,Y_{1}^t,U_{1,t},X_{1,t})\\
        &\stackrel{\text{(e)}}{=}\sum\mathbb P(X_{0,t+1}^n,U_{1,t+1}^n,Y_{1,t+1}^n,U_{2,t}|X_{0}^t, Y_{1}^t,U_{1,t},X_{1,t})\\
        &=\sum \mathbb P(X_{0,t+1}^n|X_{0}^t,Y_{1}^t,U_{1,t},X_{1,t})\tag{\text{f1}}\\
        &\qquad \times \mathbb P(U_{1,t+1}^n,Y_{1,t+1}^n|X_{0}^n,Y_{1}^t,U_{1,t},X_{1,t})\tag{\text{g1}}\\
        &\qquad \times \mathbb P(U_{2,t}|X_{0}^n,Y_{1}^t,U_{1,t},X_{1,t},U_{1,t+1}^n,Y_{1,t+1}^n),\tag{\text{h1}}
    \end{align*}
where in equality (e) the sum is over the sequences $(X_{0,t+1}^n,U_{1,t+1}^n,Y_{1,t+1}^n)$. Now, each probability term above
\begin{itemize}
    \item  $\text{(f1)}= \mathbb P(X_{0,t+1}^n|X_{0}^t,Y_{1}^t)$, because the source is generated i.i.d., and the encoder is causal.
    \item (g1) decomposes as follows
    \begin{align*}
    &\mathbb P(U_{1,t+1}^n,Y_{1,t+1}^n|X_{0}^n,Y_{1}^t,U_{1,t},X_{1,t})\\
        &=\prod_{i=t+1}^n\mathbb P(U_{1,i},Y_{1,i}|U_{1,t+1}^{i-1}, Y_{1,t+1}^{i-1},X_{0}^n,Y_{1}^t,U_{1,t},X_{1,t})\\
        &=\prod_{i=t+1}^n\mathbb P(U_{1,i}|U_{1,t+1}^{i-1}, Y_{1,t+1}^{i-1},X_{0}^n,Y_{1}^t,U_{1,t},X_{1,t})\tag{\text{i1}}\\
        &\quad \times\mathbb P(Y_{1,i}|U_{1,t+1}^{i-1}, Y_{1,t+1}^{i-1},X_{0}^n,Y_{1}^t,U_{1,t},X_{1,t},U_{1,i})\tag{\text{j1}}\\
        &=\prod_{i=t+1}^n\mathbb P(U_{1,i}|U_{1,t+1}^{i-1}, Y_{1,t+1}^{i-1},X_{0}^n,Y_{1}^t)\tag{\text{i2}}\\
        &\quad\times\mathbb P(Y_{1,i}|U_{1,t+1}^{i-1}, Y_{1,t+1}^{i-1},X_{0}^n,Y_{1}^t,U_{1,i})\tag{\text{j2}}\\
        &=\mathbb P(U_{1,t+1}^n,Y_{1,t+1}^n|X_{0}^n,Y_{1}^t)
    \end{align*}
    where (i1) $\rightarrow$ (i2) comes from the causal encoding function $f^{(\mathsf{f},i)}_{U_{1,i}|X_0^i,Y_1^{i-1}}$, and (j1) $\rightarrow$ (j2) is because of the memoryless channel $Y_{1,i} = X_{0,i}+U_{1,i} + Z_i$.
    \item (h1) $=\mathbb P(U_{2,t}|X_{0}^n,Y_{1}^t,U_{1,t+1}^n,Y_{1,t+1}^n)$ comes from the noncausal decoder $g_{U_{2}^n|Y_1^n}$. Hence $U_{2,t}$ is fully determined by $(Y_{1}^t,Y_{1,t+1}^n)$.
\end{itemize}
Therefore, we obtain that 
\begin{align*}
    &\mathbb P(U_{2,t}|X_{0,t},Y_{1,t},W_{1,t},U_{1,t},X_{1,t})\\
    &=\sum_{X_{0,t+1}^n,U_{1,t+1}^n,Y_{1,t+1}^n} \mathbb P(X_{0,t+1}^n|X_{0}^t,Y_{1}^t)\\
        &\qquad \times \mathbb P(U_{1,t+1}^n,Y_{1,t+1}^n|X_{0}^n,Y_{1}^t)\\
        &\qquad \times \mathbb P(U_{2,t}|X_{0}^n,Y_{1}^t,U_{1,t+1}^n,Y_{1,t+1}^n)\\
        &=\mathbb P(U_{2,t}|X_{0,t},Y_{1,t},W_{1,t}).
\end{align*}
\end{proof}

\section{Proof of Theorem \ref{thm: ZEC feedback cost}}\label{app: proof of ZEC feedback info constraint}
The power consumption induced by the scheme \eqref{eq: ZEC feedback scheme} is
\begin{align*}
    P = \mathbb E[U_1^2] &= \mathbb E[(W_1 + a\cdot \text{sign}(X_0)+bX_0)^2]\\
   &=\mathbb E[(W_1)^2]+\mathbb E[(a\cdot \text{sign}(X_0)+ bX_0)^2]  \\
   &= V_1 + a^2 + b^2 Q + 2ab\cdot\mathbb E[|X_0|]\\
   &=  V_1 + a^2 + b^2 Q + 2ab\sqrt{\frac{2Q}{\pi}}.
\end{align*}

The information constraint \eqref{eq: c-n w-f info result} becomes
\begin{align}
    &I(W_1; Y_1) - I(U_2; X_0 | W_1,Y_1)\nonumber\\
    & = h(Y_1) - h(Y_1|W_1) - h(U_2|W_1,Y_1) + h(U_2|X_0,W_1,Y_1)\nonumber\\
    &=h(Y_1) - h(Y_1|W_1) - h(X_1|W_1,Y_1) + \underbrace{h(X_1|X_0,W_1,Y_1)}_{=0}\nonumber\\
    &=h(Y_1) - h(X_1,Y_1|W_1)\nonumber\\
    &= h(Y_1)-h(X_1|W_1)-h(Y_1|X_1,W_1)\nonumber\\
    &= h(Y_1)-h(X_1|W_1)-h(Z_1|X_1,W_1)\nonumber\\
    &= h(Y_1)-h(X_1|W_1)-h(Z_1)\nonumber\\
    &= h(Y_1)-h(X_1|W_1)-\frac{1}{2}\log 2\pi eN,\label{eq: info constraint ZEC with feedback}
\end{align}
where in the second equality, we replace $U_2$ by $X_1$ given the fact that \eqref{eq: U_2=X_1}. Now, we examine the (conditional) PDFs of $Y_1$ and $X_1|W_1$ respectively.

Given a fixed value of $X_0 = x$, the conditional distribution $Y_1|X_0=x\sim\mathcal{N}(a\cdot\text{sign}(x) + \delta x, V_1+N)$, where we denote $\delta \triangleq b+1$ for simplicity. Hence,
\begin{align}
   f_{Y_1}(y) &= \int_{-\infty}^\infty f_{Y_1|X_0}(y|x)f_{X_0}(x) dx    \nonumber\\
   &= \int_{-\infty}^0 f_{Y_1|X_0}(y|x)f_{X_0}(x) dx \nonumber\\
   &\qquad + \int_{0}^\infty f_{Y_1|X_0}(y|x)f_{X_0}(x) dx\nonumber\\
    &=\frac{1}{2}\mathcal{SN}(a, \sigma, \alpha)+ \frac{1}{2}\mathcal{SN}(-a, \sigma, -\alpha)\label{eq: pdf Y_1, gaussian W_1}
\end{align}
is a mixture of two skew-normal distributions, where $\sigma^2 = V_1 + \delta^2Q + N, \alpha = \frac{\sqrt{Q}\delta}{\sqrt{V_1+N}}$.

On the other hand, given that $X_0\indep W_1$, we have
\begin{align*}
    &h(X_1|W_1)\\
    &= h(W_1 + a\cdot \text{sign}(X_0) + \delta X_0|W_1)\\
    &= h( a\cdot \text{sign}(X_0) + \delta X_0|W_1)\\
    &= h( a\cdot \text{sign}(X_0) + \delta X_0).
\end{align*}
Therefore, by denoting $X \triangleq a\cdot \text{sign}(X_0) + \delta X_0$, we obtain that
\begin{align*}
     f_X(x) &= \frac{1}{\sqrt{\delta^2Q}}\phi\brackets{\frac{x-a}{\sqrt{\delta^2Q}}}\mathbf{1}_{\text{sign}(\delta)\cdot (x-a)\geq 0} \\
    &\qquad + \frac{1}{\sqrt{\delta^2Q}}\phi\brackets{\frac{x+a}{\sqrt{\delta^2Q}}}\mathbf{1}_{\text{sign}(\delta)\cdot (x+a)\leq 0},
\end{align*}
which concludes our proof for Theorem \ref{thm: ZEC feedback cost}.\hfill \qedsymbol{}

\section{Proof of Corollary \ref{cor: ZECZEP}}\label{app: proof of cor ZECZEP}
We look at the special case of $V_1 = 0, a = 0, b = 0$. In this way, the ZEC-f \mz{strategy} \eqref{eq: ZEC feedback scheme} boils down to
\begin{align*}
     U_1 &= 0,\\
     X_1 &= X_0,\\
     Y_1 &= X_0+Z_1.
\end{align*}

The power cost now becomes
\begin{align*}
    P = \mathbb E[U_1^2] = 0.
\end{align*}

Moreover, the information constraint that we derived in \eqref{eq: info constraint ZEC-f} now becomes
\begin{align*}
    &h(X_0+Z_1) - h(X_0) - \frac{1}{2}\log 2\pi eN\\
    &=\frac{1}{2}\log 2\pi e(Q+N) - \frac{1}{2}\log 2\pi eQ - \frac{1}{2}\log 2\pi eN\\
    &=\frac{1}{2}\log\frac{Q+N}{2\pi e QN}.
\end{align*}
In conclusion, the information constraint \eqref{eq: c-n w-f info result} is nonnegative if and only if \eqref{eq:ZECZEP} is satisfied. \hfill\qedsymbol{}

\section{Proof of Theorem \ref{thm: best gaussian policy}}\label{app: proof of gaussian design}
To prove this theorem, we first provide the following lemma which states the general relation of Gaussian covariance coefficients given a Markov chain. It is a direct consequence of combining several well-known results.

\begin{lemma}\label{lemma: markov-covariance}
    If the jointly Gaussian random vector $(X,Y,Z)$ satisfy the Markov chain $X -\!\!\!\!\minuso\!\!\!\!- Y -\!\!\!\!\minuso\!\!\!\!- Z$ and have a covariance matrix
    \begin{align}
        \Sigma_{X,Y,Z} = \begin{pmatrix}
            P & \rho_1\sqrt{PQ} & \rho_2\sqrt{PV}\\
            \rho_1\sqrt{PQ} & Q &\rho_3\sqrt{QV}\\
            \rho_2\sqrt{PV}&\rho_3\sqrt{QV} & V
        \end{pmatrix},\label{eq: cov xyz}
    \end{align}
    with the covariance coefficients $(\rho_1,\rho_2,\rho_3)\in[-1,1]^3$ ensuring that $\det\brackets{\Sigma_{X,Y,Z}}\geq 0$, then, we have
    \begin{align}
        \rho_2=\rho_1\rho_3.\label{eq: rho_2=rho_1rho_3}
    \end{align}
    
\end{lemma}

In other words, with the context of the Markov chain of jointly Gaussian $X -\!\!\!\!\minuso\!\!\!\!- Y -\!\!\!\!\minuso\!\!\!\!- Z$, if $X$ and $Y$ (or if $Z$ and $Y$) are uncorrelated (i.e., if $\rho_1=0$ or $\rho_3=0$), it follows that $X$ and $Z$ are also uncorrelated (i.e., $\rho_2=0$). The proof of this lemma is shown in Appendix \ref{app: proof of markov-covariance lemma}.

Now, without loss of generality, we consider the joint Gaussian random variables $(X_0, W_1,W_2, U_1)\sim\mathcal{N}(0, K)$ optimal for problem \eqref{eq: gauss opt problem}, are centered with the covariance matrix 
\begin{align}
    K= \begin{pmatrix}
        Q & \rho_1\sqrt{QV_1} & \rho_2\sqrt{QV_2} & \rho_3\sqrt{QP}\\
         \rho_1\sqrt{QV_1} & V_1 &\rho_4\sqrt{V_1V_2} & \rho_5\sqrt{V_1P}\\
         \rho_2\sqrt{QV_2}&\rho_4\sqrt{V_1V_2}&V_2&\rho_6\sqrt{V_2P} \\
         \rho_3\sqrt{QP} & \rho_5\sqrt{V_1P} & \rho_6\sqrt{V_2P} & P
    \end{pmatrix}.
\end{align}
Since $X_0\indep W_1$, $\rho_1 = 0$. Also, given the Markov chain $U_1 -\!\!\!\!\minuso\!\!\!\!- (X_0, W_1) -\!\!\!\!\minuso\!\!\!\!- W_2$, from a simple extension of Lemma \ref{lemma: markov-covariance}, we can obtain that $\rho_6 $ is uniquely determined through the other four parameters by $\rho_6 =\rho_2\rho_3 + \rho_4\rho_5$. Moreover, other active correlation coefficients $(\rho_2, \rho_3, \rho_4, \rho_5)\in[-1,1]^4$ are chosen such that 
\begin{align*}
   \det(K) =QV_1V_2P(-1 + \rho_2^2 + \rho_4^2)(-1 + \rho_3^2 + \rho_5^2)\geq 0.
\end{align*}

Given covariance matrix $K$ and \eqref{eq: Y_1 generation}, the covariance matrix $K_2$ of $(X_0,W_1,W_2,Y_1)$ could be easily computed, with a determinant given by
\begin{align*}
    \det(K_2) = QV_1V_2(-1 + \rho_2^2 + \rho_4^2)(P(-1 + \rho_3^2 + \rho_5^2) - N).
\end{align*}
The positive semi-definite property of $K_2$ must also be satisfied with properly chosen $(\rho_2, \rho_3, \rho_4, \rho_5)$.

We have the following lemma determining the explicit expression of the information constraint \eqref{info result} and the optimization object \eqref{eq: gauss opt problem} under the joint Gaussian constraint.

\begin{lemma}\label{lemma, eqs to show}
    Assume $(X_0, W_1,W_2, U_1)\sim\mathcal{N}(0, K)$, then
    \begin{align}
        &I(W_1, W_2; Y_1) - I(W_2; X_0|W_1)\nonumber\\
        &= I(W_1; Y_1) - I(W_2; X_0|W_1, Y_1)\label{eq: info constraint, w/o feedback}\\
        &= \frac{1}{2}\log\brackets{\frac{T_1}{T_1 - T_2}},\label{eq: new info constraint, w/o feedback}
    \end{align}
    where the terms 
    \begin{align}
        &T_1 = (P+Q+N+2\rho_3\sqrt{QP})(-1 + \rho_2^2 + \rho_4^2),\nonumber\\
        &T_2 = N\rho_2^2 + P\rho_2^2(1-\rho_3^2) - P\rho_5^2(1-\rho_4^2).\nonumber
    \end{align}
    And the object to minimize writes
    \begin{align}
        &\mathbb E\Big[\big(X_1 - \mathbb E\big[X_1\big|W_1,W_2,Y_1\big]\big)^2\Big]\nonumber\\
        &=\frac{N\cdot f_1(\rho_2,\rho_3,\rho_4,\rho_5)}{(1-\rho_4^2)\cdot N + f_1(\rho_2,\rho_3,\rho_4,\rho_5)},\label{eq: conditional variance}
    \end{align}
    where
    \begin{align*}
        &f_1(\rho_2,\rho_3,\rho_4,\rho_5) \\
        &= -P\rho_2^2\rho_3^2- (Q+2\rho_3\sqrt{PQ})(-1+\rho_2^2+\rho_4^2)\\
        &\qquad + P(1-\rho_4^2)(1-\rho_5^2).
    \end{align*}
\end{lemma}

The proof of this lemma is given in Appendix \ref{app: proof of lemma eqs to show}

Thus, from \eqref{eq: new info constraint, w/o feedback}, the original information constraint is given explicitly by
\begin{align*}
    \frac{1}{2}\log\brackets{\frac{T_1}{T_1 - T_2}}\geq 0\Leftrightarrow T_1\geq T_2\geq 0 \text{ or } T_1\leq T_2\leq 0.
\end{align*}

Now, we examine the objective \eqref{eq: conditional variance}. First, if $1-\rho_4^2\ = 0$, from \eqref{eq: conditional variance}, we can get that 
\begin{align*}
    \mathbb E\Big[\big(X_1 - \mathbb E\big[X_1\big|W_1,W_2,Y_1\big]\big)^2\Big]=N.
\end{align*} 

Next, we focus on the case of $1-\rho_4^2\ \neq 0$. In this case, \eqref{eq: conditional variance} is of the form 
\begin{align}
    \frac{N\cdot f(\rho_2,\rho_3,\rho_4,\rho_5)}{N + f(\rho_2,\rho_3,\rho_4,\rho_5)}, \label{eq: Nf/(N+f)}
\end{align} 
where
\begin{align*}
    f(\rho_2,\rho_3,\rho_4,\rho_5) = f_1(\rho_2,\rho_3,\rho_4,\rho_5)/(1-\rho_4^2).
\end{align*}

Note that, the function $x\mapsto \frac{N\cdot x}{N + x}$ is nonnegative and strictly increasing over the region $(-\infty, -N]\cup [0, \infty)$. Therefore, our goal of minimizing \eqref{eq: Nf/(N+f)} is now transformed to either minimizing the nonnegative object 
\begin{align*}
 f(\rho_2,\rho_3,\rho_4,\rho_5)\geq 0
\end{align*} or minimizing the negative object
\begin{align*}
   f(\rho_2,\rho_3,\rho_4,\rho_5)\leq -N
\end{align*} subject to the following constraints:
 \begin{align}
        &\text{1.  }\det(K)\geq 0 \Longrightarrow\nonumber\\
        &\quad QV_1V_2P(-1 + \rho_2^2 + \rho_4^2)(-1 + \rho_3^2 + \rho_5^2)
        \stackrel{\text{(A)}}{\geq} 0,\nonumber\\
        &\text{2.  }\det(K_2)\geq 0 \Longrightarrow\nonumber\\
         &\quad QV_1V_2(-1 + \rho_2^2 + \rho_4^2)(P(-1 + \rho_3^2 + \rho_5^2) - N)
        \stackrel{\text{(B)}}{\geq} 0,\nonumber\\
        &\text{3.  }T_1\geq T_2 \geq 0\Longrightarrow\nonumber\\&\quad(Q+P+N+2\rho_3\sqrt{QP})(-1  + \rho_2^2 + \rho_4^2) \nonumber\\
        &\quad\stackrel{\text{(C1)}}{\geq} N\rho_2^2 + P\rho_2^2(1-\rho_3^2) - P\rho_5^2(1-\rho_4^2)\stackrel{\text{(D1)}}{\geq}0,\nonumber\\
        &\quad\text{or,   }T_1\leq T_2 \leq 0\Longrightarrow\nonumber\\&\quad(Q+P+N+2\rho_3\sqrt{QP})(-1  + \rho_2^2 + \rho_4^2) \nonumber\\
        &\quad\stackrel{\text{(C2)}}{\leq}  N\rho_2^2+ P\rho_2^2(1-\rho_3^2) - P\rho_5^2(1-\rho_4^2)\stackrel{\text{(D2)}}{\leq}0.\nonumber
    \end{align}

To simplify the above constraints, we consider the following two distinct cases:

\textit{Case 1, }if $-1 + \rho_2^2 + \rho_4^2\geq 0$, constraints (A) and (B) together yield $-1 + \rho_3^2 + \rho_5^2\geq \frac{N}{P}$. Moreover, constraint (C1) gives us $f(\rho_2,\rho_3,\rho_4,\rho_5) \leq -N$. In this case, our optimization problem boils down to minimizing
\begin{align}
    f(\rho_2,\rho_3,\rho_4,\rho_5)\leq -N,\label{eq: opt object, case 1}
\end{align}
subject to 
\begin{align}
    &\text{1.  }-1 + \rho_2^2 + \rho_4^2\geq 0,\nonumber
    \\
    &\text{2.  }-1 + \rho_3^2 + \rho_5^2\geq \frac{N}{P},\nonumber
    \\
    &\text{3.  } N\rho_2^2 + P\rho_2^2(1-\rho_3^2) - P\rho_5^2(1-\rho_4^2)\geq0.\label{eq: case 1, constraint 3}
\end{align}
Notice that $f(\rho_2,\rho_3,\rho_4,\rho_5)$ is decreasing function of $\rho_5^2$. From \eqref{eq: case 1, constraint 3}, we get that $\rho_5^2 \leq \frac{N\rho_2^2 + P\rho_2^2(1-\rho_3^2)}{P(1-\rho_4^2)}$. Therefore, the optimizer is given by $(\rho_5^*)^2 = \frac{N\rho_2^2 + P\rho_2^2(1-\rho_3^2)}{P(1-\rho_4^2)}$. By plugging $(\rho_5^*)^2$ into \eqref{eq: opt object, case 1}, we obtain that
\begin{align*}
    &f(\rho_2,\rho_3,\rho_4,\rho_5^*) \\
    &= \frac{-(-1+\rho_2^2+\rho_4^2)(Q+P+2\rho_3\sqrt{PQ}) - N\rho_2^2}{1-\rho_4^2}.
\end{align*}
Since
\begin{align*}
    \frac{\partial f(\rho_2,\rho_3,\rho_4,\rho_5^*)}{\partial \rho_4^2} \leq 0,
\end{align*}
we know that $f\rightarrow -\infty$ as $\rho_4^2\rightarrow 1$. Moreover, since $\mathbb E\big[\big(X_1 - \mathbb E\big[X_1\big|W_1,W_2,Y_1\big]\big)^2\big]$  is continuous and converges to $N$ when $\rho_4^2\rightarrow 1$, in this case, the minimal value of $N$ is obtained at the boundary $\rho_4^2 = 1$.

\textit{Case 2, }If $-1 + \rho_2^2 + \rho_4^2\leq 0$, conditions (A) and (B) together give us $-1 + \rho_3^2 + \rho_5^2\leq 0$, and (C2) gives us $f(\rho_2,\rho_3,\rho_4,\rho_5) \geq -N$ (but only $f\geq 0$ contributes to a nonnegative estimation cost). Therefore, in this case, our optimization problem boils down to minimizing 
\begin{align}
    f(\rho_2,\rho_3,\rho_4,\rho_5) \geq0,\label{eq: opt object, case 2}
\end{align}
subject to
\begin{align}
        &\text{1.  }-1 + \rho_2^2 + \rho_4^2\leq 0,\nonumber
        \\
        &\text{2.  }-1 + \rho_3^2 + \rho_5^2\leq 0,\label{eq: case 2, constraint 2}\\
        &\text{3.  }N\rho_2^2+ P\rho_2^2(1-\rho_3^2) - P\rho_5^2(1-\rho_4^2)\leq 0. \label{eq: case 2, constraint 3}
\end{align}
Since $f(\rho_2,\rho_3,\rho_4,\rho_5)$ is reduced especially when $\rho_5^2$ is increased, therefore, from \eqref{eq: case 2, constraint 2}, we get that $(\rho_5^*)^2 = 1-\rho_3^2$. By replacing $(\rho_5^*)^2$ into \eqref{eq: opt object, case 2}, we get that
\begin{align}    f(\rho_2,\rho_3,\rho_4,\rho_5^*) =\frac{ (1 - \rho_2^2 - \rho_4^2)(\sqrt{P}\rho_3 + \sqrt{Q})^2 }{(1 -\rho_4^2)}.\label{eq: opt object 2}
\end{align}
Therefore, when $P\geq Q$, taking $\rho_3^* = -\sqrt{\frac{Q}{P}}$ and any $\rho_2$, $\rho_4$ satisfy the constraints results in the optimal value of  $f(\rho_2,\rho_3^*,\rho_4,\rho_5^*) = 0$. In this case, $\mathbb E\big[\big(X_1 - \mathbb E\big[X_1\big|W_1,W_2,Y_1\big]\big)^2\big] = 0$.

When $P<Q$, $f(\rho_2,\rho_3,\rho_4,\rho_5^*)$ is a decreasing function of $\rho_2^2$. The constraint \eqref{eq: case 2, constraint 3} gives us the optimal value of $\rho_2^2$:
\begin{equation}
    (\rho_2^*)^2 = \frac{P(1-\rho_3^2)(1-\rho_4^2)}{N+P(1-\rho_3^2)}.\label{opt rho_2^2}
\end{equation}
Plugging the $(\rho_2^*)^2$ in \eqref{opt rho_2^2} into 
\eqref{eq: opt object 2}, we have
\begin{align*}
    f(\rho_2^*,\rho_3,\rho_4,\rho_5^*)  = \frac{N (\sqrt{P}\rho_3 + \sqrt{Q})^2}{N+P(1-\rho_3^2)}.
\end{align*}
Then, taking $\frac{\partial f}{\partial \rho_3} =0$ gives us the optimum 
\begin{align*}
    \rho_3^* = -\frac{P+N}{\sqrt{QP}},
\end{align*}
which is valid only if the following condition holds
\begin{align}
    (\rho_3^*)^2 = \frac{(P+N)^2}{QP}\leq 1
    \Longrightarrow \quad\left\{
       \begin{aligned}
           & Q > 4N,\\
           & P\in\sbrackets{P_1, P_2},
       \end{aligned}
       \right. \label{eq: opt condition, case 2}
\end{align}
where
\begin{align*}
    &P_1 = \frac{1}{2}\brackets{Q-2N-\sqrt{Q^2-4QN}},\\
    &P_2 = \frac{1}{2}\brackets{Q-2N+\sqrt{Q^2-4QN}}.
\end{align*}
In this case, $(\rho_5^*)^2 = \frac{QP - (P+N)^2}{QP}$, and $f(\rho_2^*,\rho_3^*,\rho_4^*,\rho_5^*) = \frac{N(Q - N - P)}{N+P}$, which results in the estimation cost of
\begin{align*}
    \mathbb E\Big[\big(X_1 - \mathbb E\big[X_1\big|W_1,W_2,Y_1\big]\big)^2\Big] = \frac{N(Q-N-P)}{Q}.
\end{align*}

In the case when the condition \eqref{eq: opt condition, case 2} is unmet, we always have $-\frac{P+N}{\sqrt{QP}} <-1$. Since $\frac{\partial f}{\partial \rho_3}>0$, function $f$ increases when $\rho_3\in[-1,1]$ increases. Therefore, the minimal value of $f$ achieves at the left boundary $\rho_3^* = -1$, which gives us $\rho_2^* = \rho_5^* = 0$ and $f(\rho_2^*,\rho_3^*,\rho_4^*,\rho_5^*) = (\sqrt{Q} - \sqrt{P})^2$. Hence, 
\begin{align*}
    \mathbb E\Big[(X_1 - \mathbb E\big[X_1\big|W_1,W_2,Y_1\big]\big)^2\Big] = \frac{N\cdot(\sqrt{Q} - \sqrt{P})^2}{N+(\sqrt{Q} - \sqrt{P})^2}.
\end{align*}

Obviously, the mimimal estimation cost value $N$ from case 1 is always larger than the cost derived in case 2. Therefore, summarizing our above analysis, the optimal Gaussian cost $S_{\mathsf {G}}(P)$ is given by \eqref{opt gaussian cost}. \hfill \qedsymbol{}

\subsection{Proof of Lemma \ref{lemma: markov-covariance}}\label{app: proof of markov-covariance lemma}
From the Markov chain and joint Gaussian entropy, we have that
\begin{align*}
    0&=I(X;Z|Y)\\
    &= H(X, Y) + H(Y,Z) - H(Y) - H(X,Y,Z)\\
    &= \frac{1}{2}\log\brackets{\frac{\det(\Sigma_{X,Y})\cdot\det(\Sigma_{Y,Z})}{\sigma_Y^2\cdot\det(\Sigma_{X,Y,Z})}},
\end{align*}
where all the information of the last step can be obtained from the covariance matrix \eqref{eq: cov xyz}. Therefore,
\begin{align*}
&0=\det(\Sigma_{X,Y})\cdot\det(\Sigma_{Y,Z}) - \sigma_Y^2\cdot\det(\Sigma_{X,Y,Z})\\
&= PQ^2V(\rho_1\rho_3 - \rho_2)^2,
\end{align*}
which implies the result given in \eqref{eq: rho_2=rho_1rho_3}. 
\hfill \qedsymbol{}
\subsection{Proof of Lemma \ref{lemma, eqs to show}}\label{app: proof of lemma eqs to show}
The information constraint \eqref{info result} can be rewritten in the following way:
\begin{align*}
    &\quad I(W_1, W_2; Y_1) - I(W_2; X_0|W_1)\\
    &= I(W_1; Y_1) + I(W_2; Y_1|W_1) - I(W_2; X_0|W_1)\\
    &\stackrel{\text{(a)}}{=} I(W_1; Y_1) + I(W_2; Y_1|W_1) - I(W_2; X_0, Y_1|W_1)\\
    &= I(W_1; Y_1) - I(W_2; X_0|W_1, Y_1)\\
    & = \frac{1}{2}\log\brackets{\frac{\sigma_{W_1}^2\cdot\sigma_{Y_1}^2\cdot\det{(K_2)}}{\det{(\Sigma_{W_1,W_2,Y_1})}\cdot\det{(\Sigma_{X_0,W_1,Y_1})}}}\\
    & \stackrel{}{=} \frac{1}{2}\log\brackets{\frac{T_1}{T_1 - T_2}},
\end{align*}
where (a) comes from the Markov chain $Y_1 -\!\!\!\!\minuso\!\!\!\!- (X_0, W_1) -\!\!\!\!\minuso\!\!\!\!- W_2$, and thus $I(W_2; Y_1|X_0, W_1)=0$. Additionally,
\begin{align}
    &\mathbb E\Big[\big(X_1 - \mathbb E\big[X_1\big|W_1,W_2,Y_1\big]\big)^2\Big]\nonumber\\
    &= \Var\brackets{X_1|W_1, W_2, Y_1}\nonumber\\
    &\stackrel{\text{(b)}}{=} \sigma_{X_1}^2 - \Sigma_{X_1W}^\top\Sigma_{WW}^{-1}\Sigma_{X_1W}\nonumber\\
    &= \frac{N\cdot f_1(\rho_2,\rho_3,\rho_4,\rho_5)}{(1-\rho_4^2)\cdot N + f_1(\rho_2,\rho_3,\rho_4,\rho_5)},\nonumber
\end{align}
where step (b) is obtained using the Schur complement. Here, $W = (W_1,W_2,Y_1)^\top$, $\Sigma_{X_1W} = (\sigma_{X_1,W_1}, \sigma_{X_1,W_2}, \sigma_{X_1,Y_1})^\top$, and $\Sigma_{WW}$ is the covariance matrix of $(W_1,W_2,Y_1)$. \hfill \qedsymbol{}

\section{Proof of Corollary \ref{cor: causal gaussian cost}}\label{app: proof of unified causal Gaussian cost}
From Theorem \ref{theorem: c-n w-f}, the information constraint for Witsenhausen counterexample with causal encoder and non-causal decoder with channel feedback $Y_1$ is
\begin{equation}
    I(W_1; Y_1) - I(U_2; X_0|W_1, Y_1)\geq 0. \label{eq: info result c-n w-f recall}
\end{equation}
Under the joint Gaussian constraint, the MMSE estimator of $X_1$ is linear, i.e.,
\begin{align*}
    U_2 = \mathbb E\big[X_1\big|W_1,W_2,Y_1\big]=a\cdot W_1 + b\cdot W_2+c\cdot Y_1,
\end{align*}
for some constants $a,b,c\in\mathbb R$. Substituting this into \eqref{eq: info result c-n w-f recall} yields
    \begin{align*}
        &I(W_1; Y_1) - I(U_2; X_0|W_1, Y_1)\\
        &=I(W_1; Y_1) - I(a\cdot W_1 + b\cdot W_2+c\cdot Y_1; X_0|W_1, Y_1)\\
        &=I(W_1; Y_1) - I(W_2; X_0|W_1, Y_1)\\
        &=I(W_1, W_2; Y_1) - I(W_2; X_0|W_1),
    \end{align*}
where in the last equality we use the step \eqref{eq: info constraint, w/o feedback} in Lemma \ref{lemma, eqs to show}. Hence, this recovers the information constraint for the framework without channel feedback \eqref{info result}. Thus, the optimization domain for this optimization problem remains unchanged, and the optimal joint Gaussian cost is \eqref{opt gaussian cost}, establishing case 1) in the corollary.

As for 2), \cite{Treust2024power} showed that the optimal Gaussian cost for noncausal encoding with causal decoding, both with and without source feedforward, also coincides with \eqref{opt gaussian cost}.

Finally, for case 3), note that the Gaussian strategy \eqref{opt gaussian cost} is simply a time-sharing mechanism of the best linear scheme. This applies even when both DMs are causal (achievability). \mz{Conversely}, by cases 1) and 2), \eqref{opt gaussian cost} also serves as a lowerbound of this causal Gaussian cost regardless of feedback or feedforward links (converse). 

To conclude, we have established that all the causal decision-making configurations listed in the corollary, no matter with or without feedback or feedforward, share the same optimal joint Gaussian result.\hfill \qedsymbol{}




\bibliographystyle{IEEEtran}
\bibliography{IEEEabrv,main}



 




\vfill

\end{document}